\PassOptionsToPackage{x11names}{xcolor}
\documentclass[acmsmall,screen,nonacm]{acmart} %
\usepackage[T1]{fontenc}
\usepackage{graphicx}

\usepackage{amssymb}
\usepackage{amsmath}
\usepackage{soul}
\usepackage{xfrac}
\usepackage{tabto}
\NumTabs{64}
\usepackage{wrapfig}

\usepackage{placeins} %

\usepackage{xspace}
\usepackage{stmaryrd}
\usepackage{mathtools}
\usepackage{csquotes}
\usepackage{listings}
\usepackage{tikz}
\usetikzlibrary{arrows.meta, positioning, shapes.multipart,calc,decorations.pathreplacing}
\usepackage{xcolor}
\usepackage{wrapfig}

\usepackage{relsize} %

\usepackage{graphicx}   %
\usepackage{subcaption} %

\usepackage{tikz}
\usetikzlibrary{arrows.meta, positioning,automata}

\usepackage{aliascnt}
\usepackage{thmtools,thm-restate}

\makeatletter
\providecommand{\theHexample}{\theHsection.\arabic{example}}
\providecommand{\theHtheorem}{\theHsection.\arabic{theorem}}
\providecommand{\theHlemma}{\theHsection.\arabic{lemma}}
\providecommand{\theHdefinition}{\theHsection.\arabic{definition}}
\makeatother

\usepackage[nameinlink,noabbrev,capitalise]{cleveref}

\usepackage[framemethod=TikZ]{mdframed}
\mdfdefinestyle{theoremstyle}{%
linecolor=black,linewidth=.75pt,%
innertopmargin=0pt,
innerbottommargin=4pt, %
skipabove=\topskip
skipbelow=\topskip
} 
\mdtheorem[style=theoremstyle]{problemstatement}{Problem}
\Crefname{problemstatement}{Problem}{Problems}
\crefname{problemstatement}{Problem}{Problems}

\mdtheorem[style=theoremstyle]{assumption}{Assumption}
\Crefname{assumption}{Assumption}{Assumptions}
\crefname{assumption}{Assumption}{Assumptions}

\mdtheorem[style=theoremstyle]{observation}{Observation}
\Crefname{observation}{Observation}{Observations}
\crefname{observation}{Observation}{Observations}

\theoremstyle{remark}

\input{macros}

\title{Multiobjective Preexpectation Reasoning for Probabilistic Programs}

\author{Lena Verscht}
\orcid{0000-0001-6823-7918}
\affiliation{
    \institution{Saarland University, Saarland Informatics Campus}
    \city{Saarbrücken}
    \country{Germany}
}
\affiliation{
    \institution{RWTH Aachen University}
    \city{Aachen}
    \country{Germany}
}
\email{lverscht@cs.uni-saarland.de}

\author{Hannah Mertens}
\orcid{0009-0009-6815-3285}
\affiliation{
    \institution{RWTH Aachen University}
    \city{Aachen}
    \country{Germany}
}

\author{Kevin Batz}
\orcid{0000-0001-8705-2564}
\affiliation{%
    \institution{Cornell University}
    \city{Ithaca, NY}
    \country{USA}
}

\author{Sebastian Junges}
\orcid{0000-0003-0978-8466}
\affiliation{
    \institution{Radboud University}
    \city{Nijmegen}
    \country{The Netherlands}
}

\author{Benjamin Lucien Kaminski}
\orcid{0000-0001-5185-2324}
\affiliation{
    \institution{Saarland University, Saarland Informatics Campus}
    \city{Saarbrücken}
    \country{Germany}
}
\affiliation{
    \institution{University College London}
    \city{London}
    \country{UK}
}

\author{Joost-Pieter Katoen}
\orcid{0000-0002-6143-1926}
\affiliation{%
    \institution{RWTH Aachen University}
    \city{Aachen}
    \country{Germany}
}

\begin{document}

\makeatletter
\def\theHexample{\theexample}
\def\theHdefinition{\thedefinition}
\def\theHlemma{\thelemma}
\def\theHtheorem{\thetheorem}
\makeatother

\begin{abstract}
Probabilistic programs with nondeterminism model planning problems in which a strategy typically resolves the nondeterminism to meet a specification.
We study the \emph{multiobjective} setting, meeting several specifications at once along a Pareto front, and provide a deductive, program-level account of strategy synthesis.
Its core is a \emph{multiobjective preexpectation transformer} mapping a tuple of postexpectations to the set of simultaneously achievable values, an element of the convex Hoare powerdomain.
It conservatively extends weakest preexpectations and lifts standard loop rules.
We develop rules to synthesize witnessing strategies as \emph{mixed determinizations} that solve nondeterminism probabilistically.
We prove the transformer and synthesis rules sound against an operational Markov decision process (MDP) semantics, without requiring a finite state space:
our approach can be seen as a symbolic approach---at program level---for multiobjective optimization over infinite MDPs.
We demonstrate our machinery using two case studies.

\end{abstract}

\maketitle

\section{Introduction}
\label{sec:intro}

Imperative probabilistic programs featuring both coin flips and nondeterministic choice are a powerful formalism for describing control and planning problems~\cite{morgan1996probabilistic,agentmodels,zilberstein2025demonic}.
A strategy (or policy or controller) prescribes how the nondeterminism is resolved.
In the single-objective setting, control and planning problems ask for a strategy that optimizes a \emph{single quantity}, e.g.\ the expected fuel level at the end of execution, or the probability that the variable \texttt{success} is true upon termination.
\emph{Programmatic strategy synthesis}~\cite{batz2024programmatic} resolves the nondeterminism directly at the level of the program: it refines the program into a \emph{determinization}, a program that symbolically represents the strategy. %

Many control and planning problems, however, require optimizing \emph{several quantities simultaneously}---for instance, the success probability and the remaining fuel together~\cite{roijer2017multi}.
Optimizing several quantities simultaneously generally inflicts trade-offs~\cite{white1982multi,chatterjee2006mdp}: improving one objective may come at the expense of degrading another.
In contrast to the single-objective setting, there is generally no single strategy that maximizes all objectives simultaneously.
Instead, several strategies achieve incomparable outcomes, giving rise to a \emph{Pareto front}: the set of objective-value combinations that are optimal in the sense that no objective can be improved without degrading another.
Points on or below the Pareto front are called \emph{almost achievable}.
Consequently, multiobjective strategy synthesis is not a straightforward extension of the single-objective setting:
optimizing each objective separately does neither tell us which combinations are simultaneously almost achievable, nor how to construct a strategy realizing a desired trade-off.
This raises two central questions: \emph{which trade-offs are almost achievable by a given program, and how can we synthesize a determinization realizing a prescribed one?}

To answer the first question, we develop a probabilistic predicate-transformer-style approach, building upon classical predicate transformers for probabilistic programs~\cite{kaminski2019advanced,mciver2005abstraction,morgan1996probabilistic,kozen85probabilistic}.
Given a quantity or \emph{postexpectation} $\post$ and a probabilistic program $\program$, the outcome $\wp{\program}{\post}$ of the probabilistic predicate transformer \wpsymbol is another quantity (often called a \emph{preexpectation}) that captures the maximal expected value that $\post$ has after execution of $\program$.
At the heart of our approach is a \emph{multiobjective preexpectation transformer} \mopsymbol, which conservatively extends weakest preexpectations.
Whereas the classical transformer maps a single postexpectation $\post$ to a preexpectation, ours maps a tuple $\post_1,\dots,\post_n$ to the set of points on or under the Pareto front. 
Formally, these sets live in the \emph{convex Hoare powerdomain}~\cite{abramsky1995domain}. %
The \mopsymbol transformer is well-defined, monotone, and $\omega$-continuous, so loop semantics can be defined by least fixpoints.
Determining these fixpoints directly, however, is generally difficult.
To enable deductive reasoning about programs with loops, we lift classical invariant-based proof techniques to the multiobjective setting.
In particular, we develop rules for establishing \emph{both} upper \emph{and} lower bounds on multiobjective preexpectations.

The \mopsymbol transformer tells us \emph{which} trade-offs are almost achievable, but it does not itself \emph{construct} a determinization realizing a prescribed one.
We therefore turn to the second question: %
given a probabilistic program with nondeterministic choices and a specification---postexpectations $\post_1,\dots,\post_n$ with desired lower bounds $g_1,\dots,g_n$---we seek a determinization %
that resolves the nondeterminism and meets the specification.
For a single objective, if there is an optimal strategy, then it suffices to consider deterministic strategies. %
With multiple objectives, however, (1)~instead of a single optimal determinization we must consider multiple Pareto optimal determinizations, and (2)~realizing a trade-off may require \emph{mixing} different determinizations, i.e.\ randomizing over several determinizations rather than committing to a single \emph{pure} (meaning non-randomized) determinization.
We therefore extend the notion of determinizations to \emph{mixed determinizations}, which encode a finite distribution over pure determinizations using initial coin flips.

To target a particular trade-off, we adopt a weighted-sum technique from probabilistic model checking \cite{forejt2012pareto}.
We combine the objectives into a single weighted one---assigning weight $\weight_i$ to objective $\post_i$---and synthesize a determinization for the resulting postexpectation $\sum_i \weight_i \post_i$.
This is an ordinary single-objective query, so we can fall back to existing programmatic strategy synthesis~\cite{batz2024programmatic}. %
By choosing suitable weight vectors and mixing the resulting determinizations, we obtain a constructive approach to multiobjective strategy synthesis and provide sufficient conditions under which this approach succeeds.

\paragraph{Contributions}
In summary, we contribute the following:
\begin{itemize}
  \item We introduce the \emph{multiobjective preexpectation transformer} on the convex Hoare powerdomain---a conservative extension of weakest preexpectations---and prove its correspondence with almost achievable values in the operational semantics (\Cref{sec:mop}). %
  \item Building on this transformer, we lift invariant-based loop reasoning to the multiobjective setting, for both upper and lower bounds (\Cref{sec:loops}).
  \item We develop a weighted-sum-based approach to \emph{multiobjective programmatic strategy synthesis}, providing sufficient conditions for synthesizing mixed determinizations realizing prescribed specifications (\Cref{sec:wsum}).
\end{itemize}
We demonstrate our approach %
on two case studies (\Cref{sec:examples}).
All theorems presented are proven in the appendix \emph{and} have been formally verified in Lean.
Our Lean formalization was co-developed by an LLM in a careful, iterative process that uncovered subtle edge cases. %

\section{Overview}
\label{sec:overview}

\begin{figure}[t]
    \centering

    \begin{subfigure}[t]{0.32\textwidth}
        \centering
\begin{lstlisting}[mathescape,basicstyle=\footnotesize\ttfamily]
$\ASSIGN{p}{0} \fatsemi \ASSIGN{f}{0} \fatsemi \ASSIGN{t}{g} \fatsemi$
$\WHILE{f=0 \land p<g} \{$
  $\ASSIGN{t}{t \monus 1} \fatsemi$
  $\GCFSTOPEN{~\boldsymbol{\blue{\trueBold}}}{}$
    $\gray{\text{\% fast but risky}}$
    $\gray{\text{\% (maybe chosen)}}$
    $\ASSIGN{p}{p+2} \fatsemi$ 
    $\PCHOICE{\ASSIGN{f}{1}}{0.3}{\ASSIGN{f}{0}}$
  $\GCSECCLOSE \GCSECOPEN{~\boldsymbol{\blue{\trueBold}}}{}$
    $\gray{\text{\% slow but safe}}$
    $\gray{\text{\% (maybe chosen)}}$
    $\ASSIGN{p}{p+1} \fatsemi$
    $\PCHOICE{\ASSIGN{f}{1}}{0.05}{\ASSIGN{f}{0}}$
  $\GCSECCLOSE$
$\}$
\end{lstlisting}%
        \caption{$\program_{\text{robot}}$ with variables $p,t,g$ ranging over $\Nats$ and $f$ ranging over $\mathbb{B}$}
        \label{fig:robot_program_nondet}
    \end{subfigure}
    \hfill
    \begin{subfigure}[t]{0.32\textwidth}
        \centering
\begin{lstlisting}[mathescape,basicstyle=\footnotesize\ttfamily]
$\ASSIGN{p}{0} \fatsemi \ASSIGN{f}{0} \fatsemi \ASSIGN{t}{g} \fatsemi$
$\WHILE{f=0 \land p<g} \{$
  $\ASSIGN{t}{t \monus 1} \fatsemi$
  $\GCFSTOPEN{~\boldsymbol{\red{g\monus p \leq 2}}}$ 
    $\gray{\text{\% fast but risky}}$
    $\gray{\text{\% (sometimes chosen)}}$
    $\ASSIGN{p}{p+2} \fatsemi$ 
    $\PCHOICE{\ASSIGN{f}{1}}{0.3}{\ASSIGN{f}{0}}$
  $\GCSECCLOSE \GCSECOPEN{~\boldsymbol{\red{g\monus p > 2}}}{}$
    $\gray{\text{\% slow but safe}}$
    $\gray{\text{\% (sometimes chosen)}}$
    $\ASSIGN{p}{p+1} \fatsemi$
    $\PCHOICE{\ASSIGN{f}{1}}{0.05}{\ASSIGN{f}{0}}$
  $\GCSECCLOSE$
$\}$
\end{lstlisting}
        \caption{$\program_{\text{robot}}^{\iverson{p \geq g}}$}
        \label{fig:robot_program_det}
    \end{subfigure}
    \hfill
\begin{subfigure}[t]{0.32\textwidth}
        \centering
\begin{lstlisting}[mathescape,basicstyle=\footnotesize\ttfamily]
$\ASSIGN{p}{0} \fatsemi \ASSIGN{f}{0} \fatsemi \ASSIGN{t}{g} \fatsemi$
$\WHILE{f=0 \land p<g} \{$
  $\ASSIGN{t}{t \monus 1} \fatsemi$
  $\GCFSTOPEN{~\boldsymbol{\red{\trueBold}}}$ 
    $\gray{\text{\% fast but risky}}$
    $\gray{\text{\% (always chosen)}}$
    $\ASSIGN{p}{p+2} \fatsemi$ 
    $\PCHOICE{\ASSIGN{f}{1}}{0.3}{\ASSIGN{f}{0}}$
  $\GCSECCLOSE \GCSECOPEN{~\boldsymbol{\red{\falseBold}}}{}$
    $\gray{\text{\% slow but safe}}$
    $\gray{\text{\% (never chosen)}}$
    $\ASSIGN{p}{p+1} \fatsemi$
    $\PCHOICE{\ASSIGN{f}{1}}{0.05}{\ASSIGN{f}{0}}$
  $\GCSECCLOSE$
$\}$
\end{lstlisting}
        \caption{$\program_{\text{robot}}^{t}$}
        \label{fig:robot_program_det2}
    \end{subfigure}
    
    \caption{Nondeterministic program $\program_{\text{robot}}$ as well as two determinizations maximizing $\iverson{p \geq g}$ and $t$, respectively.}
    \label{fig:robot_program}
    \Description{
    Three program listings are shown side by side.
    The left listing is the original nondeterministic robot program.
    Inside a loop, the robot chooses nondeterministically between a fast movement, which advances by two positions with a higher probability of failure, and a slow movement, which advances by one position with a lower probability of failure.

    The middle listing is a determinization maximizing the probability of reaching the goal.
    The choice has been replaced by a deterministic guard that selects the slow action until the robot is within two positions of the goal, after which it always selects the fast action.

    The right listing is a determinization maximizing the expected remaining time.
    The nondeterministic choice has been resolved to always execute the fast movement.
}
\end{figure}

\noindent%
We illustrate the paper's main contributions---multiobjective preexpectation reasoning and strategy synthesis---through a motivating example.
A robot shall travel along a one-dimensional path to a goal position $g$, subject to a risk of system failure along the way.
At each step, the robot spends one unit of time and chooses how to move forward: \emph{fast}, covering more ground but at a higher risk of failure; or \emph{slow}, which is less risky but also makes less progress.
The robot takes steps until it either reaches $g$ or fails.
Two quantities shall be optimized: the success probability of the robot reaching $g$ and the remaining time budget.

We model this scenario as the probabilistic program $\program_{\text{robot}}$ in \Cref{fig:robot_program_nondet}.
The variable $p$ holds the robot's position, $f$ is a flag indicating whether ($1$) or not ($0$) the robot experiences system failure, and $t$ tracks the remaining time budget, initialized with the initial distance to the goal $g$.
The \texttt{while} loop runs as long as no failure occurred ($f=0$) and the robot is still short of the goal ($p < g$).

In each loop iteration, the time budget is first decremented, using the monus operation $\monus$ so that~$t$ never drops below 0.
The robot chooses between two branches: 
The fast branch advances its position by $2$ but fails with probability 0.3.
The slow branch advances by~$1$ but fails with probability 0.05.
Both branches are guarded by $\true$, so \emph{either} branch may be taken in \emph{every} step: the program leaves this choice \emph{nondeterministic}, modeling a decision that a strategy has yet to resolve.
The task is now to resolve this nondeterminism so as to \emph{optimize both objectives simultaneously}---the probability of reaching the goal \emph{and} the expected remaining time budget $t$ upon termination.

\subsection{State of the Art: Maximizing Single Objectives}
Let us first consider a simpler question: 
given all resolutions of the nondeterminism, what is the maximal probability of terminating the loop by reaching the goal position, i.e.\ with $p \geq g$?
This single objective question is well-studied for probabilistic programs.

First, to compute this probability, we use the angelic weakest preexpectation of $\program_{\text{robot}}$ with respect to postexpectation $\iverson{p \geq g}$ (a Boolean predicate which evaluates to 1 iff $p \geq g$) \cite{morgan1996probabilistic}.
This preexpectation describes exactly this maximal probability, %
namely $\wp{\program_{\text{robot}}}{\iverson{p \geq g}} = 0.95^{g \monus 2}$, and can be computed mechanically, up to loops, by deductive verification \cite{caesar}.
The optimum is attained by taking the \emph{safe} branch up until the remaining distance is at most two, and then finishing with a fast step, since then the probability of reaching the goal is 1.
This strategy dominates the one taking always the safe branch, as there the probability of reaching the goal from a position two steps away is just 0.95.
The weakest preexpectation calculus yields this optimal \emph{value}, but no \emph{strategy} \emph{(resolving the nondeterminism)} to achieve it.

Second, to determine such a strategy that attains the optimum, we use  \emph{programmatic strategy synthesis}~\cite{batz2024programmatic}, which resolves the nondeterminism by \emph{refining} the program itself.
Rather than describing a strategy externally, the refinement strengthens the guards of the nondeterministic choices, turning the original program into a deterministic (but still probabilistic) one---a \emph{determinization}---that represents the strategy.
For $\program_{\text{robot}}$, this yields the determinization $\program_{\text{robot}}^{\iverson{p \geq g}}$ in \Cref{fig:robot_program_det}: the fast branch's guard becomes $p \geq g \monus 2$ and the safe branch's becomes $p < g \monus 2$, so the program deterministically takes the safe branch in all but the last iteration and attains the optimal reach-probability $0.95^{g \monus 2}$.

The same approach applies to the other objective:
the maximal expected remaining time budget is $\wp{\program_{\text{robot}}}{t} = g - \sfrac{10}{3} \cdot \bigl(1-0.7^{\left\lceil \sfrac{g}{2} \right\rceil}\bigr)$,
taking into account the probability of ending early.
This optimum is realized by the determinization $\program_{\text{robot}}^{t}$ (\Cref{fig:robot_program_det2}) that always takes the fast branch.

\begin{wrapfigure}[16]{r}{.4\textwidth}
\vspace*{-2\intextsep}
\centering
\scalebox{0.86}{
\begin{tikzpicture}[
    x={4.3cm/3.2}, %
    y={4.3cm/1.1},
    >=stealth
]

\coordinate (V1) at (1.290125,0.857375);
\coordinate (V2) at (2.635,0.665);
\coordinate (V3) at (2.81,0.49);

\filldraw[
    fill=mygreen!30,
    draw=mygreen,
    line join=round
]
    (0,0)
    -- (0,0.857375)
    -- (V1)
    -- (V2)
    -- (V3)
    -- (2.81,0)
    -- cycle;

\draw[->] (0,0) -- (3.2,0) node[below] {$t$};
\draw[->] (0,0) -- (0,1.1) node[left] {$\iverson{p\geq g}$};
\foreach \x in {1,2,3}
    \draw (\x,0.02) -- (\x,-0.02) node[below] {\x};

\foreach \y/\label in {{0.5},{1}}
    \draw (0.05,\y) -- (-0.05,\y) node[left] {\label};

\draw[draw=myblue,very thick] (V1)
    -- (V2)
    -- (V3);

\fill[mygreen!80!black] (V1) circle (2.2pt);
\node[mygreen!80!black] at ($(V1)+(-0,-0.1)$) {$\kappa_1$};

\fill[mygreen!80!black] (V2) circle (2.2pt);
\node[mygreen!80!black] at ($(V2)+(-0.08,-0.07)$) {$\kappa_2$};

\fill[mygreen!80!black] (V3) circle (2.2pt);
\node[mygreen!80!black] at ($(V3)+(-0.15,-0.03)$) {$\kappa_3$};

\coordinate (P) at ($(V1)!0.6!(V2)$);
\fill[myblue] (P) circle (2.2pt);
\node[myblue] at ($(P)+(0,-0.08)$) {$\target$};

\end{tikzpicture}}
\caption{The Pareto front (in blue) for the program $\program_{\text{robot}}$ with respect to maximizing $t$ and $\iverson{p \geq g}$ for the initial state $\pstate$ with $\pstate(g) = 5$, including a target $\target$ for strategy synthesis. %
Throughout, we write value pairs in the order $\vvechorizon{t}{\iverson{p \geq g}}$.
}
\label{fig:robot-pareto}
\Description{
    A two-dimensional plot shows the Pareto front for the robot program with respect to maximizing the remaining time $t$ and the probability of reaching the goal $[p \geq g]$.
    The horizontal axis represents $t$, while the vertical axis represents $[p \geq g]$.
    The achievable region is shaded light green and is bounded above by a blue Pareto front.
    The Pareto front consists of two descending line segments connecting three green points labeled $\kappa_1$, $\kappa_2$, and $\kappa_3$.
    The point $\kappa_1$ has the highest probability of reaching the goal but the smallest remaining time, while $\kappa_3$ has the largest remaining time but the lowest probability of reaching the goal.
    The point $\kappa_2$ represents an intermediate trade-off between the two objectives.
    A blue point labeled $\tau$ lies on the Pareto front between $\kappa_1$ and $\kappa_2$, indicating the target for strategy synthesis.
}
\end{wrapfigure}

\subsection{From Single to Multiple Objectives}

We now consider two objectives, each maximized by a \emph{different} determinization: $\program_{\text{robot}}^{\iverson{p \geq g}}$ takes the safe branch up until the last step, while $\program_{\text{robot}}^{t}$ takes the fast branch throughout.
These strategies conflict---the safe route maximizes the reach-probability but spends almost all of the time budget, whereas the fast route preserves time but is far more likely to fail.
No single strategy attains both remaining time $g - \frac{1-0.7^{\left\lceil g/2 \right\rceil}}{0.3}$ \emph{and} reach-probability $0.95^{g \monus 2}$ at once. %
This raises the question of which value pairs $(x,\, y)$ \emph{are} simultaneously achievable, and by which strategies.
We interpret objective vectors as lower bounds: a vector is achievable if some strategy attains \emph{at least} its value in every component.
Consequently, achievable sets are downward closed.
Answering this is well understood for finite-state MDPs (e.g.\ \cite{etessami2008multi,chatterjee2006mdp}), but no program-level calculus exists to reason about multiple expectations simultaneously.

The safe determinization achieves the corner $\kappa_1 = \bigl(g-(20-19 \cdot 0.95^{g \monus 2}),\, 0.95^{g \monus 2} \bigr)$ by taking the fast branch only once, and the fast determinization the corner $\kappa_3 = \bigl(g - \frac{1-0.7^{\left\lceil g/2 \right\rceil}}{0.3},\, 0.7^{\lceil g/2 \rceil-1} \bigr)$ by taking it thrice; %
neither dominates the other: $\kappa_1$ has the higher success probability, whereas $\kappa_3$ has the higher expected remaining time.
Such mutually incomparable optima are called \emph{Pareto optimal}, and together with the trade-offs obtained by mixing determinizations, they form the \emph{Pareto front}.
For the robot, a third Pareto optimal corner $\kappa_2$ is achieved by a determinization that takes the fast branch twice---in the first and last iteration.
\Cref{fig:robot-pareto} shows the Pareto front for $\program_{\text{robot}}$ with respect to objectives $t$ (x axis) and $\iverson{p \geq g}$ (y axis) for a fixed initial state with $g = 5$;
every point on or below it is simultaneously achievable, and the front itself is the boundary of what is possible.

In the robot example, every point on or below the Pareto front is realizable by a suitable determinization.
For general programs, however, determinizations may approach a Pareto optimal point arbitrarily closely without any determinization achieving it exactly.
We therefore consider the larger set of \emph{almost achievable} points, comprising all achievable points as well as their limiting points.
This phenomenon is due to the infinite state space; the operational semantics of \mopsymbol are defined in terms of \emph{infinite state} Markov decision processes (MDPs), implying in particular that our approach supports multiobjective reasoning and synthesis beyond the predominantly finite-state setting of existing approaches.
For finite-state models, every Pareto optimal point is \emph{exactly} achievable\footnote{This holds for reachability rewards, as considered here, but not necessarily for more general objectives such as mean payoff; see, e.g., \cite{main2026mixing}.}.
For the robot, the almost achievable points form the green region in \Cref{fig:robot-pareto}.

We set out to answer two questions.
The first concerns \emph{what} can be achieved:
\begin{quote}
\textbf{(Q1)}\ \emph{Given a program $\program$ and postexpectations $\post_1,\dots,\post_n$, what is the multiobjective preexpectation of $\program$---i.e.\ the set of simultaneously almost achievable value tuples---and can it be systematically determined by a mechanizable, program-level calculus?}
\end{quote}
Answering \textbf{(Q1)} yields a multiobjective preexpectation transformer, providing the Pareto front but no corresponding strategies.
As in the single-objective case, we treat strategies by refinement:
\begin{quote}
\textbf{(Q2)}\ \emph{Given a program $\program$, postexpectations $\post_1,\dots,\post_n$, and an achievable target, how do we determine a refinement of $\program$ that realizes it?}
\end{quote}

\subsection{The Pareto Front via the Multiobjective Preexpectation Transformer}

We answer \textbf{(Q1)} with a weakest preexpectation-style transformer for multiple objectives, denoted $\mopsymbol$, that operates at the level of the program $\program$.
Whereas the single-objective transformer maps a postexpectation to a preexpectation, $\mopsymbol$ maps a tuple of postexpectations $\post_1,\dots,\post_n$ to a \emph{multiobjective preexpectation}: the set of value tuples simultaneously almost achievable by refinements of $\program$.
Such a set is downward closed, convex, and Scott closed---an element of the \emph{convex Hoare powerdomain}. %
Each closure has a reading: downward closure because we seek lower bounds, so achieving a tuple means achieving every tuple below it; convexity because \emph{mixing} refinements achieves convex combinations of their outcomes; and Scott closure because it includes limiting points that can be approached arbitrarily closely but are not necessarily realizable by any refinement---precisely the points that are \emph{almost} achievable but not achievable.

 \begin{wrapfigure}[17]{r}{0.34\textwidth}
     \vspace{-1.3\intextsep}
\begin{lstlisting}[mathescape, basicstyle=\small]
  $\annotate{\dwc{\convset{\vvechorizon{1}{0},\vvechorizon{0}{1}}}}$
  $\GCFSTOPEN{~\true}$
      $\annotate{\dwcset{\vvechorizon{1}{0}}}$
      $\ASSIGN{x}{1} \fatsemi \ASSIGN{y}{0}$
      $\annotate{\dwcset{\vvechorizon{x}{y}}}$
  $\GCSECCLOSE \GCSECOPEN{\true}$
      $\annotate{\dwcset{\vvechorizon{0}{1}}}$
      $\ASSIGN{x}{0} \fatsemi \ASSIGN{y}{1}$
      $\annotate{\dwcset{\vvechorizon{x}{y}}}$
  $\GCSECCLOSE$
  $\annotate{\dwcset{\vvechorizon{x}{y}}}$
\end{lstlisting} 
\caption{Program annotations with respect to the multiobjective postexpectation $\dwcset{\vvechorizon{x}{y}}$. Read from bottom to top.}
\label{fig:mop-annotations}
\Description{
    A program listing shows annotations for a nondeterministic choice between two assignments, with respect to the multiobjective postexpectation $\operatorname{dwc}\{(x,y)\}$.
    The program contains two branches, each guarded by the condition $\texttt{true}$.
    The first branch assigns $x := 1$ and $y := 0$, while the second branch assigns $x := 0$ and $y := 1$.
    Orange annotations indicate the corresponding multiobjective expectations at different program locations.
    At the end of both branches and after the nondeterministic choice, the annotation is $\operatorname{dwc}\{(x,y)\}$.
    Before the assignments, the first branch is annotated with $\operatorname{dwc}\{(1,0)\}$, while the second branch is annotated with $\operatorname{dwc}\{(0,1)\}$.
    The annotation at the top of the program is $\operatorname{dwc}(\operatorname{conv}\{(1,0),(0,1)\})$, representing the convex hull of the two possible outcomes.
    The annotations are intended to be read from bottom to top.
}
\end{wrapfigure}

\Cref{fig:mop-annotations} illustrates the transformer on a minimal program that nondeterministically sets $(x,y)$ to either $(1,0)$ or $(0,1)$; the annotations, read best from bottom to top, show $\mopsymbol$ applied to the multiobjective expectation $\dwcset{\vvechorizon{x}{y}}$. %
Each branch is deterministic, so its preexpectation is a downward closed singleton: the first branch yields $\dwcset{\vvechorizon{1}{0}}$, the second $\dwcset{\vvechorizon{0}{1}}$.
At the nondeterministic choice, $\mopsymbol$ takes the \emph{convex hull} of the two, which simplifies to $\dwc{\convset{\vvechorizon{1}{0},\vvechorizon{0}{1}}}$.
This is exactly the set of tuples achievable by \emph{mixing} the two branches---the segment between $(1,0)$ and $(0,1)$---capturing that a coin flip between the branches realizes any convex combination of their outcomes, as well as everything below it.

For the robot program, \mopsymbol similarly characterizes the entire region on and below the Pareto front in \Cref{fig:robot-pareto}.
We will thus answer \textbf{(Q1)} as follows: the multiobjective preexpectation of $\program$ with respect to $\post_1,\dots,\post_n$ is $\mop{\program}{\dwcset{\vvvechorizon{\post_1}{\dots}{\post_n}}}$, whose maximal elements form the Pareto front.

\subsection{Optimal Determinizations for Multiple Expectations: Existence and Synthesis}
\label{sec:overview-optimal-determ}

We now turn to \textbf{(Q2)}: given an achievable target $\target$, calculate a refinement realizing it.
Suppose we aim for a trade-off strictly between the corners $\kappa_1$ and $\kappa_2$ of the robot's Pareto front, say $\target \approx \vvechorizon{2.09}{0.74}$, see \Cref{fig:robot-pareto}.
This is already different from the single-objective setting: we are not looking for a determinization achieving the one unambiguous optimal value, but instead pick one Pareto optimal value that we aim to achieve.
More fundamentally, a pure (i.e., non-randomized) determinization may no longer suffice at all.

Reconsider the two-branch program of \Cref{fig:mop-annotations}: for any fixed initial state, it admits exactly two deterministic refinements, realizing $\vvechorizon{1}{0}$ and $\vvechorizon{0}{1}$.
Its Pareto front is the segment between them, so every trade-off other than these two endpoints---say $\vvechorizon{\sfrac{1}{2}}{\sfrac{1}{2}}$---is achievable \emph{only} by flipping a coin between the two refinements.
The robot exhibits the same phenomenon on a larger scale: its deterministic refinements realize finitely many points, the Pareto optimal ones being $\kappa_1$, $\kappa_2$, and $\kappa_3$, and our target $\target$ lies strictly between $\kappa_1$ and $\kappa_2$.
Reaching it requires \emph{mixing}.
The refinement our method produces is%
\[
    \PCHOICE{\program_{\text{robot}}^{\iverson{p \geq g}}}{0.4}{\program_{\text{robot}}^{\kappa_2}},
\]
an initial coin flip choosing, with probability $0.4$, the safe determinization achieving $\kappa_1$, and otherwise the determinization $\program_{\text{robot}}^{\kappa_2}$ achieving $\kappa_2$.

The remaining question is how to synthesize the pure determinizations achieving $\kappa_1$ and $\kappa_2$.
Our approach is to combine the objectives into a single \emph{weighted sum}, assigning a weight to each objective.
By optimizing this weighted objective, we can use existing single-objective programmatic strategy synthesis to construct a pure determinization. %
Geometrically, choosing different weights selects different points on the Pareto front.
In particular, if a weight vector uniquely identifies a point on the front, any determinization optimal for the corresponding weighted objective must achieve that point.
For the robot, suitable weights identify the corners $\kappa_1$ and $\kappa_2$, allowing us to synthesize the two pure determinizations above.
Mixing them with the appropriate probabilities then realizes our prescribed target $\target$.

This construction generalizes beyond the robot example.
Whenever a target is expressible as a convex combination of points uniquely identified by suitable weight vectors, and optimal pure determinizations exist for the corresponding weighted objectives, we can synthesize a mixed determinization achieving the target.
This provides a constructive answer to \textbf{(Q2)}. %
However, these conditions do not always hold: optimal single-objective determinizations may fail to exist, and not every Pareto optimal extreme point can be uniquely identified by a weighted objective.
Nevertheless, our approach constitutes a first step toward general multiobjective programmatic strategy synthesis and applies successfully to all our case studies (\Cref{sec:examples}).

\section{Probabilistic Programs}
\label{sec:pp}

Let $\Vars = \{x,\, y,\, z,\, {\ldots}\}$ be a countably infinite set of \emph{(program) variables} that can take values in~$\VarsDom$.
The countably infinite set of \emph{(program) states} is given by $\States \eeq \set{ \pstate \colon \Vars \to \VarsDom} $, where we assume that $\pstate(x)=0$ for all but finitely many $x\in\Vars$.
We consider programs in a simple \emph{nondeterministic probabilistic guarded command language} ($\pGCL$):
\par\medskip
\noindent
\begin{tabular}{r l l @{\qquad} l l}
    $\program \coloneqq$
    & $\SKIP$ & \textit{\gray{effectless program}}
    & $\mid \GC{\guard_1}{\program}{\guard_2}{\program}$ & \textit{\gray{guarded choice}} \\
    & $\mid \ASSIGN{x}{\aexpr}$ & \textit{\gray{variable assignment}}
    & $\mid \PCHOICE{\program}{\pexpr}{\program}$ & \textit{\gray{probabilistic choice}} \\
    & $\mid \COMPOSE{\program}{\program}$ & \textit{\gray{sequential composition}}
    & $\mid \WHILEDO{\guard}{\program}$ & \textit{\gray{loop}}
\end{tabular}
\par\medskip
\noindent
where $\aexpr\colon \States \to \VarsDom$ is an \emph{expression};
$\guard_1,\,  \guard_2$, and $\guard$ are \emph{predicates} of type $\States \to \{\true,\, \false \}$; %
and $\pexpr\colon \States \to [0,1]$ is a \emph{probability expression}.
For the guarded choice, we require that $\guard_1 \vee \guard_2 = \true$, i.e.\ at least one guard is enabled.
We sometimes write $\pstate \models \guard$ instead of $\pstate \in \guard$.

Most \pGCL constructs are standard.
We comment on the two least common ones:
the guarded choice $\GC{\guard_1}{\program_1}{\guard_2}{\program_2}$ first checks which of the guards $\guard_1$ and $\guard_2$ hold in the current state.
If only one guard, say $\guard_1$, holds, then the guarded choice \emph{deterministically} executes $\program_1$.
If \emph{both} guards hold, then the guarded choice \emph{nondeterministically} executes \emph{either} $\program_1$ \emph{or} $\program_2$.
The probabilistic choice $\PCHOICE{\program_1}{\pexpr}{\program_2}$ introduces randomization: in state $\pstate$, the program $\program_1$ is executed with probability $\pexpr(\pstate)$ and $\program_2$ is executed with the remaining probability $1-\pexpr(\pstate)$.

We denote the set of all finitely supported probability distributions over a set $\Omega$ as $\distr{\Omega}$.
For a finite family of programs $\mathbb{C} = \set{\program_i}_{i \in \set{1,\dots,m}}$ and $\determDist \in \distr{\mathbb{C}}$, we use syntactic sugar
$\PCHOICEFINITE{\program_i}{}{\determDist(\program_i)}{}$ to denote the program which executes~$\program_i$ with probability $\determDist(\program_i)$.

\subsection{Program Determinizations}
\label{sec:determ}

A program $\program \in \pGCL$ is \emph{deterministic} if for all $\GC{\guard_1}{\program_1'}{\guard_2}{\program_2'}$ occurring in $\program$, $\guard_1 \xor \guard_2$ holds, where $\xor$ denotes the exclusive or; %
and otherwise \emph{nondeterministic}.
In the spirit of \cite{batz2024programmatic}, we introduce a notion of \emph{determinizations} for \pGCL programs.%

\begin{figure}[t]
\small
\centering
\begin{gather*}
	\frac{\program \text{ deterministic}}{\program \ddeterm \program}
	\qquad\qquad
	\frac{\program_1' \ddeterm \program_1 \qquad \program_2' \ddeterm \program_2}{\COMPOSE{\program_1'}{\program_2'} \hqdeterm \COMPOSE{\program_1}{\program_2}} 
	\\[.75em]
	\frac{\guard_1' \eentails \guard_1 \qquad \guard_2' \eentails \guard_2 \qquad {\entails}~ \guard_1' \xor \guard_2' \qquad %
		\program_1' \ddeterm \program_1  \qquad \program_2' \ddeterm \program_2}{\GC{\guard_1'}{\program_1'}{\guard_2'}{\program_2'} \qdeterm \GC{\guard_1}{\program_1}{\guard_2}{\program_2}} 
    \\[.75em]
	\frac{\program_1' \determ \program_1 \qquad \program_2' \determ \program_2}{\PCHOICE{\program_1'}{p}{\program_2'} \hqdeterm \PCHOICE{\program_1}{p}{\program_2}}
	\qquad\qquad
	\frac{\program' \determ \program}{\WHILEDO{\guard}{\program'} \!\hqdeterm \WHILEDO{\guard}{\program}}
\end{gather*}
\caption{Rules defining the determinization relation $\determ$. Here $\entails$ denotes \emph{entailment} between predicates, i.e.\ $\guard \entails \guard'$ if for all states $\pstate$ it holds that $\pstate \models \guard$ implies $\pstate \models \guard'$. %
}
\label{fig:refinementrules}
\Description{
Inference rules defining the program determinization relation.
Rules are given for deterministic programs, sequential composition, guarded nondeterministic choice, probabilistic choice, and while loops.
The guarded-choice rule replaces overlapping guards by mutually exclusive, stronger guards while preserving the enabled behaviour.
}
\end{figure}%
\begin{definition}[Determinizations of Programs]
\label{def:determ}
	The \emph{determinization} relation ${\determ} \subseteq \pGCL \times \pGCL$ is the smallest relation on $\pGCL$ satisfying the rules given in \Cref{fig:refinementrules}.
	If $\program' \determ \program$, then $\program'$ is a \emph{pure} determinization of $\program$.
	The set of pure determinizations is denoted by $\DetermSimple{\program}$.

	Let $\determDist \in \distr{\DetermSimple{\program}}$ be a finitely supported distribution over pure determinizations.
	The program 
	$
		\program' = \PCHOICEFINITE{\program_i}{}{\determDist(\program_i)}{}
	$
	executing $\program_i$ with probability $\determDist(\program_i)$ is a \emph{mixed} determinization of $\program$, denoted by $\program' \determmixed \program$.
	The set of mixed determinizations is denoted by $\DetermMixed{\program}$. %
	\myendmark
\end{definition}
\noindent
If $\program' \determ \program$, then $\program'$ and $\program$ coincide syntactically up to the guards occurring in the guarded choices.
The guards in $\program'$ are necessarily stronger than those in $\program$ as we require them to resolve the nondeterministic choices from $\program$.
This is formalized by the premises $\guard_1' \entails \guard_1$ and $\guard_2' \entails \guard_2$ in the $\determ$-rule for guarded choices.
The premise $\entails \guard_1' \xor \guard_2'$ ensures (1) that $\program'$ does not eliminate \emph{all} choices from some guarded choice in $\program$, i.e.\ $\entails \guard_1' \vee \guard_2'$,
and (2) that $\program'$ is deterministic, i.e.\ $\entails \neg \guard_1' \vee \neg \guard_2'$.
The latter differs from \cite{batz2024programmatic}, where $\program'$ resolves nondeterminism permissively, i.e.\ only as far as necessary. %
We restrict ourselves to fully deterministic refinements for simplicity; our approach could also be extended to permissive determinizations.

Mixed determinizations are the analogue to mixed memoryless schedulers in MDPs, being a finitely supported distribution over pure memoryless determinizations.
This randomization is required for Pareto optimal determinizations, as we have seen in \Cref{sec:overview}.
For MDPs as well as programs, it suffices to mix at the top level, nested mixing yields no more points.

\begin{example}[Determinizations]
	Consider the program $\GC{\boldsymbol{\blue{\trueBold}}}{\ASSIGN{x}{0}}{\boldsymbol{\blue{\trueBold}}}{\ASSIGN{x}{1}}$ which nondeterministically chooses between setting $x$ to 0 or 1.
	Two different pure determinizations of this program are, for example, $\program_5 = \GC{\boldsymbol{\red{x \geq 5}}}{\ASSIGN{x}{0}}{\boldsymbol{\red{x < 5}}}{\ASSIGN{x}{1}}$,
    and $\program_9 = \GC{\boldsymbol{\red{x \geq 9}}}{\ASSIGN{x}{0}}{\boldsymbol{\red{x < 9}}}{\ASSIGN{x}{1}}$.
    A mixed determinization then is $\PCHOICE{\program_5}{\boldsymbol{\red{0.3}}}{\program_9}$, choosing the above pure determinizations with probability $0.3$ and $0.7$, respectively.
\myendmark
\end{example}

\section{Reasoning about Multiple Expectations}
\label{sec:problem-mop}

Having formalized determinizations as programmatic strategies, we now turn to characterizing the objective values they can achieve.
One of the most widely studied approaches to reasoning about the behavior of programs is predicate-transformer semantics, originally proposed by \citet{dijkstra1976discipline} for non-probabilistic programs. %
For nondeterministic probabilistic programs, the standard framework is the \emph{weakest preexpectation calculus} of \citet{mciver2005abstraction}.
In this setting, an \emph{expectation} is a function $\post \colon \States \to \PosRealsInf$, where $\PosRealsInf = \PosReals \cup \{\infty\}$, that assigns a non-negative extended real value to each program state.
For $a \in \PosRealsInf$, we denote the constant expectation $\mylambda{\pstate} a$ by $a$.
By convention, we let $0 \cdot \infty = 0$ and $a \cdot \infty = \infty$ for $a > 0$.
Operations on expectations are defined pointwise as usual.
In particular, suprema are given as pointwise liftings of supremum on $\PosRealsInf$, i.e.\ for $E \subseteq \Exp$, we have
$\bigsqcup E \eeq \mylambda \pstate \sup \{\post(\pstate) \mid \post \in E \}$.

Given a program $\program$ and some \emph{postexpectation} $\post \in \Exp$, the \emph{weakest preexpectation} $\wp{\program}{\post}$ characterizes the \emph{supremum}\footnote{This is a sup since we consider the (less common) \emph{angelic} variant of \wpsymbol.
The dual \emph{demonic} variant computes their inf.}
over all possible expected values of $\post$ after executing $\program$, see \Cref{app:wp}:

\begin{restatable}{lemma}{wpsupdeterm}
\label[lemma]{theo:wp-sup-determ}
    For $\post \in \Exp$, %
	$
		\wp{\program}{\post} = \bigsqcup\ \setcomp{\wp{\program'}{\post} }{ \program' \determ \program}.
	$
\end{restatable}%
\noindent%
To the best of our knowledge, \autoref{theo:wp-sup-determ} is novel, as determinizations correspond to memoryless schedulers, whereas previous work established the result only for schedulers with memory.

Let $\programrun = \GC{\true}{\ASSIGN{x}{1} \fatsemi \ASSIGN{y}{0}}{\true}{\ASSIGN{x}{0} \fatsemi \ASSIGN{y}{1}}$.
Then $\wp{\programrun}{x} = 1 \sqcup 0 = 1$, since the maximal value of $x$ is 1, which is achieved by always choosing the left branch.

The \wpsymbol calculus thus can establish the optimal value of each single objective in isolation, but it does not capture trade-offs between multiple competing objectives.
To reason about multiple expectations, we first require some notation.

\paragraph{Notation for vectors and various closures}
For a vector $u \in \PosRealsInfVect$, we write $u_i$ for its $i$-th component and allow addition and multiplication with a scalar $r \in \PosReals$ by setting $(u+r)_i = u_i + r$ and $(u \cdot r)_i = u_i \cdot r$. 
For two vectors $u,v \in \PosRealsInfVect$, we denote by $u \lneq v$ that $v$ \emph{dominates} $u$, meaning that $u \leq v$ component-wise and that additionally $u_i < v_i$ in at least one component $i$.

Let $X \subseteq \PosRealsInfVect$.
The \emph{downward closure} of $X$ is given by $\dwc{X} = \bigl\{y \in \PosRealsInfVect ~\big|~ \exists x \in X \colon y \leq x \bigr\}$.
We say that $X$ is downward closed if $\dwc{X} = X$.
The \emph{convex hull} of $X$ is given by
\[
    \conv{X} \eeq \setcomp{ \sum_{i=0}^{k} \weightc_i \cdot x_i }{k \in \Nats, \quad \weightc_i \in [0,1], \quad \sum_{i=0}^{k} \weightc_i = 1,\quad x_i \in X}.
\]
We say that $X$ is convex if $\conv{X} = X$. %
Following \cite[Def. 2.3.1.]{abramsky1995domain}, the \emph{Scott closure} $\cl{X}$ of $X$ is given by the smallest superset of $X$ such that it is downward closed and closed under suprema of directed subsets.
Intuitively, $\cl{X}$ contains all points in $\dwc{X}$ plus those at $X$'s boundary.
We say that $X$ is Scott closed if $X = \cl{X}$.

To ease notation when computing weakest preexpectations for a vector of $n$ postexpectations $\post = \vvvechorizon{\post_1}{\dots}{\post_n} \in \Exp^n$, we denote its componentwise application by
\[
    \wpVec{\program}{\post} = \vvvechorizon{\wp{\program}{\post_1}}{\dots}{\wp{\program}{\post_n}}.
\]

\subsection{Achievable and Pareto Optimal Points}

Reconsider $\programrun = \GC{\true}{\ASSIGN{x}{1} \fatsemi \ASSIGN{y}{0}}{\true}{\ASSIGN{x}{0} \fatsemi \ASSIGN{y}{1}}$.
Say we are interested in the maximal values for $x$ and $y$. %
Applying \wpsymbol componentwise to $\vvechorizon{x}{y}$ gives $\wpVec{\programrun}{\vvechorizon{x}{y}} = \vvechorizon{1}{1}$.
While this correctly tells us that the maximal value is $1$ for each $x$ and $y$, there is no determinization of $\programrun$ for which \wpsymbol of $x$ and $y$ are \emph{both} 1:
maximizing $x$ requires choosing the left branch, whereas maximizing $y$ requires choosing the right.
We say that $\vvechorizon{1}{1}$ is \emph{not achievable}.%
\begin{definition}[Achievable Points]
\label{def:achievable-exp}
	A point $p \in \PosRealsInfVect$ is \emph{achievable} from $\pstate \in \States$ with respect to $\program \in \pGCL$ and $\post \in \Exp^n$
	iff there exists a mixed determinization $\program' \determmixed \program$ 
	such that $p \leq \wpVec{\program'}{\post}(\pstate)$. 
	The set of achievable points is given by $\AchExp{\program}{\post}{\pstate} = \dwcsetcomp{\wpVec{\program'}{\post}(\pstate) }{ \program' \determmixed \program}$.
	\myendmark
\end{definition}%

\begin{wrapfigure}[14]{r}{0.3\textwidth}
	\centering
   \scalebox{0.65}{
\begin{tikzpicture}[
    x={5cm/3.2}, %
    y={5cm/1.1},
    >=stealth
]

\coordinate (V1) at (1.290125,0.857375);
\coordinate (V2) at (2.635,0.665);
\coordinate (V3) at (2.81,0.49);

\filldraw[
    fill=mygreen!30,
    draw=mygreen,
    line join=round
]
    (0,0)
    -- (0,0.857375)
    -- (V1)
    -- (V2)
    -- (V3)
    -- (2.81,0)
    -- cycle;

\draw[->] (0,0) -- (3.2,0); %
\draw[->] (0,0) -- (0,1.1); %

\draw[draw=myblue,ultra thick] (V1)
    -- (V2)
    -- (V3);
\end{tikzpicture}}
\caption{The Pareto front (blue) lying in between the achievable (green) and the unachievable (white) points. %
}
\label{fig:pareto-front}
\Description{
A two-dimensional illustration of achievable and unachievable objective values.
The achievable region is shaded green, the unachievable region is white, and a blue curve separates them.
This boundary is labelled as the Pareto front and consists of all nondominated achievable points.
}
\end{wrapfigure}

\noindent%
    For $\programrun$ above, $\vvechorizon{1}{0}$ and $\vvechorizon{0}{1}$ are achievable by pure determinizations choosing the left respectively the right branch.
    A mixed determinization can, for example, choose each of these determinizations with probability $\sfrac{1}{2}$, yielding the point $\vvechorizon{\sfrac{1}{2}}{\sfrac{1}{2}}$.
    The set of achievable points for this program from any initial state is $\dwc{\convset{\vvechorizon{0}{1},\vvechorizon{1}{0}}}$.
	
In trying to determine which of these achievable points is best, we see that this has no clear answer:
$\vvechorizon{0}{1}$ exceeds $\vvechorizon{1}{0}$ in the second component but not the first, and vice versa for the first component.
In this case, we cannot (or rather do not want to) decide which is optimal.
Instead, we say that both are \emph{Pareto optimal} or \emph{on the Pareto front}. %

The Pareto front for each state $\pstate$ is precisely the set of points %
on the boundary between achievable and unachievable outcomes, and not strictly worse than any other achievable point, see \Cref{fig:pareto-front}.
In particular, Pareto optimal points need not be achievable, but only \emph{almost} achievable, meaning that they are contained in $\cl{\AchExp{\program}{\post}{\pstate}}$ but not necessarily in $\AchExp{\program}{\post}{\pstate}$.
This phenomenon is specific to infinite state spaces.
For reachability-reward objectives in finite-state MDPs, the set of achievable points is a polyhedron and hence already Scott closed (e.g.\ \cite{etessami2008multi}).
Thus, the natural object of interest is not the set of achievable points itself, but its closure.
This motivates our first problem:

\begin{problemstatement}%
\label{prob:mop}%
Given a pGCL program $\program$, a vector of postexpectations $\post \in \Exp^n$ and an initial state $\pstate$,
	what are the almost achievable points $\cl{\AchExp{\program}{\post}{\pstate}}$?
\end{problemstatement}

\section{The Multiobjective Preexpectation Transformer}
\label{sec:mop}
\setul{1pt}{.4pt}

In this section, we introduce the main technical tool for solving \Cref{prob:mop}: 
the \mopsymbol\ transformer for calculating \emph{\ul{{m}}ulti\ul{{o}}bjective \ul{{p}}reexpectations}.
It provides a program-level approach for reasoning about Pareto optimality of multiple expectations simultaneously.
Omitted proofs are in \Cref{proofs:mop}.%

\subsection{A Domain for Reasoning about Multiple Expectations}
\label{sec:mexp-domain}

To reason about preexpectations induced by program determinizations, we require a domain that simultaneously represents all almost achievable points from a given initial state.
As discussed in \Cref{sec:overview}, these sets must reflect both our interpretation of objective values as lower bounds and the possibility of mixing determinizations.
In particular, we require the following properties:
\begin{enumerate}
    \item \emph{Non-emptiness:} $\zeroHoare$ is always achievable, since we are interested in lower bounds. %
    \item \emph{Downward closedness:} If $p$ is achievable and $p' \leq p$, then $p'$ is also achievable by \Cref{def:achievable-exp}.
    \item \emph{Convexity:} Mixing determinizations enables us to achieve any convex combination of achievable points.
    \item \emph{Closedness under limit points:} As in the single-objective setting, optimal values need not be attained by any determinization, but may only be approached arbitrarily closely.
    We therefore include the limits of increasingly precise lower bounds, capturing almost-achievable points that are not necessarily achievable.
\end{enumerate}%

We build on established powerdomain constructions for combining nondeterministic and probabilistic computation. %
Classical Hoare, Smyth, and Plotkin powerdomains capture different forms of nondeterministic choice~\cite{plotkin79dijkstras,smyth83power}, while probabilistic and mixed powerdomains extend these ideas to probabilistic programs~\cite{jones90probabilistic,varacca2002powerdomain,tix2000convex}.
In particular, convex mixed powerdomains such as Kegelspitzen relate state and predicate transformer semantics for probabilistic programs with nondeterminism~\cite{keimel2009predicate,keimel2016mixed}.
Unlike these scalar-valued approaches, our setting is vector-valued.
We therefore use a convex sub-powerdomain of the Hoare powerdomain (see \cite[Theorem 6.2.13]{abramsky1995domain}).%

\begin{definition}[Hoare Powerdomain]
\label{def:hoare-power}
    The convex Hoare powerdomain of $\PosRealsInfVect$, denoted $(\HoarePowerDom, \subseteq)$, is the complete lattice whose carrier $\HoarePowerDom$ is the set of all subsets of $\PosRealsInfVect$ which are
    (1) non-empty, (2)~downward closed, (3) convex, and (4) Scott closed,
    ordered by subset inclusion.
    \myendmark
\end{definition}%
\noindent%
The four defining conditions of the Hoare powerdomain correspond precisely to the four semantic requirements for capturing almost-achievable points identified above.
In particular, Scott closedness ensures the inclusion of limit points.
In our finite-dimensional setting, Scott closure of downward-closed sets coincides with ordinary topological closure and closure under limits of $\omega$-chains. 
Although Scott closedness already implies downward closedness, we list both properties separately to emphasize their distinct semantic roles.

The least element of $\HoarePowerDom$ is $\zeroHoare=\set{0^n}$.
Suprema and infima of $X \subseteq \HoarePowerDom$ are constructed as:
\[
    \bighoaresup X \eeq \bigcap \setcomp{\HDel \in \HoarePowerDom }{ \bigcup X \subseteq \HDel}
    \qand
    \bighoareinf X \eeq \bigcap X.
\]
A much easier-to-handle way of expressing suprema is given as follows:
\begin{restatable}{theorem}{supiscl}
\label{theo:sup-is-cl}
    For nonempty $X \subseteq \HoarePowerDom$, \bighoaresup X = \cl{\conv{ \bigcup X }}.
\end{restatable}

\noindent%
Let $X,Y \in \HoarePowerDom$ and $r \in \PosReals$.
We define scalar multiplication on $\HoarePowerDom$ pointwise as $r \cdot X = \set{r \cdot x \mid x \in X}$.
Further, we define addition as the Minkowski sum $\minkow$ with
\(
    X \minkow Y \eeq \setcomp{x+y}{x \in X,\, y \in Y}.
\)
The Minkowski sum is well-defined on $\HoarePowerDom$ %
and satisfies standard properties of addition.
As usual, we let multiplication bind stronger than Minkowski sums.

\paragraph{Multiobjective expectations}

For reasoning about multiple objectives, we define a \emph{multiobjective expectation} %
as a mapping $\me \colon \States \to \HoarePowerDom$ from program states to elements of the Hoare powerdomain, instead of non-negative reals in the single-objective case.%
\begin{definition}[Multiobjective Expectations]
\label{def:mexp}
    The complete lattice of \emph{multiobjective expectations} of order $n$ is $(\mExp,\, \sqsubseteq)$, where
    \begin{itemize}
        \item $\mExp = \States \to \HoarePowerDom$ is the set of multiobjective expectations, mapping states to elements of the convex Hoare powerdomain of $\PosRealsInfVect$, and
        \item $\sqsubseteq$ is the pointwise lifted subset inclusion, i.e.\ 
        $
            \me \morespace{\mesubseteq} \me' \iff \forall\, \pstate \in \States \colon \me(\pstate) \morespace{\subseteq} \me'(\pstate).
        $\myendmark
    \end{itemize}
\end{definition}%
\noindent%
The least element is $\zeroME = \mylambda \pstate \{0^n\}$. %
Operations on the Hoare powerdomain and expectations, in particular closure operations, are lifted pointwise to multiobjective expectations. %
For instance, the supremum $\mecup$ and infimum $\mecap$ of multiobjective expectations are pointwise liftings from the Hoare powerdomain, so for $\meset \subseteq \mExp$,
\[
    \bigsqcup\, \meset \eeq \mylambda{\pstate} \bigcap \setcomp{\HDel \in \HoarePowerDom }{ \bigcup_{\me \in \meset} \me(\pstate) \subseteq \HDel}
    \qand 
    \bigsqcap\, \meset \eeq \mylambda{\pstate} \bigcap_{\me \in \meset} \me(\pstate).
\]

\begin{example}[Multiobjective Expectations]
    Suppose we are interested in the expected values of program variables $x$ and $y$.
    This (multi)objective is represented by the multiobjective expectation
    $\mylambda \pstate \dwcset{\vvechorizon{\pstate(x)}{\pstate(y)}}$,
    mapping each state $\pstate$ to the downward closure of the vector whose first component is the value of $x$ in $\pstate$ and whose second component is the value of $y$ in $\pstate$.
    For all $\pstate$, this set is nonempty.
    Further, it is convex: the convex combination of vectors below $\vvechorizon{\pstate(x)}{\pstate(y)}$ remains below $\vvechorizon{\pstate(x)}{\pstate(y)}$.
    It is also Scott closed, as any directed set below $\vvechorizon{\pstate(x)}{\pstate(y)}$ is bounded by $\vvechorizon{\pstate(x)}{\pstate(y)}$, so its supremum again lies below $\vvechorizon{\pstate(x)}{\pstate(y)}$.\myendmark
\end{example}
\noindent%
Indeed, for any expectations $\post_1,\dots,\post_n$, we have that $\mylambda \pstate \dwcset{\vvvechorizon{\post_1(\pstate)}{\dots}{\post_n(\pstate)}} \in \mExp$.
To ease notation, we write $\dwcset{\vvvechorizon{\post_1}{\dots}{\post_n}}$ for $\mylambda \pstate \dwcset{\vvvechorizon{\post_1(\pstate)}{\dots}{\post_n(\pstate)}}$.

\subsection{Inductive Computation of Multiobjective Preexpectations}
\label{sec:mop-rules}

A multiobjective \emph{post}expectation $\dwcset{\vvvechorizon{\post_1}{\dots}{\post_n}}$ specifies the quantities whose values \emph{after} program termination are of interest.
Starting from such a postexpectation, our goal is to compute the corresponding multiobjective \emph{pre}expectation, which characterizes all points that are almost achievable.
To this end, we introduce the following transformer on multiobjective expectations.

\begin{definition}[Multiobjective Preexpectation Transformer]
\label{def:mop}
    For a \pGCL program $C$ and a multi{\-}objective postexpectation $\me \in \mExp$, we define the \emph{multiobjective preexpectation transformer} %
    \[
        \mopC{C}\colon \mExp \to \mExp
    \]
    according to the rules in \Cref{table:mop}.
    \myendmark
\end{definition}

\begin{table}[t]
\renewcommand{\arraystretch}{1.5}
\begin{tabular}{@{\hspace{.5em}}l@{\hspace{2em}}l@{\hspace{2em}}l@{\hspace{.5em}}}
		\hline\hline
		$\boldsymbol{\program}$			& $\boldsymbol{\mop{\program}{\mpost}}$ \\
		\hline
		$\SKIP$				& $\mpost$ \\
        $\ASSIGN{x}{\ee}$ & $\mpost\subst{x}{\ee}$ \\
        $\COMPOSE{\program_1}{\program_2}$ & $\mop{\program_1}{\mop{\program_2}{\mpost}}$ \\
        $\GC{\guard_1}{\program_1}{\guard_2}{\program_2}$ & $[\guard_1] \cdot \mop{\program_1}{\mpost} \mmecup [\guard_2] \cdot \mop{\program_2}{\mpost}$\\
        $\PCHOICE{\program_1}{\ps}{\program_2}$ & $\ps \cdot \mop{\program_1}{\mpost} \mminkow (1-\ps) \cdot \mop{\program_2}{\mpost}$ \\
        $\WHILEDO{\guard}{\program'}$ & $\lfp \mylambda X\quad [\neg \guard] \cdot \mpost \mminkow [\guard] \cdot \mop{\program'}{X}$ \\
	\bottomrule\bottomrule
    \end{tabular}%
	\vspace{1em}
\caption{Rules for the $\mopsymbol$ transformer.
        $\lfp \mylambda X \Phi(X)$ denotes the least fixpoint of $\Phi$.}%
\label{table:mop}
\end{table}

\noindent%
The rules for \mopsymbol closely resemble those of the standard \wpsymbol\ transformer; see \Cref{tab:wp} for reference.
Indeed, the rules are obtained by lifting the \wpsymbol\ rules to the domain of multiobjective expectations.
We should emphasize that such resemblance was only achieved through a careful choice of the domain of multiobjective expectations in accordance with the semantic goals.

Let us go over the rules in \Cref{table:mop}:
$\SKIP$ leaves the multiobjective postexpectation unchanged.
For the assignment $\ASSIGN{x}{\ee}$, as usual, we evaluate $\mpost$ in the state where $x$ is substituted by $\ee$.
The \mopsymbol transformer is a compositional backward transformer, so we can push multiobjective expectations from end to beginning through the program.
Consequently, \mopsymbol of the composition $\COMPOSE{\program_1}{\program_2}$ is obtained by computing $\mop{\program_2}{\mpost}$ first and then passing the result into $\mop{\program_1}{\:\cdot\:}$.

For guarded choice, the rule has the same form as in \wpsymbol: it takes the supremum of the expectations obtained from the individual branches.
The key difference is encoded in the underlying domains, and hence in the meaning of this supremum.
For \wpsymbol, the supremum is the pointwise maximum of two scalar expectations, whereas for \mopsymbol, it is the pointwise Scott closed convex hull of the union of the two sets.
Thus, the familiar form of the rule is preserved by our choice of domain, even though \wpsymbol computes a single optimal expectation while \mopsymbol retains the set of almost-achievable trade-offs, including those obtained by mixing determinizations.

For probabilistic choice, the weighted sum of the branches' expectations is similarly lifted from values to sets of vectors.
Concretely, if the two branches yield sets $A$ and $B$ and are chosen with probabilities $q$ and $1-q$, respectively, we obtain all vectors $q \cdot a + (1-q) \cdot b$ for $a \in A$ and $b \in B$.
This is precisely the weighted Minkowski sum appearing in \Cref{table:mop}.

For loops, we take the least fixpoint of the characteristic function
\[
    \phiMopNo \eeq \mylambda X\quad [\neg \guard] \cdot \mpost \mminkow [\guard] \cdot \mop{\program'}{X}.
\]
As $\mExp$ is a complete lattice and this function is continuous, we know by Kleene's fixpoint theorem (\Cref{theo:Kleene})
that $\phiMopNo$ has a least fixpoint and further that we can approximate it arbitrarily closely by iterating on the bottom element $\zeroME$.
We discuss how to handle loops in \Cref{sec:loops}.

\begin{example}[Multiobjective Preexpectations of Deterministic Programs]
\label{ex:simple-mop-assignment}
    Consider the program $\ASSIGN{x}{1} \fatsemi \ASSIGN{y}{0}$ and the multiobjective postexpectation $\mpost = \dwcset{\vvechorizon{x}{y}}$.
    We compute \mopsymbol for the second assignment $\ASSIGN{y}{0}$ as well as the composed program as follows:
\begin{center}
    \begin{minipage}[t]{.45\textwidth}%
    \abovedisplayskip=-.5em%
    \begin{align*}
        & \mop{\ASSIGN{y}{0}}{\dwcset{\vvechorizon{x}{y}}} \\
        =\ & \dwcset{\vvechorizon{x}{y}}\subst{y}{0} \\
        =\ & \mylambda \pstate \dwcset{\vvechorizon{x}{y}}(\pstate[y \mapsto 0]) \\
        =\ & \dwcset{\vvechorizon{x}{0}}\\[-1em]
    \end{align*}%
    \normalsize%
    \end{minipage}\hfill%
    \begin{minipage}[t]{.45\textwidth}%
    \abovedisplayskip=-.5em%
        \begin{align*}
        & \mop{\ASSIGN{x}{1} \fatsemi \ASSIGN{y}{0}}{\dwcset{\vvechorizon{x}{y}}} \\
        =\ & \mop{\ASSIGN{x}{1}}{\mop{\ASSIGN{y}{0}}{\dwcset{\vvechorizon{x}{y}}}} \\
        =\ & \mop{\ASSIGN{x}{1}}{ \, \dwcset{\vvechorizon{x}{0}} \, } \\
        =\ & \dwcset{\vvechorizon{1}{0}}\\
    \end{align*}%
    \normalsize%
    \end{minipage}
\end{center}
    Unsurprisingly, $\mopsymbol$ tells us that we can at best achieve the values 1 for $x$ and 0 for $y$.
    This point as well as those below are the only values achievable in this program, since it is deterministic to begin with, i.e.\ there is no nondeterminism to be resolved.
    \myendmark
\end{example}
\noindent
\begin{example}[Multiobjective Preexpectations of Nondeterministic Programs]
\label{ex:mop-stanni-program}
    Consider again the program $\programrun = \GC{\true}{\ASSIGN{x}{1} \fatsemi \ASSIGN{y}{0}}{\true}{\ASSIGN{x}{0} \fatsemi \ASSIGN{y}{1}}$
    and the postexpectation $\dwcset{\vvechorizon{x}{y}}$. %
    We verify the annotations already given in \Cref{fig:mop-annotations}. %
    
    As in \Cref{ex:simple-mop-assignment}, we obtain $\dwcset{\vvechorizon{0}{1}}$ for the second branch and $\dwcset{\vvechorizon{1}{0}}$ for the first branch.
    Applying the rule for the guarded choice gives us $\dwcset{\vvechorizon{1}{0}} \mecup \dwcset{\vvechorizon{0}{1}}$, with $\mecup$ taking the Scott closed and convex union.
    The downward closure of this convex union is already Scott closed:
    the only elements added by $\cl{\:\cdot\:}$ are thus the same ones that are added by $\dwc{\:\cdot\:}$.
    Consequently, we get
    \begin{align*}    
        \mop{\programrun}{\dwcset{\vvechorizon{x}{y}}}(\pstate) 
        \eeq & \dwc{\convset{\vvechorizon{1}{0},\vvechorizon{0}{1}}}
    \end{align*}
    for all initial states $\pstate$, proving what we had already discussed in \Cref{sec:overview}:
    we can at best achieve any convex combination of $\vvechorizon{1}{0}$ and $\vvechorizon{0}{1}$.
    This proves the point $\vvechorizon{1}{1}$ to be unachievable since $\vvechorizon{1}{1} \not \in \mop{\programrun}{\dwcset{\vvechorizon{x}{y}}}(\pstate)$ for all $\pstate \in \States$.
    \myendmark
\end{example}

\subsection{Equivalence of Multiobjective Preexpectations and Almost Achievability}
\label{sec:sound}
The \mopsymbol transformer precisely characterizes the almost-achievable points, addressing \Cref{prob:mop}:
\begin{restatable}{theorem}{mopispareto}
\label{theo:mop-is-pareto}
    $
        \mop{\program}{\dwcset{\post}}(\pstate) \eeq \cl{\AchExp{\program}{\post}{\pstate}}. %
    $
\end{restatable}
\begin{proof}
    We prove both inclusions, giving the overall ideas here and spelling out the proof in \Cref{sec:mdp,sec:mop-mdp}.
    The inclusion \enquote{$\supseteq$} %
    essentially follows from straightforward structural induction on $\program$ (cf.\ \Cref{theo:mop-includes-ach}) and the fact that $\mop{\program}{\dwcset{\post}}(\pstate)$ is Scott closed.
    
    The reverse inclusion is more involved.
    Its proof connects the denotational semantics given by \mopsymbol to the operational semantics of \pGCL.
    This is analogous to classical soundness results for weakest preexpectations, where the expectation transformer is shown to coincide with expected rewards in the operational MDP \cite{batz2024jp}. %
    The argument consists of three correspondences.
    First, almost-achievable values in an MDP can be characterized as the least fixpoint of a multiobjective Bellman operator.
    Second, the \mopsymbol transformer coincides with the Bellman operator on the operational MDP.
    Finally, almost achievable points of a program coincide with those of the operational MDP.
    
    For the first step, consider an MDP $\MDP$ with a reward function $\OpRew$.
    We use reachability-reward objectives: rewards are collected when a designated target in $\MDP$ is reached for the first time.
    We denote the set of values achievable from a state $\opState$ by $\Ach{\MDP}{\OpRew}{\opState}$.
    Its closure is characterized by the least fixpoint of \emph{generalized Bellman operator} $\BellmanParetoNo$ on convex sets of reward vectors, see \Cref{theo:bellman-is-pareto}:%
    \begin{align}
        \bigl(\lfp \BellmanPareto{\MDP}{\OpRew}\bigr)(\opState) \eeq \cl{\Ach{\MDP}{\OpRew}{\opState}}.
        \label{eq:bellman-is-pareto}
    \end{align}%
    Probabilistic programs induce an MDP $\ProgMDP$, whose states record both the remaining program and the current program state, so that $(\program,\pstate)$ represents the execution of $\program$ from $\pstate$.
    The reward function $\ProgR{\dwcset{\post}}$ translates the multiobjective postexpectation to a reward function on $\ProgMDP$. %
    We prove that our transformer computes exactly the set of almost achievable values in $\ProgMDP$; \mopsymbol is the denotational counterpart of the Bellman operator.
    Consequently, their $\textsf{lfp}$'s coincide, see \Cref{theo:sound-mop}:
    \begin{align}
        \mop{\program}{\dwcset{\post}}(\pstate) \eeq \left(\lfp \BellmanPareto{\ProgR{\dwcset{\post}}}{\ProgMDP}\right)((\program,\pstate)),
        \label{eq:sound-mop}
    \end{align}
    Combining \Cref{eq:sound-mop,eq:bellman-is-pareto}, we get
    \[
        \mop{\program}{\dwcset{\post}}(\pstate)
        \overset{}{\eeq} \left(\lfp \BellmanPareto{\ProgR{\dwcset{\post}}}{\ProgMDP}\right)((\program,\pstate)) 
        \overset{}{\eeq} \cl{\Ach{\ProgMDP}{\ProgR{\post}}{(\program,\pstate)}}.
    \]
    Finally, we observe that mixed memoryless schedulers on $\ProgMDP$ give rise to determinizations, i.e.\ we have $\AchMemless{\ProgMDP}{\ProgR{\post}}{(\program,\pstate)} \subseteq \AchExp{\program}{\post}{\pstate}$ (\Cref{theo:sched-leq-determ}), implying the same inclusion for their respective closures.
    Together with the observation that memoryless schedulers suffice when considering the \emph{almost} achievable values (cf.\ \Cref{theo:memless-suff}), we get that
    $
        \mop{\program}{\dwcset{\post}}(\pstate) \, \subseteq\,  \cl{\AchExp{\program}{\post}{\pstate}}.
    $
    Both inclusions yield the desired equality.%
\end{proof}
\noindent%
Thus, although \mopsymbol is defined entirely at the program level, it has a precise operational interpretation.
Crucially, our results hold for uncountable state,\footnote{The state space of the operational MDP is uncountable as there are uncountably many \pGCL programs.
This gives us uncountably many initial states, but from each initial state only countably many reachable states.} finite-action MDPs, and thus without the finite-state restriction underlying seminal work on multiobjective MDP %
verification~\cite{chen2013stochastic,barret2008learning}.

Since the points achieved by mixed determinizations are precisely the convex combinations of points achieved by pure determinizations, \Cref{theo:mop-is-pareto} immediately gives:
\begin{restatable}{corollary}{mopwppure}
\label[corollary]{theo:mop-wp-pure}
    $
        \mop{\program}{\dwcset{\post}}(\pstate)
        \eeq
        \cl{\convsetcomp{ \wpVec{\program'}{\post}(\pstate) }{ \program' \determ \program }}.
    $
\end{restatable}

\noindent%
This characterization also recovers some special cases: if $\program$ is deterministic (cf.\ \Cref{theo:mop-wp-determ}) or $\post$ is single objective (cf.\ \Cref{theo:mop-wp-single}), then
\[
    \mop{\program}{\dwcset{\post}}
    \eeq
    \dwcset{\wpVec{\program}{\post}}.
\]
Thus, \mopsymbol is a conservative generalization of \wpsymbol.

\subsection{Healthiness Conditions}
\label{sec:health}

Healthiness conditions are routine results in predicate transformer semantics to characterize well-behaved operators and ensure compatibility with the underlying order-theoretic structure.
For the remainder of this section, we fix $\program$  and $\mpost,\meg  \in \mExp$.
We begin by showing that the \mopsymbol\ transformer is well-defined and satisfies the fundamental order-theoretic properties of monotonicity and continuity, which are stated in the following theorem.

\begin{restatable}[Basic Healthiness]{theorem}{mopwellcontmonotone}
\label{theo:mop-well-cont-monotone}
    $\mopC{\program}$ is well-defined, $\omega$-continuous, and monotone.
\end{restatable}

\noindent%
The \mopsymbol\ transformer enjoys several further properties, as Dijkstra-style transformers classically do.
An expectation $\me \in \mExp$ is \emph{bounded} by $b \in \PosReals^n$ if for all states $\pstate$ and $x \in \me(\pstate)$ we have $x \leq b$, which we denote by $\me \sqsubseteq b$.
    
\begin{restatable}[More Healthiness]{theorem}{mophealth}
\label{theo:mop-health}
    Let $r \in \PosReals$.
    The $\mopsymbol$ transformer is 
    \begin{itemize}
        \item sublinear: \tabto{\dimexpr 0.34\textwidth - \the\leftmargini}$\mop{\program}{r \cdot \me \mminkow \meg} \ssqsubseteq r \cdot \mop{\program}{\me} \mminkow \mop{\program}{\meg}$,
        
        \item positively homogeneous: \tabto{\dimexpr 0.34\textwidth - \the\leftmargini}$\mop{\program}{r \cdot \me} \eeq r \cdot \mop{\program}{\me}$, \quad and

        \item strict: \tabto{\dimexpr 0.34\textwidth - \the\leftmargini}$\mop{\program}{\zeroME} \eeq \zeroME$.
    \end{itemize}
    
    Moreover, if $\me \sqsubseteq b$ is bounded by $b \in \PosReals^n$, then $\mopsymbol$ is feasible:~{} $\mop{\program}{\me} \ssqsubseteq b$.
\end{restatable}

\section{Invariant-Based Reasoning for Loops}
\label{sec:loops}
To apply the multiobjective preexpectation transformer to programs containing loops, we need deductive techniques for establishing bounds on its least fixpoints.
Fix a loop $\program = \WHILEDO{\guard}{\program'}$ and a postexpectation $\mpost \in \mExp$ throughout this section.
Reasoning about $\mop{\program}{\mpost}$ amounts to reasoning about the least fixpoint of
\[
    \phiMopNo \eeq \mylambda{X} [\neg \guard] \cdot \mpost \mminkow [\guard] \cdot \mop{\program'}{X}~.
\]
In principle, least fixpoints can be determined by \hyperref[theo:Kleene]{Kleene's fixpoint theorem}, which applies to $\phiMopNo$ since $\mExp$ is a complete lattice and $\phiMopNo$ is $\omega$-continuous (\Cref{theo:mop-well-cont-monotone}).
This amounts to iteratively applying $\phiMopNo$ to the least element $\zeroME$ and taking the supremum of the iterates.
Computing this supremum is however impossible in general.
Instead, we approximate least fixpoints by \emph{invariants}, exploiting that \mopsymbol operates on a complete lattice.
We say that $\invariant \in \mExp$ is a \emph{superinvariant} of $\phiMopNo$ if $\phiMopNo(\invariant) \sqsubseteq \invariant$, and a \emph{subinvariant} if $\invariant \sqsubseteq \phiMopNo(\invariant)$.

Additional material on this section including proofs can be found in \Cref{proofs:loops}.

\subsection{Upper Bounds}

Upper bounds on least fixpoints are obtained elegantly via \emph{Park induction}:
\begin{restatable}[Park Induction for Upper Bounds]{theorem}{parkupper}
\label[theorem]{theo:park-upper-mop}
    If $\invariant$ is a superinvariant of $\phiMopNo$, then 
    \(
    \lfp \phiMopNo \sqsubseteq \invariant ~.
    \)
\end{restatable}

\noindent%
This rule applies without any termination or boundedness assumption on expectations.

\begin{example}[Park Induction]
\label{ex:split-upper}
    Consider the program
    \[
        \program_{\text{split}} \eeq \WHILEDO{0 < k}{\ \ASSIGN{k}{k-1} \fatsemi \GC{\true}{\ASSIGN{x}{x+1}}{\true}{\ASSIGN{y}{y+1}}\,},
    \]
    where all variables range over $\Nats$, together with the postexpectation $\dwcset{\vvechorizon{x}{y}}$:
    the program distributes $k$ increments between $x$ and $y$, and we are interested in the trade-off between the two final values.
    We guess that from a state with $k$ remaining iterations, the achievable points are the distributions of the $k$ increments and their mixtures, i.e.\ we guess the invariant
    \[
        \invariant \eeq \dwc{\convset{\vvechorizon{x + k}{y},\ \vvechorizon{x}{y + k}}}~.
    \]
    We can show that $\phiMopNo(\invariant) = \invariant$ (see \Cref{ex:split-upper-full}), i.e.\ 
    the invariant $\invariant$ is a fixpoint, and in particular a superinvariant.
    \Cref{theo:park-upper-mop} thus yields 
    $
    \mop{\program_{\text{split}}}{\dwcset{\vvechorizon{x}{y}}} \sqsubseteq \invariant~.
    $\myendmark
\end{example}
\noindent%
Further induction principles for upper bounds on least fixpoints of monotone functions, such as $k$-induction~\cite{batz2021latticed}, carry over to our setting as well but are omitted here.

\subsection{Lower Bounds}

For lower bounds, one might hope for the dual rule: if $\invariant$ is a subinvariant, then $\invariant \sqsubseteq \lfp \phiMopNo$.
This rule is \emph{unsound}, already for the standard \wpsymbol transformer, which \mopsymbol subsumes (cf.\ \Cref{theo:mop-wp-single}), and obtaining sound lower-bound rules is considerably more challenging in general~\cite{hark2019aiming,kura2026supermartingales}. We now develop \emph{sufficient} conditions for subinvariants to be sound.

As a tool for formulating and proving these conditions, we introduce the \emph{multiobjective liberal preexpectation transformer} $\molpC{\program}\colon \mExp \to \mExp$, defined by the rules of \Cref{table:mop}, with the single difference that loops are interpreted by \emph{greatest} fixpoints, i.e.\ 
\[
    \molp{\WHILEDO{\guard}{\program'}}{\mpost}
    \eeq
    \gfp \phiMolpNo
    \quad \text{with} \quad
    \phiMolpNo \eeq \mylambda{X} [\neg \guard] \cdot \mpost \mminkow [\guard] \cdot \molp{\program'}{X}~.
\]
By construction, $\mop{\program}{\mpost} \sqsubseteq \molp{\program}{\mpost}$ for all $\program$ and $\mpost$.
Now, our techniques are driven by two key observations. First, and entirely dual to \Cref{theo:park-upper-mop}, subinvariants are sound for $\molpsymbol$ by \emph{Park co-induction}, i.e.\ $\invariant \sqsubseteq \phiMolpNo(\invariant)$ implies $\invariant \sqsubseteq \gfp \phiMolpNo$. Second, if $\lfp \phiMolpNo = \gfp \phiMolpNo$, then subinvariants become sound for $\mopsymbol$ as well because 
\[
    \invariant \sqsubseteq \phiMolpNo(\invariant)
    \qquad\underbrace{\text{implies}}_{\text{Park co-induction}}\qquad 
    \invariant ~{}\sqsubseteq{}~ \underbrace{\gfp \phiMolpNo ~{}={}~ \lfp \phiMolpNo}_{\text{assumption}}~.
\]
It thus remains to obtain sufficient conditions for $\lfp \phiMolpNo$ and  $\gfp \phiMolpNo$ to coincide.

For that, we observe that $\molpsymbol$ yields characterizations of program termination. $\molp{\program}{\zeroME}$ measures what can be \enquote{gained through nontermination} over \mopsymbol: terminating runs contribute the postexpectation $\zeroME$, so any nonzero outcome must be attained by divergence. 
Below, whenever expectations are \emph{bounded} by some $b \in \PosRealsVect$, i.e.\ $\mpost(\pstate) \leq b$ for all $\pstate$ (cf.\ \Cref{sec:health}), all constructions (in particular the greatest element $\topME$ and greatest fixpoints) are taken over the complete lattice of multiobjective expectations bounded by $b$. All omitted proofs can be found in \Cref{proofs:loops}.

\newpage
\begin{definition}[Demonic Termination]
\label{def:dast-dct}
We define the following notions:
    \begin{enumerate}
        \item
        \label{def:dast-dct:dast}
        $\program$ is \emph{demonically almost-surely terminating (dAST) on state $\pstate$} if
        \(
            \molp{\program}{\zeroME}(\pstate) = \{0\}~,
        \)
        where we take\footnote{This is necessary for $\molp{\program}{\zeroME}(\pstate) = \{0\}$ to actually characterize dAST. When applying this definition, we will always restrict to $b$-bounded expectations for a suitable $b$. We also remark that for $b=0$, this equation does not characterize dAST which is, however, not a problem since in this degenerate case $\mopsymbol$ and $\molpsymbol$ coincide trivially.%
        } the above $\molpsymbol$ over the complete sub-lattice of \emph{$b$-bounded} multiobjective postexpectations for some arbitrary but fixed bound $b<\infty$.
        \item 
        \label{def:dast-dct:dct}
        $\program$ is \emph{demonically certainly terminating (dCT) on state $\pstate$} if there is a $k \in \Nats$ with
        $
            (\phiMolpNo)^k(\topME)(\pstate) = \{0\},
        $
        where $\phiMolpNo$ is taken with respect to the post $\zeroME$.\myendmark
    \end{enumerate}
\end{definition}

\noindent%
Both notions capture \emph{demonic} termination, i.e.\ termination under every resolution of nondeterminism, even though \mopsymbol and \molpsymbol resolve nondeterminism angelically:
dAST states that no resolution gains anything from divergence, capturing that every determinization of $\program$ terminates with probability 1;
dCT states that
all executions terminate after some finite number of iterations.

The key step towards our lower-bound rules is that dAST (under boundedness of the post $\mpost$) and dCT (not requiring boundedness) make the liberal and the strict transformer (and thus $\lfp \phiMolpNo$ and  $\gfp \phiMolpNo$) coincide. Towards this end, we first bound the gap between $\mopsymbol$ and $\molpsymbol$:
\begin{restatable}{lemma}{residual}
\label[lemma]{theo:residual}
    For all $\program \in \pGCL$ and $\mpost, \meg \in \mExp$, we have
    \[
        \molp{\program}{\mpost \mminkow \meg} \ssqsubseteq \mop{\program}{\mpost} \mminkow \molp{\program}{\meg}.
    \]
\end{restatable}
\noindent%
With this result at hand, we get the following:
\begin{restatable}[Induction for Lower Bounds]{theorem}{inductionlower}
    \label[theorem]{theo:induction_lower}
    Let $\pstate \in \States$. We have:
    \begin{enumerate}
        \item
        \label{theo:induction_lower:dast}
        If $\mpost$ and $\invariant$ are bounded, $\program$ is dAST on $\pstate$, and $\invariant$ is a subinvariant of $\phiMopNo$, then 
        \[
            \invariant(\pstate) \subseteq \mop{\program}{\mpost}(\pstate) ~. %
        \]
        \item 
        \label{theo:induction_lower:dct}
        If $\program$ is dCT on $\pstate$ and $\invariant$ is a subinvariant of $\phiMopNo$, then 
        $
            \invariant(\pstate) \subseteq \mop{\program}{\mpost}(\pstate)~.
        $
    \end{enumerate}
\end{restatable}

Revisiting \Cref{ex:split-upper}, the invariant $\invariant$ there is an exact fixpoint, hence in particular a subinvariant.
The loop $\program_{\text{split}}$ is dCT for every $\pstate$:
every iteration decreases $k$, so from a state with counter value $k$, the descending iterates of $\phiMolpNo$ with respect to $\zeroME$ reach $\zeroHoare$ after $k+1$ unfoldings.
Since the postexpectation $\dwcset{\vvechorizon{x}{y}}$ is unbounded, \Cref{theo:induction_lower}.\ref{theo:induction_lower:dast} is not applicable, but \Cref{theo:induction_lower}.\ref{theo:induction_lower:dct} is, and together with the upper bound from \Cref{ex:split-upper} we obtain exactness:
$
    \mop{\program_{\text{split}}}{\dwcset{\vvechorizon{x}{y}}} = \invariant.
$
See \Cref{ex:geo-lower} for an example of the application of \Cref{theo:induction_lower}.\ref{theo:induction_lower:dast}.

\paragraph{Lower $\omega$-invariants}

The rules above hinge on termination.
Our final rule applies to \emph{every} loop and \emph{every} postexpectation.
Instead of certifying one invariant against the fixpoint, it certifies a whole \emph{sequence of invariants} against the ascending Kleene iterates, stage by stage.
A sequence $\invariant\colon \Nats \to \mExp$ is a \emph{lower $\omega$-invariant} if
\[
    \invariant_0 \ssqsubseteq \phiMopNo(\zeroME)
    \qquad \text{and} \qquad
    \invariant_{k+1} \ssqsubseteq \phiMopNo(\invariant_k) \quad \text{for all } k \in \Nats.
\]

\begin{restatable}[Lower Bounds from $\omega$-invariants]{theorem}{cousotmop}
\label{theo:cousot-mop}
    If $\invariant$ is a lower $\omega$-invariant, then 
    \[\bigsqcup_k \invariant_k \ssqsubseteq \lfp \phiMopNo~. \]
\end{restatable}

\noindent%
This rule requires no termination or boundedness side conditions whatsoever.
The price is that it comes with infinitely many proof obligations, and that the stages $\invariant_k$ must track the unrolling depth explicitly, which quickly becomes cumbersome.
See \Cref{ex:geo-omega} for an example.

\section{From Weighted Expectations to Synthesized Determinizations}
\label{sec:wsum}

So far, we were reasoning about the \emph{almost} achievable points for a given program, which can be approximated arbitrarily closely by mixed determinizations.
We now turn to the task of synthesizing a determinization that realizes a given (\emph{exactly}, rather than almost) achievable point: %

\begin{problemstatement}
\label{prob:mop-synth}
    Given a program $\program \in \pGCL$, a postexpectation vector $\post \in \Exp^n$, an initial state $\pstate \in \States$, and a target $p \in \mop{\program}{\dwcset{\post}}(\pstate)$, synthesize---if one exists---a mixed determinization $\program' \determmixed \program$ such that $p \leq \wpVec{\program'}{\post}(\pstate)$.
\end{problemstatement}

\medskip%
\noindent%
We approach this problem through a weighted-sum approach, originating in convex multiobjective optimization~\cite{solanki1993,rennenDH2011} and further adapted to multiobjective model checking~\cite{forejt2012pareto,quatmann2023verification}. %
By combining multiple postexpectations into a single weighted expectation, we can apply existing programmatic strategy synthesis techniques by \citet{batz2024programmatic} to construct pure determinizations.
We then investigate how such determinizations can be mixed to realize a prescribed target.
Proofs and additional material can be found in \Cref{proofs:wsum}.

\begin{figure}[t]
    \centering
        \begin{subfigure}[t]{0.31\textwidth}
        \centering
        \scalebox{0.8}{
\begin{tikzpicture}[
    >=stealth
]

\coordinate (V1) at (2,3);
\coordinate (V2) at (3,2);

\filldraw[
    fill=mygreen!30,
    draw=mygreen,
    line join=round
]
    (0,0)
    -- (0,3)
    -- (V1)
    -- (V2)
    -- (3,0)
    -- cycle;

\draw[->] (0,0) -- (4,0);
\draw[->] (0,0) -- (0,4.7);

\draw[->, very thick, myblue]
    (0,0) -- (1.4,4.2)
    node[right] {$\weight$};

\node[mygreen!80!black] at (0.5,0.5) {$\HDel$};

\coordinate (Support1) at (1.1,3.3);

\draw[myred, thick,->]
    (V1) -- (Support1);

\fill[mygreen!80!black] (V1) circle (2pt);

\coordinate (P2) at (2.3,1.1);

\coordinate (ProjP2) at (0.56,1.68);

\draw[myred, thick, ->]
    (P2) -- (ProjP2);

\fill[mygreen!80!black] (P2) circle (2pt);

\end{tikzpicture}}
\caption{The projections (red) of two points onto a weight vector $\weight$.
The upper point $x$ %
maximizes %
$\weight \cdot x$.
}
\label{fig:proj}
\end{subfigure}
\hfill
\begin{subfigure}[t]{0.31\textwidth}
\centering
\scalebox{0.8}{
\begin{tikzpicture}[
    >=stealth
]

\coordinate (V1) at (2,3);
\coordinate (V2) at (3,2);

\filldraw[
    fill=mygreen!30,
    draw=mygreen,
    line join=round
]
    (0,0)
    -- (0,3)
    -- (V1)
    -- (V2)
    -- (3,0)
    -- cycle;

\draw[->] (0,0) -- (4,0);
\draw[->] (0,0) -- (0,4.7);

\draw[->, very thick, myblue]
    (0,0) -- (1.4,4.2)
    node[right] {$\weight$};

\draw[myred, thick,->]
    (V1) -- (Support1);

\draw[->, thick, shorten >=2pt, mygreen!80!black]
  ($(V1)+(-0.6,-0.8)$)
  to[bend left=20]
  node[below=9pt] {$\face{\weight}{\HDel}$}
  (V1);

\node[mygreen!80!black] at (0.5,0.5) {$\HDel$};

\fill[mygreen!80!black] (V1) circle (2pt);

\coordinate (V3) at (3,0);

\draw[->, very thick, myblue]
    (0,0) -- (3.6,0)
    node[above] {$\weight'$};

\draw[ultra thick, mygreen!80!black] (V3) -- (V2);

\draw[->, thick, shorten >=2pt, mygreen!80!black]
  ($(V2)+(0.6,0.6)$)
  to[bend left=15]
  node[above=20pt] {$\face{\weight'}{\HDel}$}
  ($(V2)!0.5!(V3)$);

\end{tikzpicture}}
\caption{The faces of weight vectors $\weight$ and $\weight'$, one being a singleton and another an entire line segment.
}
\label{fig:face}
\end{subfigure}
\hfill
\begin{subfigure}[t]{0.31\textwidth}
\centering
\scalebox{0.8}{
\begin{tikzpicture}[
    >=stealth
]

\coordinate (V1) at (2,3);
\coordinate (V2) at (3,2);

\fill[myblue!15]
    (0,0) -- (3.9,0) -- (3.9,{71/30}) -- (0,{11/3}) -- cycle;

\filldraw[
    fill=mygreen!30,
    draw=mygreen,
    line join=round
]
    (0,0)
    -- (0,3)
    -- (V1)
    -- (V2)
    -- (3,0)
    -- cycle;

\draw[->] (0,0) -- (4,0);
\draw[->] (0,0) -- (0,4.7);

\draw[draw=myblue]
(0,{11/3})
-- (3.9,{71/30});

\draw[->, very thick, myblue]
    (0,0) -- (1.4,4.2)
    node[right] {$\weight$};

\node[mygreen!80!black] at (0.5,0.5) {$\HDel$};

\end{tikzpicture}}
\caption{The supporting halfspace (light blue) with respect to $\weight$ and $\HDel$.
}
\label{fig:halfspace}
\end{subfigure}
    \caption{Illustrations of the geometric notions used in \Cref{sec:wsum}.}
    \label{fig:geometric}
    \Description{
    Three geometric illustrations are shown side by side, depicting projections, faces, and supporting halfspaces of a convex set $H$.
    In all three illustrations, the set $H$ is represented by a light green shaded polygon, and the weight vector $w$ is shown as a blue arrow pointing upward and to the right.

    The left illustration shows the projections of two points onto the weight vector $w$.
    The two points are marked in dark green, and their projections are indicated by red arrows.
    The upper point maximizes the weighted scalar product $w \cdot x$ over the set $H$.

    The middle illustration shows the faces of $H$ associated with two different weight vectors $w$ and $w'$.
    The face $\operatorname{face}(w,H)$ consists of a single vertex at the upper boundary of $H$.
    In contrast, the face $\operatorname{face}(w',H)$ consists of an entire vertical line segment on the right boundary of $H$.
    The weight vector $w'$ is represented by a horizontal blue arrow pointing to the right.

    The right illustration shows the supporting halfspace of $H$ with respect to the weight vector $w$.
    A descending blue line passes through the upper boundary of $H$ and defines the boundary of the supporting halfspace.
    The supporting halfspace is shaded light blue and contains the entire green set $H$.
}
\end{figure}

\paragraph{Geometric background}
Illustrations of the required notions are provided in \Cref{fig:geometric}.
A \emph{weight vector} $\weight \in [0,1]^n$ is a vector such that $\sum_i \weight_i = 1$. %
We denote the set of all weight vectors by $\Weights$.
For $S \subseteq \PosRealsInfVect$, we define its weighting with respect to a weight vector $\weight \in \Weights$ as
$
    \weight \cdot S = \set{\weight \cdot x \mid x \in S} \subseteq \PosRealsInf.
$
Let $\HDel \in \HoarePowerDom$.
Geometrically, the weighted sum $\weight \cdot x$ for $x \in \HDel$ is equal to the length of the projection of $x$ onto the direction $\weight$, scaled by $\lVert w \rVert_2$.
The \emph{face} induced by $\weight$ is the set of points $x \in \HDel$ maximizing this projection, i.e.\ 
\(
    \face{\weight}{\HDel}
    \eeq
    \set{x \in \HDel \mid
    \weight \cdot x =
    {\sup_{\smash{y \in \HDel}}} \weight \cdot y
    }.
\)
The corresponding \emph{supporting halfspace}
$
    \set{ x \mid \weight \cdot x \leq {\sup_{\smash{y \in \HDel}}} \weight \cdot y}
$
contains all points whose weighted sum is no larger than the maximum achieved in direction $\weight$.

For the remainder of this section, let $\program \in \pGCL$, $\pstate \in \States$, $\post \in \Exp^n$, and $\weight \in \Weights$.

\subsection{Weighted Sums of Multiple Expectations}
\label{sec:wsum-exp}

A weight vector $\weight \in \Weights$ allows us to combine multiple postexpectations $\post= \vvvechorizon{\post_1}{\dots}{\post_n}$ into a single weighted postexpectation. %
For example, for $\weight = \vvechorizon{\sfrac{2}{3}}{\sfrac{1}{3}}$ and the two expectations $f_1=x$ and $f_2 = y$, we get
$
    \vvechorizon{\sfrac{2}{3}}{\sfrac{1}{3}} \cdot \vvechorizon{x}{y}
    \eeq \sfrac{2}{3}\cdot x + \sfrac{1}{3}\cdot y \in \Exp.
$
The following theorem establishes an exact correspondence between classical \wpsymbol applied to such weighted postexpectations and \mopsymbol:

\begin{restatable}{theorem}{wpasweightedmop}
\label{theo:wp-as-weighted-mop}
    $
        \wp{\program}{\weight \cdot \post}(\pstate) \eeq \max \set{\weight \cdot x \mid x \in \mop{\program}{\dwcset{\post}}(\pstate)}.
    $
\end{restatable}

\begin{figure}[t]
    \centering
    \begin{subfigure}[t]{0.45\textwidth}
        \centering
        \scalebox{0.85}{
\begin{tikzpicture}[
    >=stealth
]

\coordinate (V1) at (2,3);
\coordinate (V2) at (4,2);

\coordinate (W) at (4,3);

\coordinate (ProjV1) at (2.72,2.04);

\filldraw[
    fill=mygreen!30,
    draw=mygreen,
    line join=round
]
    (0,0)
    -- (0,3)
    -- (V1)
    -- (V2)
    -- (4,0)
    -- cycle;

\draw[->] (0,0) -- (5,0);
\draw[->] (0,0) -- (0,4.7);

\coordinate (Support) at (1.1,3.3);

\draw[->, thick, myblue]
    (0,0) -- (1.4,4.2)
    node[right] {$\weight$};

\draw[myred, thick,->]
    (V1) -- (Support);

\draw[
    myblue, %
    thick,
    decorate,
    decoration={brace, mirror, amplitude=6pt}
]
    (0.1,0.05) -- ($(Support)+(0.08,-0.1)$)
    node[midway,below=8pt,rotate=70,myblue,font=\small]
    {\phantom{a}$\mathit{dist}_\weight$};

    \fill[mygreen!80!black] (V1) circle (2pt);

\node[mygreen!80!black,font=\small] at (2.5,0.3) {$\mop{\program}{\dwcset{\post}}(\pstate)$};

\end{tikzpicture}}

\caption{Illustration of \Cref{theo:wp-as-weighted-mop}: $\wp{\program}{\weight \cdot \post}(\pstate)$ is equal to $\mathit{dist}_\weight \cdot \lVert w \rVert_2$, i.e., the scaled length of the maximal projection of $\mop{\program}{\dwcset{\post}}(\pstate)$ (green) onto $\weight$.
}
\label{fig:wp-as-weighted-mop}
    \end{subfigure}
    \hfill
\begin{subfigure}[t]{0.47\textwidth}
        \centering
        \scalebox{0.85}{
\begin{tikzpicture}[
    >=stealth
]

\coordinate (V1) at (2,3);
\coordinate (V2) at (4,2);

\filldraw[
    fill=mygreen!30,
    draw=mygreen,
    line join=round
]
    (0,0)
    -- (0,3)
    -- (V1)
    -- (V2)
    -- (4,0)
    -- cycle;

\draw[->] (0,0) -- (5,0);
\draw[->] (0,0) -- (0,4.7);

\draw[draw=myblue,thick]
(0,{11/3})
-- (4,{7/3});

\draw[->, thick, myblue]
    (0,0) -- (1.4,4.2)
    node[right] {$\weight$};

\draw[->, thick, shorten >=2pt, mygreen!80!black]
  ($(V1)+(-0.4,-0.8)$)
  to[bend left=20]
  node[below right=12pt and -20pt,font=\small] {$\wpVec{\program_\weight}{\post}(\pstate)$}
  (V1);

\fill[mygreen!80!black] (V1) circle (2pt);

\node[mygreen!80!black,font=\small] at (2.5,0.3) {$\mop{\program}{\dwcset{\post}}(\pstate)$};

\end{tikzpicture}}
\caption{Illustration of \Cref{theo:determ-wp-as-mop}: in this example, $\face{\weight}{\mop{\program}{\dwcset{\post}}(\pstate)}$ contains only the point $\wpVec{\program_\weight}{\post}(\pstate)$.
}
\label{fig:wsum-nonach}
    \end{subfigure}
    
    \caption{Illustrations of the relation between \mopsymbol and \wpsymbol for weighted expectations.}
    \label{fig:mop-wp-ill}

\Description{
    Two geometric illustrations are shown side by side, depicting the relationship between multiobjective weakest preexpectations and weighted weakest preexpectations.
    In both illustrations, the multiobjective expectation set is represented by a light green shaded polygon, and the weight vector $w$ is shown as a blue arrow pointing upward and to the right.

    The left illustration illustrates Theorem 7.1 by showing the maximal projection of the green set onto the weight vector $w$.
    A dark green point marks the vertex that maximizes the weighted scalar product.
    A red arrow indicates the projection of this point onto the direction of $w$.
    A curved blue annotation along the weight vector represents the weighted weakest preexpectation $\operatorname{wp}\llbracket C \rrbracket(w \cdot f)(\sigma)$.
    This value corresponds to the length of the maximal projection of the multiobjective expectation set onto $w$.

    The right illustration illustrates Lemma 7.3 by showing the face of the green set associated with the weight vector $w$.
    A descending blue line passes through the upper vertex of the polygon and represents a supporting hyperplane.
    The dark green point at this vertex is labeled with the vector-valued weighted weakest preexpectation $\overrightarrow{\operatorname{wp}}\llbracket C_w \rrbracket(f)(\sigma)$.
    A curved green arrow points to this vertex, indicating that the corresponding face consists of exactly this single point.
    Together, the illustrations demonstrate how weighted weakest preexpectations relate to projections and supporting faces of multiobjective expectation sets.
}

\end{figure}

\noindent%
As illustrated in \Cref{fig:wp-as-weighted-mop}, \Cref{theo:wp-as-weighted-mop} shows that $\wp{\program}{\weight \cdot \post}(\pstate)$ is the scaled %
length of the maximal projection of any point in $\mop{\program}{\dwcset{\post}}(\pstate)$ onto the direction $\weight$.
In particular, the maximal projection is attained at a boundary point.
Since $\mopsymbol$ includes its boundary, we obtain a maximum rather than merely a supremum in \Cref{theo:wp-as-weighted-mop}. %
See \Cref{proof:wp-as-weighted-mop} for a detailed proof.

By varying the weight vectors $\weight$ we can recover the entire multiobjective preexpectation, due to its convexity---this reduction is known as \emph{scalarization} \cite{forejt2012pareto}: %
by \Cref{theo:wp-as-weighted-mop} we obtain an overapproximation of \mopsymbol for each weight vector $\weight$ as the supporting halfspace, and taking the intersection of these halfspaces over all weight vectors recovers $\mopsymbol$ exactly (see \Cref{theo:mop-as-weighted-wp}).
Thus, scalarization is lossless for characterizing the multiobjective preexpectation, but not necessarily for \emph{targeting} a prescribed point on it, as we will discuss in the next section.

\paragraph{Optimal determinizations for weighted expectations}
So far, we have used weighted sums to connect multiobjective preexpectations to classic single objective preexpectations.
We now turn to their role in synthesizing concrete determinizations.

A determinization $\program' \determmixed \program$ is \emph{optimal} with respect to a single objective $\postg \in \Exp$ if
$
    \wp{\program'}{\postg}
    =
    \wp{\program}{\postg}.
$
Such an optimal determinization need not exist, but whenever one exists, a pure determinization suffices \cite[Theorem 2.5]{batz2024programmatic}.
Suppose that there exists a pure determinization $\program_\weight$ optimal for $\weight \cdot \post$.
We observe the following.%
\begin{restatable}{lemma}{wpofdeterm}
\label[lemma]{theo:wpofdeterm}
    $
        \weight \cdot \wpVec{\program_\weight}{\post} = \wp{\program}{\weight \cdot\post}.
    $
\end{restatable}
\noindent%
Importantly, optimizing the single weighted expectation yields more than just its optimal value:
the resulting determinization $\program_\weight$ also provides the expected values of \emph{all individual objectives} through $\wpVec{\program_\weight}{\post}$.
By \Cref{theo:wpofdeterm}, their weighted sum is optimal, so this vector lies on the face of \mopsymbol induced by $\weight$, as formalized in the following lemma.

\begin{restatable}{lemma}{determwpasmop}
\label[lemma]{theo:determ-wp-as-mop}
    $
        \wpVec{\program_\weight}{\post}(\pstate) \in \face{\weight}{\mop{\program}{\dwcset{\post}}(\pstate)}.
    $
\end{restatable}

\noindent
As illustrated in \Cref{fig:wsum-nonach}, \Cref{theo:determ-wp-as-mop} proves that $\program_\weight$ is, for suitable weight vectors, exactly the kind of program we seek:
a determinization achieving a point on the Pareto front.
Crucially, we can obtain it using existing single-objective strategy synthesis techniques à la \citet{batz2024programmatic} by optimizing for $\weight \cdot \post$.
This provides the foundation for our synthesis approach. %

\subsection{Synthesizing Optimal Determinizations}
\label{sec:synthesis}

Addressing \Cref{prob:mop-synth}, we now discuss which points can be targeted using weighted expectations.
This leads us to the notions of extreme and exposed points:
an \emph{extreme point} of $\HDel \in \HoarePowerDom$ is a point $x \in \HDel$ such that %
there are no $a,b \in \HDel$ with $a \neq x$, $b \neq x$, and $\weightc \in (0,1)$ such that $x = \weightc \cdot a + (1-\weightc) \cdot b$.
An extreme point $x$ is \emph{exposed} by $\weight$ if $\face{\weight}{\HDel} = \set{x}$.
\Cref{theo:determ-wp-as-mop} immediately yields a way of synthesizing a determinization that realizes exposed points:

\begin{corollary}
    If $p \in \mop{\program}{\dwcset{\post}}(\pstate)$ is exposed by $\weight \in \Weights$ and $\program_\weight \determ \program$ is optimal for $\weight \cdot \post$, then $\wpVec{\program_\weight}{\post}(\pstate) = p$.
\end{corollary}

\noindent%
Thus, an exposed point can be realized by synthesizing a determinization optimal for its exposing weighted objective, provided such a determinization exists.
We now extend this approach to convex combinations of exposed points, which can be realized by mixing the corresponding determinizations.
Formally, a point $p \in \HDel$ is \emph{convexly exposed} by $\vvechorizon{\weightc_1}{\weight_1},\, \ldots,\,  \vvechorizon{\weightc_m}{\weight_m}$ if $\weightc_i \in [0,1]$ with $\sum_i \weightc_i = 1$ and $p = \sum_i \weightc_i \cdot p_i$, where each $p_i \in \HDel$ is a point exposed by $\weight_i \in \Weights$.

\begin{corollary}
\label{theo:convexexposed}
    Let $p \in \mop{\program}{\dwcset{\post}}(\pstate)$ be convexly exposed by $\vvechorizon{\weightc_1}{\weight_1}, \dots, \vvechorizon{\weightc_m}{\weight_m}$.
    Suppose that for all $i$, there exists a pure determinization $\program_i \determ \program$ optimal for $\weight_i \cdot \post$.
    Then, the mixed determinization $\program' \determmixed \program$ which executes $\program_i$ with probability $\weightc_i$ satisfies $\wpVec{\program'}{\post}(\pstate) = p$.
\end{corollary}

\noindent%
\Cref{theo:convexexposed} provides a constructive approach to synthesizing a determinization that meets a prescribed target---provided suitable exposed points and optimal pure determinizations can be found---thus yielding the following synthesis recipe:
\begin{enumerate}
    \item Determine $\mop{\program}{\dwcset{\post}}(\pstate)$.
    \item Find a point $p' \geq p$ that is convexly exposed by $\vvechorizon{\weightc_1}{\weight_1}, \dots \vvechorizon{\weightc_m}{\weight_m}$.
    \item For each $\weight_i$, apply programmatic strategy synthesis %
    to obtain a pure determinization $\program_i \determ \program$ optimal for $\weight_i \cdot \post$.
    In order to obtain invariants for the single-objective \wpsymbol computations, we can scalarize the multiobjective invariants from Step~1, avoiding the need to discover new invariants.
\end{enumerate}
The mixed determinization $\program' \determmixed \program$ that executes $\program_i$ with probability $\weightc_i$ achieves $p'$ (and therefore~$p$), thus solving \Cref{prob:mop-synth}.

\begin{wrapfigure}[16]{R}{.3\textwidth}
\centering
\vspace{-0.8em}
\scalebox{0.9}{
\begin{tikzpicture}[
    scale=2,
    >=stealth
]

\filldraw[fill=mygreen!30,draw=mygreen]
    (0,2)
    arc[start angle=90,end angle=0,radius=1]
    -- (1,0)
    -- (0,0)
    -- cycle;

\draw[->] (0,0) -- (1.6,0) node[below] {$x$};
\draw[->] (0,0) -- (0,2.3) node[left] {$y$};

\foreach \x in {1}
    \draw (\x,0.01) -- (\x,-0.01) node[below] {\x};

\foreach \y in {1,2}
    \draw (0.01,\y) -- (-0.01,\y) node[left] {\y};

\draw[->, very thick, myblue]
    (0,0) -- (1.4,0)
    node[above] {$\weight$};
\draw[draw=myblue,thick]
(1,0)
-- (1,1.8);

\coordinate (t2) at (1,1);
\fill[myred] (t2) circle (1pt);
\node[myred,right] at ($(t2)$) {$\target$};
\end{tikzpicture}}%
\vspace*{-.75em}%
\caption{A non-exposed extreme point $\target$.
The projection onto $\weight$ is maximized not only at $\target$, but at every point on the segment from $\vvechorizon{1}{0}$ to $\vvechorizon{1}{1}$.}
\label{fig:non-exposed}
\Description{
    A two-dimensional plot illustrates a non-exposed extreme point $\tau$ of a convex set.
    The horizontal axis represents $x$, while the vertical axis represents $y$.
    The convex set is shaded light green and is bounded by the coordinate axes, a curved upper boundary, and a vertical line segment at $x=1$.
    The curved boundary extends from $(0,2)$ to $(1,1)$, while the vertical segment extends from $(1,1)$ to $(1,0)$.
    A red point labeled $\tau$ is located at $(1,1)$, where the curved boundary meets the vertical segment.
    The weight vector $w$ is represented by a blue horizontal arrow pointing to the right along the positive $x$-axis.
    A blue vertical line at $x=1$ indicates the points maximizing the projection onto $w$.
    Although $\tau$ is an extreme point of the convex set, the projection onto $w$ is maximized by every point on the vertical segment from $(1,0)$ to $(1,1)$.
    This illustrates that $\tau$ is a non-exposed extreme point because it cannot be uniquely identified as the maximizer of a linear functional.
}

\end{wrapfigure}

The synthesis recipe applies successfully to all our case studies in \Cref{sec:examples}.
However, its applicability depends on two things:
first, we must identify suitable weight vectors exposing points whose convex combination meets the target, and
second, optimal pure determinizations must exist for the corresponding weighted objectives.
We discuss these two aspects in the following.

\paragraph{Finding exposing weights}
In our case studies, we determine suitable weights %
by geometric analysis of the multiobjective preexpectation.
We seek exposed points whose convex combination dominates the prescribed target, together with weight vectors exposing these points; see \Cref{sec:robotics} for a concrete example.

If an extreme point is not exposed, it cannot be recovered by our approach:
only for such a point, any weight vector that maximizes the weighted expectation at the desired point also maximizes it at other points on the same face.
A synthesized determinization optimal for this weighted expectation therefore is not guaranteed to achieve the desired point. %

Indeed, in general, multiobjective preexpectations might contain non-exposed extreme points.
\Cref{fig:non-exposed} shows a set with one non-exposed extreme point; a concrete program whose multiobjective preexpectation has this shape is given in \Cref{ex:circle}.
However, %
such points can be arbitrarily approximated; in other words, the set of exposed points is dense in the extreme points:
\begin{restatable}{theorem}{strasz}
\label{theo:strasz}
    Every Pareto optimal extreme point $p \in \mop{\program}{\dwcset{\post}}(\pstate) \cap \PosRealsVect$ is the limit of a sequence of exposed points in $\mop{\program}{\dwcset{\post}}(\pstate)\cap \PosRealsVect$.
\end{restatable}
\noindent%
Thus, although our approach cannot create determinizations that achieve non-exposed extreme points, we can create determinizations achieving arbitrarily close approximations of those points, provided they are finitely valued and the respective optimal determinizations exist.

This limitation disappears for certain classes of programs.
In particular, if $\program$ is dCT (cf.\ \autoref{def:dast-dct}.\ref{def:dast-dct:dct}) on initial state $\pstate$ and post{\-}expectation $\post$ is finitely valued, i.e.\ $\post(\pstate') \in \PosRealsVect$ for all $\pstate' \in \States$, then every Pareto optimal extreme point of $\mop{\program}{\dwcset{\post}}(\pstate)$ is exposed.
Indeed, dCT ensures that only finitely many states are reachable from $\pstate$, yielding a finite operational MDP.
For finite MDPs, the set of achievable points is a polytope \cite[Lemma 2.3]{quatmann2023verification}, whose Pareto optimal extreme points are exposed. %

\paragraph{Existence of optimal determinizations}
Our synthesis approach depends on optimal pure determinizations for the corresponding weighted objectives.
Such determinizations need not exist, as is already known from the single-objective setting. In this case, our synthesis construction cannot be applied with the respective weights, although a determinization achieving the prescribed multiobjective target may still exist.

\section{Case Studies}
\label{sec:examples}

\subsection{The Robot: Trading Safety with Speed}
\label{sec:robotics}

\noindent%
We begin this section by applying the \mopsymbol transformer to the example $\program_{\text{robot}}$ (see \Cref{fig:robot_program_nondet}) with which we started in \Cref{sec:overview}.
Recall that, aiming to reach the goal position $g$ from the current position $p$, the robot has to choose which step to take:
fast and risky, or slow and safe.
We aim to simultaneously maximize the remaining time budget $t$ and the probability of reaching the goal $\iverson{p\geq g}$, encoded in the postexpectation $\dwcset{\vvechorizon{t}{\iverson{p\geq g}}}$.
To compute
\(
    \mop{\program_{\text{robot}}}{\dwcset{\vvechorizon{t}{\iverson{p \geq g}}}},
\)
we need to handle the present loop.
Since the robot needs at most $g$ steps to reach the goal, and the minimal progress in each iteration is one, the loop terminates certainly and we can apply \Cref{theo:induction_lower}.\ref{theo:induction_lower:dct} to prove an invariant exact. 
To understand the proposed invariant, we discuss the shape of the Pareto optimal strategies.

First, note that the last probabilistic choice never matters.
Within one round, the robot moves first and flips second.
If the move reaches the goal, the flip is irrelevant:
the guard has already failed through $p \geq g$, and neither the guard nor the post ever reads $f$ again.
Thus every Pareto optimal point is achieved by a route that ends with a fast step. %

For each $k = 1, \dots, \lceil \gap/2 \rceil$, consider the route ``commit to $k$ fast steps'':
$k-1$ fast steps first, then $\gap \monus 2k$ slow steps (the $k$ fast steps cover $2k$ cells; the slow ones make up the rest), then the free fast finish.
We do not need to consider larger $k$:
beyond $\lceil \gap/2 \rceil$ the leading $k-1$ and the final fast steps alone already suffice.

Two observations establish the optimality of these strategies:
First, ending fast is never worse than not ending fast due to the risk-free last move.
Second, we go \emph{riskiest first}:
the success probability is the product over the relevant (all but last) coins, hence order-invariant.
However, the expected number of rounds is not: fronting the risky coins lets failing runs fail early and terminate with a higher remaining time budget.
Risky-first hence maximizes the %
remaining time budget at fixed risk.

Denote by 
$
    \kappa_k %
$
the point achieved by the route taking $k$ fast steps.
Our proposed invariant $\invariant$\footnote{
Indeed, this invariant was proposed by Anthropic's Claude Fable 5 and subsequently verified in Lean.
}
combines these points, in the spirit of the Hoare domain, over all $k = 1, \dots, \lceil \gap/2 \rceil$; see \Cref{sec:inv-robot} for the concrete expression.
Then, we show that $\lfp \phiMopNo = \invariant$, which by \Cref{theo:induction_lower}.\ref{theo:induction_lower:dct} reduces to showing that $\phiMopNo(\invariant) = \invariant$.
This follows by a straightforward application of the rules in \Cref{table:mop}. %
For the entire program including initializations and assuming $g > 0$, we obtain %
\[ 
    \mop{\program_{\text{robot}}}{\dwcset{\vvechorizon{t}{\iverson{p \geq g}}}}(\pstate)
    \;=\;
    \bigsqcup_{k=1}^{\lceil g/2 \rceil} \dwcset{\, \vvechorizon{\, g - L_k(g)}{\ 0.7^{\,k-1} \cdot 0.95^{\,g \monus 2k} }\, };
\]
where 
\(
    L_k(g) \coloneqq \sfrac{(1-0.7^k)}{0.3} + 19\cdot 0.7^{\,k-1} \left(1-0.95^{\,g\monus 2k}\right).
\)

\begin{figure}[t]
    \centering

    \begin{subfigure}[t]{0.32\textwidth}
        \centering
        \scalebox{0.9}{
\begin{tikzpicture}[
    x={4cm/3.2}, %
    y={4cm/1.1},
    >=stealth
]

\foreach \x in {1,2,3}
    \draw (\x,0.01) -- (\x,-0.01) node[below] {\x};

\foreach \y/\label in {0.5/{0.5},1/{1}}
    \draw (0.01,\y) -- (-0.01,\y) node[left] {\label};

\coordinate (V1) at (1.1475,0.9025); %
\coordinate (V2) at (2.3,0.7);

\filldraw[
    fill=mygreen!30,
    draw=mygreen,
    line join=round
]
    (0,0)
    -- (0,0.9)
    -- (V1)
    -- (V2)
    -- (2.3,0)
    -- cycle;

\draw[->] (0,0) -- (3.2,0); %
\draw[->] (0,0) -- (0,1.1); %

\fill[mygreen!80!black] (V1) circle (2.2pt);
\node[mygreen!80!black] at ($(V1)+(-0,-0.1)$) {$\kappa_1$};

\fill[mygreen!80!black] (V2) circle (2.2pt);
\node[mygreen!80!black] at ($(V2)+(-0.14,-0.07)$) {$\kappa_2$};
\end{tikzpicture}}
\caption{$g=4$}
\label{fig:robot-pareto-mop-1}%
    \end{subfigure}%
    \hfill
    \begin{subfigure}[t]{0.32\textwidth}
        \centering
        \scalebox{0.9}{
\begin{tikzpicture}[
    x={4cm/3.2}, %
    y={4cm/1.1},
    >=stealth
]

\foreach \x in {1,2,3}
    \draw (\x,0.01) -- (\x,-0.01) node[below] {\x};

\foreach \y/\label in {0.5/{0.5},1/{1}}
    \draw (0.01,\y) -- (-0.01,\y) node[left] {\label};

\coordinate (V1) at (1.290125,0.857375);
\coordinate (V2) at (2.635,0.665);
\coordinate (V3) at (2.81,0.49);

\filldraw[
    fill=mygreen!30,
    draw=mygreen,
    line join=round
]
    (0,0)
    -- (0,0.857375)
    -- (V1)
    -- (V2)
    -- (V3)
    -- (2.81,0)
    -- cycle;

\draw[->] (0,0) -- (3.2,0); %
\draw[->] (0,0) -- (0,1.1); %
\foreach \x in {1,2,3}
    \draw (\x,0.01) -- (\x,-0.01) node[below] {\x};

\foreach \y/\label in {0.5/{0.5},1/{1}}
    \draw (0.01,\y) -- (-0.01,\y) node[left] {\label};

\fill[mygreen!80!black] (V1) circle (2.2pt);
\node[mygreen!80!black] at ($(V1)+(-0,-0.1)$) {$\kappa_1$};

\fill[mygreen!80!black] (V2) circle (2.2pt);
\node[mygreen!80!black] at ($(V2)+(-0.08,-0.07)$) {$\kappa_2$};

\fill[mygreen!80!black] (V3) circle (2.2pt);
\node[mygreen!80!black] at ($(V3)+(-0.15,-0.03)$) {$\kappa_3$};

\coordinate (P) at ($(V1)!0.6!(V2)$);
\fill[myblue] (P) circle (2.2pt);
\node[myblue,right] at ($(P)+(0,-0.07)$) {$\target$};

\end{tikzpicture}}
\caption{$g = 5$, plus a target $\target$}
\label{fig:robot-pareto-mop-2}
    \end{subfigure}%
    \hfill
\begin{subfigure}[t]{0.32\textwidth}
        \centering
        \scalebox{0.9}{
\begin{tikzpicture}[
    x={4cm/6.6}, %
    y={4cm/1.1},
    >=stealth
]

\coordinate (V1) at (2.26841, 0.698337);
\coordinate (V2) at (4.29129, 0.541647);
\coordinate (V3) at (5.48216, 0.420114);
\coordinate (V4) at (6.14115, 0.32585);
\coordinate (V5) at (6.2269, 0.2401);

\filldraw[
    fill=mygreen!30,
    draw=mygreen,
    line join=round
]
    (0,0)
    -- (0,0.7)
    -- (V1)
    -- (V2)
    -- (V3)
    -- (V4)
    -- (V5)
    -- (6.2269,0)
    -- cycle;

\draw[->] (0,0) -- (6.6,0); %
\draw[->] (0,0) -- (0,1.1); %
\foreach \x in {1,2,3,4,5,6}
    \draw (\x,0.01) -- (\x,-0.01) node[below] {\x};

\foreach \y/\label in {0.5/{0.5},1/{1}}
    \draw (0.01,\y) -- (-0.01,\y) node[left] {\label};

\fill[mygreen!80!black] (V1) circle (2.2pt);
\node[mygreen!80!black] at ($(V1)+(-0,-0.1)$) {$\kappa_1$};

\fill[mygreen!80!black] (V2) circle (2.2pt);
\node[mygreen!80!black] at ($(V2)+(-0.08,-0.1)$) {$\kappa_2$};

\fill[mygreen!80!black] (V3) circle (2.2pt);
\node[mygreen!80!black] at ($(V3)+(-0.03,-0.1)$) {$\kappa_3$};

\fill[mygreen!80!black] (V4) circle (2.2pt);
\node[mygreen!80!black] at ($(V4)+(0.48,0.02)$) {$\kappa_4$};

\fill[mygreen!80!black] (V5) circle (2.2pt);
\node[mygreen!80!black] at ($(V5)+(0.48,-0.02)$) {$\kappa_5$};

\end{tikzpicture}}
\caption{$g = 9$}
\label{fig:robot-pareto-mop-3}
    \end{subfigure}%
    \caption{$\mop{\program_{\text{robot}}}{\dwcset{\vvechorizon{t}{\iverson{p\geq g}}}}(\pstate) \;=\;
    \dwc{\convset{\, \kappa_1,\dots,\kappa_k \, }}$ for several initial states $\pstate$.
    The x-axis shows $t$ and the y-axis shows $\iverson{p\geq g}$.}
    \label{fig:robot-pareto-mop}
    
\Description{
    Three two-dimensional plots are shown side by side, illustrating multiobjective expectation sets for the robot program with different initial goal values $g$.
    In all three plots, the horizontal axis represents the remaining time $t$, while the vertical axis represents the probability of reaching the goal $[p \geq g]$.
    The achievable regions are shaded light green, with their upper boundaries forming piecewise-linear Pareto fronts.
    Dark green points labeled $\kappa_i$ indicate the vertices defining these fronts.

    The left plot corresponds to $g=4$ and shows a Pareto front defined by two vertices, $\kappa_1$ and $\kappa_2$.
    The first vertex has a higher probability of reaching the goal but a smaller remaining time, while the second vertex offers a larger remaining time at the cost of a lower success probability.
    The front consists of a horizontal segment followed by a descending segment connecting the two vertices.

    The middle plot corresponds to $g=5$ and shows a Pareto front defined by three vertices, $\kappa_1$, $\kappa_2$, and $\kappa_3$.
    The front consists of a horizontal segment followed by two descending segments, representing different trade-offs between the objectives.
    A blue point labeled $\tau$ lies on the segment connecting $\kappa_1$ and $\kappa_2$, indicating a target for strategy synthesis.

    The right plot corresponds to $g=9$ and shows a Pareto front defined by five vertices, $\kappa_1$ through $\kappa_5$.
    The front extends farther along the horizontal axis and exhibits a gradual decrease in the probability of reaching the goal as the remaining time increases.
    The vertices form a sequence of descending segments, with the final segment approaching a vertical boundary.

    Together, the three plots illustrate how changing the initial goal value $g$ affects the shape and complexity of the multiobjective expectation set and its Pareto front.
}

\end{figure}

\paragraph{Concrete frontiers}
\Cref{fig:robot-pareto-mop} illustrates the multiobjective preexpectation for three initial values of $g$. %
Let us look closer at $\pstate_5 = (g \mapsto 5)$ %
and denote $R \coloneqq \mop{\program_{\text{robot}}}{\dwcset{\vvechorizon{t}{\iverson{p\geq g}}}}(\pstate_5)$.
Assume we want to have, in expectation, at least two time units remaining. %
The result presented above proves that we can, for example, definitely find a determinization achieving at least $\sfrac{1}{2}$ probability of reaching the goal as $\vvechorizon{2}{\sfrac{1}{2}} \in R$.
Conversely, it also proves that we cannot reach the goal with $\sfrac{4}{5}$ probability as $\vvechorizon{2}{\sfrac{4}{5}} \not \in R$.

\paragraph{Synthesizing determinizations}
Keep $\pstate_5$ fixed. %
As discussed in \Cref{sec:synthesis}, we can systematically synthesize determinizations achieving certain targets.
As in \Cref{sec:overview-optimal-determ}, let $\target = 0.4 \cdot \kappa_1 + 0.6 \cdot \kappa_2 \approx \vvechorizon{2.09}{0.74}$ be the target we want the robot to achieve. %
The two closest extreme points of the Pareto front are $\kappa_1$ and $\kappa_2$, exposed by $\weight_1 = \vvechorizon{\sfrac{1}{10}}{\sfrac{9}{10}}$ respectively $\weight_2 = \vvechorizon{\sfrac{1}{3}}{\sfrac{2}{3}}$, 
making $\target$ convexly exposed by $((0.4,\ \weight_1),\; (0.6,\ \weight_2))$, see \Cref{fig:robot-pareto-weights}.

Next, we apply programmatic strategy synthesis and obtain pure determinizations $\program_{\text{robot}}^{\kappa_1}$ %
and $\program_{\text{robot}}^{\kappa_2}$ %
optimal for $\weight_1 \cdot \vvechorizon{t}{\iverson{p\geq g}}$ respectively $\weight_2 \cdot \vvechorizon{t}{\iverson{p\geq g}}$.
By reusing the \mopsymbol invariant above, this can be mechanized. %
For $\pstate_5$, the obtained determinizations commit to one respectively two fast steps:
$\program_{\text{robot}}^{\kappa_1}$ takes the fast branch if $g\monus p \leq 2$ (this is equivalent to the program $\program_{\text{robot}}^{\iverson{p \geq g}}$ displayed in \Cref{fig:robot_program_det});
$\program_{\text{robot}}^{\kappa_2}$ takes the fast branch if $g\monus p \neq 3$.
By mixing we obtain
\(
    \program_{\text{robot}}^{\vvechorizon{t}{\iverson{p \geq g}}}\; \coloneqq\; \PCHOICE{\program_{\text{robot}}^{\kappa_1}}{0.4}{\program_{\text{robot}}^{\kappa_2}},
\)
which by \Cref{theo:convexexposed} precisely achieves $\target$, i.e.\ 
\begin{align*}
    \wpVec{\program_{\text{robot}}^{\vvechorizon{t}{\iverson{p \geq g}}}}{t, \iverson{p \geq g}}
    \eeq & 0.4 \cdot \underbrace{\wpVec{\program_{\text{robot}}^{\kappa_1}}{t, \iverson{p \geq g}} }_{\kappa_1}
    \,+\, 0.6 \cdot \underbrace{\wpVec{\program_{\text{robot}}^{\kappa_2}}{t, \iverson{p \geq g}}}_{\kappa_2}  \eeq \target~.\\
\end{align*}

\subsection{The Casino: Gambling under Risk of Ruin}
\label{sec:casino}

Next, consider the game modeled by the program $\program_{\text{casino}}$ %
in \Cref{fig:casino-program}.
Each round, we choose between stopping or continuing to gamble.
If we stop, the game terminates immediately and we cash out the winnings $n$ accumulated so far.
If we continue, the potential winnings increase, but with some probability we lose all winnings altogether. %
The objective is to simultaneously optimize the game duration $i$ (more gambling fun) and---of course---the accumulated winnings $n$, thus trading off longer play for increased risk of total loss.

\begin{figure}[t]
\begin{subfigure}[t]{0.45\textwidth}
    \vspace{0pt}
    \centering
\begin{lstlisting}[mathescape,basicstyle=\small]
$\ASSIGN{n}{0} \fatsemi \ASSIGN{i}{0} \fatsemi \ASSIGN{c}{1} \fatsemi$
$\WHILE{c=1} \{$
    $\ASSIGN{i}{i+1} \fatsemi$
    $\GCFSTOPEN{~\boldsymbol{\blue{\trueBold}}}$ 
        $\gray{\text{\% stop and cash out}}$
        $\ASSIGN{c}{0} \fatsemi$
    $\GCSECCLOSE \GCSECOPEN{\boldsymbol{\blue{\trueBold}}}{}$
        $\gray{\text{\% gamble}}$
        $\PCHOICE{\ASSIGN{n}{n+1}}{\sfrac{2}{3}}{\COMPOSE{\ASSIGN{n}{0}}{\ASSIGN{c}{0}}}$
    $\GCSECCLOSE$
$\}$
\end{lstlisting}
\caption{The program $\program_{\text{casino}}$ with variables $n$ and $i$ ranging over $\Nats$ and $c$ over $\mathbb{B}$.}
\label{fig:casino-program}
\end{subfigure}
\hfill
\begin{subfigure}[t]{0.45\textwidth}
    \vspace{0pt}
    \centering
    \scalebox{0.7}{
\begin{tikzpicture}[
    x={6cm/3.2}, %
    y={6cm/1.1},
    >=stealth
]

\foreach \x in {1,2,3}
    \draw (\x,0.01) -- (\x,-0.01) node[below] {\x};

\foreach \y/\label in {0.5/{0.5},1/{1}}
    \draw (0.01,\y) -- (-0.01,\y) node[left] {\label};

\coordinate (V1)  at ({3-2^4/3^3},  {(2/3)^3*3});
\coordinate (V2)  at ({3-2^5/3^4},  {(2/3)^4*4});
\coordinate (V3)  at ({3-2^6/3^5},  {(2/3)^5*5});
\coordinate (V4)  at ({3-2^7/3^6},  {(2/3)^6*6});
\coordinate (V5)  at ({3-2^8/3^7},  {(2/3)^7*7});
\coordinate (V6)  at (2.921963115,0.312147538); %
\coordinate (V7)  at (2.947975410,0.234110654);
\coordinate (V8)  at (2.965316940,0.173415299);
\coordinate (V9)  at (2.976877960,0.127171219);
\coordinate (V10) at (2.984585307,0.092488160);
\coordinate (V11) at (2.989723538,0.066797004);
\coordinate (V12) at (2.993149025,0.047956823);
\coordinate (V13) at (2.995432683,0.034254874);
\coordinate (V14) at (2.996955122,0.024359021);
\coordinate (V15) at (2.997970082,0.017254307);
\coordinate (VL) at ({3},{0});

\filldraw[
    fill=mygreen!30,
    draw=mygreen,
    line join=round
]
    (0,0)
    -- (0,{8/9})
    -- (V1)
    -- (V2)
    -- (V3)
    -- (V4)
    -- (V5)
    -- (V6)
    -- (V7)
    -- (V8)
    -- (V9)
    -- (V10)
    -- (V11)
    -- (V12)
    -- (V13)
    -- (V14)
    -- (V15)
    -- (VL)
    -- cycle;

\draw[->] (0,0) -- (3.2,0) node[below] {$i$};
\draw[->] (0,0) -- (0,1.1) node[left] {$n$};

\foreach \i in {1,...,15}
    \fill[mygreen] (V\i) circle (1.5pt);

\fill[mygreen] (VL) circle (1.5pt);

\end{tikzpicture}}
\caption{$\mop{\program_{\text{casino}}}{\dwcset{\vvechorizon{i}{n}}}$ for any initial state.
}
\label{fig:casino-mop}
\Description{
    Two panels are shown side by side, illustrating the casino program and its corresponding multiobjective preexpectation set.

    The left panel presents the nondeterministic casino program $C_{\mathrm{casino}}$.
    The program initializes the variables $n$ and $i$ to zero and the Boolean control variable $c$ to one.
    While $c=1$, the counter $i$ is incremented by one, followed by a nondeterministic choice between stopping and gambling.
    The first branch stops the game by setting $c$ to zero.
    The second branch gambles, increasing $n$ by one with probability $2/3$ or resetting $n$ and $c$ to zero with probability $1/3$.
    The loop continues until the control variable $c$ becomes zero.

    The right panel shows the multiobjective preexpectation set of the casino program with respect to the objectives $i$ and $n$.
    The horizontal axis represents $i$, while the vertical axis represents $n$.
    The achievable region is shaded light green and is bounded above by a horizontal segment followed by a descending curve.
    The horizontal segment extends from the vertical axis to approximately $i=2.4$ at $n=0.89$.
    The descending boundary consists of a sequence of dark green points connected by line segments, approaching $i=3$ as $n$ tends to zero.
    The points become increasingly dense near the horizontal axis, illustrating the accumulation of achievable trade-offs.
    The plot demonstrates the relationship between the expected number of iterations and the expected gambling reward under different strategies.
}

\end{subfigure}
\caption{The program $\program_{\text{casino}}$ and its multiobjective preexpectation with respect to $\dwcset{\vvechorizon{i}{n}}$.}
\label{fig:casino}
\end{figure}

To compute \mopsymbol, we again prove an exact fixpoint of the characteristic function using rules for lower and upper bounds.
The invariant we propose for the upper bound is
\[
    \invariant = [c=0] \cdot \dwcset{\vvechorizon{i}{n}} \minkow  [c=1] \cdot \bigmecup_{m \in \Nats} \dwcset{\vvechorizon{i + 3-\sfrac{(2^{m+1})}{3^m}}{\left(\sfrac{2}{3}\right)^m \cdot (n+m)}}.
\]
$\invariant$ describes the possible states of the game depending on whether the player has stopped or is still gambling.
If $c=0$, the game has terminated and the current state $(i,n)$ represents the final cashed-out result.
If $c=1$, the game is still running.
The parameter $m$ can be thought of as the number of successful gambling rounds so far: 
after $m$ wins the expected duration is
$
    i+3-\sfrac{(2^{m+1})}{3^m},
$
and the winnings become $n+m$ with probability
$
    \left(\sfrac23\right)^m.
$
The Scott closed and convex union over $m\in\mathbb{N}$ captures all possible numbers of successful gambling rounds.
For this invariant, we can show that $\phiMopNo(\invariant) \, \eeq\, \invariant$ (see \Cref{sec:inv-casino} for detailed calculations), so $\invariant$ is a fixpoint of $\phiMopNo$.
With \Cref{theo:park-upper-mop}, we get that $\lfp \phiMopNo \leq \invariant$.

In contrast to \Cref{sec:robotics}, the program is not dCT and we cannot apply \Cref{theo:induction_lower}.\ref{theo:induction_lower:dct} for a lower bound.
It \emph{is} dAST: there is an infinite execution corresponding to gambling forever, but the probability of this execution is 0.
However, the postexpectation is unbounded, ruling out \Cref{theo:induction_lower}.\ref{theo:induction_lower:dast}.
Instead, we use \Cref{theo:cousot-mop} and propose an $\omega$-invariant $\invariant_k$, which essentially takes the supremum in $\invariant$ up to $m \leq k$, so $\invariant = \bigsqcup_k \invariant_k$.
We again refer to \Cref{sec:inv-casino} for details.
By \Cref{theo:cousot-mop}, we get that $\invariant \leq \lfp \phiMopNo$.
Both inclusions then let us conclude that we have indeed found the exact result for \mopsymbol. %
Including the initial assignments, we get
\[
    \mop{\program_{\text{casino}}}{\dwcset{\vvechorizon{i}{n}}} \eeq \dwc{\conv{\set{\vvechorizon{3}{0}} \cup \setcomp{\vvechorizon{3-\sfrac{(2^{m+1})}{3^m}}{\left(\sfrac{2}{3}\right)^m \cdot m} }{ m \in \Nats, m \geq 3}}},
\]
visualized in \Cref{fig:casino-mop}.
The Pareto front is formed by line segments between increasingly dense points converging to $\vvechorizon{3}{0}$.
All Pareto optimal extreme points are exposed, so our strategy synthesis approach applies to any target.
The Pareto optimal point with the highest expected winnings at approximately $\vvechorizon{2.4}{0.89}$ is %
achieved by gambling at most three times and cashing out in the fourth round if all three gambles succeed.
Notably, synthesis also applies to the limit point $\vvechorizon{3}{0}$, which is exposed by the weight vector $\vvechorizon{1}{0}$.
The corresponding weighted expectation $1 \cdot i + 0 \cdot n$ optimizes solely for the length of the game, and thus results in always choosing to continue gambling.
Consequently, with probability 1, the \enquote{lose-it-all} event eventually occurs, which is why the expected value for $n$ in this determinization is 0.

\section{Related Work}
\label{sec:related}

Existing predicate-transformer calculi primarily consider single objectives, such as expected values or runtimes \cite{morgan1996probabilistic,kaminski2018weakest}, and do not support simultaneous reasoning about multiple objectives.
Closest to our setting is \cite{batz2022weighted}, which instantiates weakest-pre-style calculi over general semirings, but admits only a single branching construct, thus not allowing the combination of nondeterministic and probabilistic choice.

Beyond predicate-transformer semantics, expected costs in probabilistic programs can be bounded using type-based techniques and abstract interpretation.
Type-based approaches derive bounds on expected costs, variances, and higher moments \cite{ankush2023probabilistic,wang2021central,wang2020raising}, while probabilistic abstract interpretation reasons about quantitative distribution properties using abstractions such as intervals and moments \cite{ngo2018bounded,cousot2012probabilistic}.
However, these approaches target a single objective rather than multiple ones. %

Programmatic strategy synthesis has been approached via knowledge compilation, albeit for finite-horizon objectives~\cite{cho2025scaling}.

\paragraph{Multiobjective model checking}
Markov decision processes (MDPs) are widely used for planning and decision-making.
This in particular includes learning-based approaches to programmatic strategy synthesis~\cite{dweep2021learning,rubens2023choosing}, which typically do not provide formal guarantees. 
In formal multiobjective model checking, MDPs have been studied for discounted vector-valued rewards~\cite{white1982multi,henig1983vector}, mean-payoff objectives, and $\omega$-regular specifications~\cite{chatterjee2006mdp,chatterjee2007,etessami2008multi,forejt2011quantitative}.
These approaches largely assume finite-state models.

Notable exceptions include the work by \citet{etessami2020qualitative}, who provide verification and synthesis for qualitative objectives, and \citet{main2026mixing}, who study achievable sets for payoffs in partially observable MDPs, but without considering synthesis.
At program level, \citet{watanabe2026posterior} verify safety properties of posterior distributions using weighted-sum optimization to refine convex overapproximations of achievable termination probabilities.
For indicator expectations, they establish correspondences analogous to \Cref{theo:mop-as-weighted-wp,theo:wp-as-weighted-mop}.
Our work generalizes from indicators to general multiobjective expectations and supports refutation---proving that a point is \emph{not} achievable---which they identify as a limitation of their approach.
Further, synthesizing a witnessing strategy is left open there.

\section{Conclusion}
\label{sec:conclusion}

In conclusion, we have developed a semantic foundation for multiobjective reasoning about probabilistic programs through the \mopsymbol transformer, which extends weakest preexpectations to the multiobjective setting and supports invariant-based reasoning for loops.
Further, we presented first steps toward multiobjective programmatic strategy synthesis via a weighted-sum-based algorithm.

\paragraph{Limitations and future work}
While the \mopsymbol transformer captures all almost achievable points, our synthesis approach is inherently incomplete: even if a Pareto optimal determinization exists, it may not be recoverable through the proposed weighted-sum approach, e.g., if the target is not exposed.
Efficiently characterizing when Pareto optimal points are \emph{precisely} achievable is an open problem.
Our language can be extended with richer features such as procedures, recursion, heap manipulation, and continuous distributions, while our expectation-based objectives do not yet capture discounted or long-run average rewards.
Further directions for future work include objectives combining maximization and minimization, and multi-player settings. %

\bibliographystyle{ACM-Reference-Format}
\bibliography{literature}

\clearpage
\appendix
\section{Weakest Preexpectations}
\label{app:wp}

We briefly recall the weakest preexpectation semantics used throughout the paper and in particular in \Cref{sec:problem-mop}. 
We begin with the notion of an \emph{expectation}.

\begin{definition}[Expectations]
    The complete lattice of \emph{expectations} is $(\Exp, \sqsubseteq)$, where
    \begin{itemize}
        \item $\Exp = \States \to \PosRealsInf$ is the set of expectations, \qand
        \item $\sqsubseteq$ is the pointwise lifted order on $\PosRealsInf$, i.e.
        $
            f \ssqsubseteq f' \, \iff \, \forall\, \pstate \in \States \colon \, f(\pstate) \lleq f'(\pstate)~.
        $
    \end{itemize}
\end{definition}

Arithmetic operations and suprema on expectations are likewise understood pointwise.
We set $0 \cdot \infty = 0$ and $a \cdot \infty = \infty$ for $a > 0$.
As usual, we identify a constant $a \in \mathbb{R}_{\geq 0}^{\infty}$ with the constant expectation $\mylambda{\pstate} a$.
Predicates are identified with their indicator expectations; for a predicate $\varphi$, we write $\iverson{\varphi}$ for the expectation that evaluates to $1$ in states satisfying $\varphi$ and to $0$ otherwise.
Substitution is defined as
$
    \post\subst{x}{\ee} \eeq \mylambda \pstate \post\bigl( \substState{\pstate}{x}{\ee} \bigr).
$

For a probabilistic program $C$ and a postexpectation $f \in \Exp$, the \emph{weakest preexpectation} $\wp{C}{f}$ maps each initial state to the optimal expected value of $f$ after executing $C$.
The transformer is defined compositionally over the syntax of pGCL.

\begin{definition}[Weakest Preexpectations]
\label{def:wp}
    The (angelic) weakest preexpectation transformer
    $
        \wpsymbol\colon \Exp \to \Exp
    $
    is defined according to the rules in \Cref{tab:wp}.
\end{definition}

\begin{table}[t]
\renewcommand{\arraystretch}{1.5}
\begin{tabular}{@{\hspace{.5em}}l@{\hspace{2em}}l@{\hspace{.5em}}}
		\toprule\toprule
		$\boldsymbol{\program}$			& $\wp{\program}{\post}$ \\
		\midrule
		$\SKIP$				& $\post$ \\
        $\ASSIGN{x}{\ee}$ & $\post \subst{x}{\ee}$ \\
        $\COMPOSE{\program_1}{\program_2}$ & $\wp{\program_1}{\wp{\program_2}{\post}}$ \\
        $\GC{\guard_1}{\program_1}{\guard_2}{\program_2}$ & $[\guard_1] \cdot \wp{\program_1}{\post} \ssqcup [\guard_2] \cdot \wp{\program_2}{\post}$\\
        $\PCHOICE{\program_1}{\ps}{\program_2}$ & $\ps \cdot \wp{\program_1}{\post} \pplus (1-\ps) \cdot \wp{\program_2}{\post}$ \\
        $\WHILEDO{\guard}{\program'}$ & $\lfp X \mydot \quad [\neg \guard] \cdot \post \pplus [\guard] \cdot \wp{\program'}{X}$ \\
	\bottomrule\bottomrule
\end{tabular}%
\vspace{1em}
\caption{Rules for the $\wpsymbol$ transformer. %
        $\lfp \mylambda{X} \Phi(X)$ denotes the least fixpoint of $\Phi$. } %
\label{tab:wp}
\end{table}

Assignments act by substitution, sequential composition by transformer composition, probabilistic choice by the corresponding convex combination, and nondeterministic
choice by the pointwise supremum of the branch preexpectations.
Since we consider nondeterminism angelically, nondeterministic choices are resolved so as to maximize the expected value, dual to the more common demonic variant of weakest-preexpectation semantics.
Loops are interpreted as least fixpoints of their characteristic functions.

Thus, as in the multiobjective transformer introduced in the main paper, weakest preexpectation semantics proceeds backwards from a postexpectation and computes the values that can be guaranteed---under the chosen resolution of nondeterminism---from each initial state.

Of particular importance to us is the relationship between weakest preexpectations and program determinizations.
Classical operational characterizations of weakest preexpectations quantify over general, potentially history-dependent schedulers (e.g.\ \cite[Theorem 2.15]{batz2024automated} or \cite[Theorem 4.11]{kaminski2019advanced}).
In contrast, program determinizations relate to \emph{memoryless} schedulers (see \Cref{theo:sched-leq-determ}).
The following result shows that this restricted class already suffices to characterize weakest preexpectations: the supremum over pure determinizations coincides with the ordinary weakest preexpectation.

\wpsupdeterm*

\begin{proof}
\label{proof:wp-sup-determ}
    By \Cref{theo:mop-wp-pure},
    \[
        \mop{\program}{\dwcset{\post}}(\pstate)
        \eeq
        \cl{\convsetcomp{\wp{\program'}{\post}(\pstate)}{\program' \determ \program}}.
    \]
    In the one-dimensional setting, taking the convex hull and Scott closure reduces to taking the downward closure of the supremum, i.e.,
    \[
        \mop{\program}{\dwcset{\post}}(\pstate)
        \eeq
        \dwc{\bigsqcup\ \setcomp{\wp{\program'}{\post}(\pstate)}{\program' \determ \program}}.
    \]
    Combining this with \Cref{theo:wp-as-weighted-mop} instantiated with the weight $\weight = 1$ gives
    \[
        \wp{\program}{\post}(\pstate)
        \eeq
        \max \mop{\program}{\dwcset{\post}}(\pstate)
        \eeq
        \bigsqcup\ \setcomp{\wp{\program'}{\post}(\pstate)}{\program' \determ \program}.
    \]
\end{proof}

\section{Infinite-State Multiobjective Markov Decision Processes}
\label{sec:mdp}

We establish the operational correspondence underlying \Cref{theo:mop-is-pareto}.
First, \Cref{sec:mdp} shows that memoryless schedulers suffice for almost-achievable rewards (\Cref{theo:memless-suff}) and characterizes these rewards as the least fixpoint of a multiobjective Bellman operator (\Cref{theo:bellman-is-pareto}).
\Cref{sec:mop-mdp} then relates this operator to the program-level \mopsymbol transformer (\Cref{theo:sound-mop}).
Finally, \Cref{theo:sched-leq-determ,theo:mop-includes-ach} connect operational rewards to program determinizations, yielding the desired correspondence in \Cref{theo:mop-is-pareto}.

We consider finite-action MDPs~\cite{puterman2014markov} with possibly uncountable state space as operational semantics for probabilistic programs.
Where possible, we follow the notation from \cite{batz2024jp,baier2008principles}.
Toward the analysis of multiobjective MDPs, we introduce multidimensional reward structures.

Formally, an MDP is a quadruple $\MDP = \left( \OpStates, \MI, \OpAct, \OpP \right)$ where
$\OpStates$ is a (possibly uncountable) non-empty set of states,
$\MI \subseteq \OpStates$ is a \emph{set} of initial states,
$\OpAct$ is a finite non-empty set of action labels,
and $\OpP \colon \OpStates \times \OpAct \times \OpStates \rightarrow [0,1]$ is a transition probability function such that for all $\opState \in \OpStates$ and all $\opAct\in\OpAct$, $\sum_{\opState'\in\OpStates} \OpP(\opState,\opAct,\opState')\in\{0,1\}$
and the support $\set{\opState' \in \OpStates \mid \OpP(\opState,\opAct,\opState') > 0} = \OpSucc(\opState,\opAct)$ is finite.
We define $\OpSucc(\opState) = \bigcup_{\opAct \in \OpAct} \OpSucc(\opState,\opAct)$.
For $\opState \in \OpStates$ we write $\OpAct(\opState) = \{\opAct \in \OpAct \mid \sum_{\opState'\in\OpStates} \OpP(\opState,\opAct,\opState') = 1\}$ and require that $|\OpAct(\opState)|  \geq 1$ for every $\opState$.
Although we allow uncountably many initial states, we assume that only countably many states are reachable from any fixed initial state.
An MDP $\MDP$ is a Markov chain if $|\OpAct(\opState)| = 1$ for all $\opState \in \OpStates$.

To resolve nondeterminism in MDPs, we define \emph{schedulers}.
\label{def:scheduler}%
    A (deterministic) \emph{scheduler} for $\MDP = \left( \OpStates, \MI, \OpAct, \OpP \right)$ is a function $\msched \colon \OpStates^{+} \to \OpAct$ satisfying $\msched(\opState_0\ldots \opState_n) \in \OpAct(\opState_n)$ for all $\opState_0\ldots \opState_n \in \OpStates^{+}$.
    A scheduler $\msched$ is called \emph{memoryless} if $\msched(\opState_0\ldots \opState_n)$ depends only on $\opState_n$. %
    A memoryless scheduler can be identified with a map of type $\OpStates \to \OpAct$. %
    We denote the set of deterministic schedulers by $\SchedPure{\MDP}$.
    A memoryless scheduler $\msched \colon \OpStates \to \OpAct$ for $\MDP$ induces a Markov chain $\MDP^\msched = \left( \OpStates, \MI, \OpAct, \OpP' \right)$ where for $\opState,\opState' \in \OpStates$ and $\opAct \in \OpAct$,
    \[
    \OpP'(\opState,\opAct,\opState') \eeq 
    \begin{cases}
        \OpP(\opState,\opAct,\opState') &\text{if $a = \msched(\opState)$, and}\\
        0 &\text{otherwise.}
    \end{cases}~
    \]

    A \emph{mixed scheduler} %
    $\sched \in \distr{\SchedPure{\MDP}}$ is a finitely supported distribution over schedulers.
    Let $\supp{\determDist}$ denote the support of a distribution $\determDist \in \distr{\Omega}$.
    We call a mixed scheduler $\sched$ pure if $|\supp{\sched}|=1$, i.e., if it is a (deterministic) scheduler.
    We call a mixed scheduler memoryless if all schedulers in its support are memoryless. 
    We denote by $\Sched{\MDP}$ the set of mixed schedulers and by $\SchedMemless{\MDP}$ the set of mixed memoryless schedulers. %

We call an MDP $\MDP'$ a \emph{determinization} of $\MDP$ if there exists a memoryless scheduler $\msched$ such that $\MDP' = \MDP^\msched$.

Following \citet{batz2024programmatic}, we consider \emph{reachability-reward} objectives.
Each state in a target set $\MT$ is a sink and has a $\PosRealsInfVect$-valued reward that is collected when that state is reached for the first time.
The use of reward vectors supports the multiobjective reasoning developed below.

For $\MT\subseteq \OpStates$, the set of finite paths eventually reaching $\MT$ is
\[
   \pathstotarget{\MT} \eeq \{\, \opState_0 \ldots \opState_m \in \OpStates^{+}~\mid~  \opState_m \in \MT,\, \forall k \in\{0,\ldots,m-1\}\colon \opState_k\not\in \MT  \,\} ~.
\]
For a pure scheduler $\sched \in \SchedPure{\MDP}$, the probability of a path $\finPath = \opState_0 \dots \opState_k \in \OpStates^+$ in MDP $\MDP$ is 
\[
    \pathProb{\MDP}{\sched}(\finPath) = \prod_{0 \leq i < k} \OpP(\opState_{i},\sched(\opState_0 \dots \opState_i),\opState_{i+1}).
\]
We say $\finPath$ is \emph{possible under $\sched$} if $\pathProb{\MDP}{\sched}(\finPath) > 0$ and denote the set of possible paths as $\finPaths{\sched}$.
For a mixed scheduler $\sched \in \Sched{\MDP}$, the probability of a path $\finPath \in \OpStates^+$ is %
\[
    \pathProb{\MDP}{\sched}(\finPath) = \sum_{\sched' \in \supp{\sched}}  \sched(\sched') \cdot \pathProb{\MDP}{\sched'}(\finPath). %
\]

Given $\OpRew \colon \MT \to \PosRealsInfVect$, we define the function $\exprew{\msched}{\MDP}{\OpRew} \colon \MI \to \PosRealsInfVect$, which maps an initial state $\mi$ to its expected (reachability-)reward $\ExpRewN{\MDP}{\OpRew}{\mi}{\sched}$ after at most $k$ steps under scheduler $\msched$ by
\begin{align}
    \ExpRewN{\MDP}{\OpRew}{\mi}{\sched} \eeq
    \sum_{0 \leq i \leq k} \sum_{\substack{\opState_0\ldots\opState_i\in \pathstotarget{\MT} \\ \opState_0 = \mi}} \pathProb{\MDP}{\msched}(\mi\opState_1\ldots\opState_i) \cdot \OpRew(\opState_i)~.
    \label{def:etr-bound}
\end{align}
Finally, we define for all states $\mi \in \MI$
\begin{align}
   \ExpRew{\MDP}{\OpRew}{\mi}{\sched} \eeq& \sup_{k \in \Nats} \ExpRewN{\MDP}{\OpRew}{\mi}{\sched}. %
    \label{def:etr}
\end{align}

\paragraph{Notation}
For the remainder of this section, we let $\OpRew \colon \MT \to \PosRealsInfVect$ be a reward assignment for an MDP $\MDP = \left( \OpStates, \MI, \OpAct, \OpP \right)$.

\subsection{Auxiliary Results}

For the proof-heavy sections that follow, we require some basic helper theorems, which we collect in this section.

\begin{lemma}%
\label{theo:ach-dwc-minkowski}
    Let $J$ be a finite index set, $X_j \subseteq \PosRealsInfVect$ and $p_j \in [0,1]$.
    Then
    \[
        \dwc{ \Minkow_j p_j \cdot X_j }
        \eeq
        \Minkow_j p_j \cdot \dwc{X_j}.
    \]
\end{lemma}
\begin{proof}
    We prove both inclusions.
    For $\supseteq$, let
    $
        y \in \Minkow_j p_j \cdot \dwc{X_j}.
    $
    Then there exist $x_j \in X_j$ and $y_j \leq p_j x_j$ such that $y = \sum_j y_j.$
    Hence
    $
        y \leq \sum_j p_j x_j,
    $
    which implies
    $
        y \in \dwc{\Minkow_j p_j \cdot X_j}.
    $

    For $\subseteq$, let $y \in \dwc{\Minkow_j p_j \cdot X_j}.$
    Then there exist $x_j \in X_j$ such that $y \leq \sum_j p_j x_j.$
    Since all vectors are non-negative, the inequality can be decomposed componentwise: there exist vectors $y_j$ such that $y = \sum_j y_j$ and $y_j \leq p_j x_j$ for all $j$.
    Thus $y_j \in \dwc{p_j\cdot X_j}$ for every $j$, and hence $y \in \Minkow_j p_j\cdot\dwc{X_j}.$

    Therefore, both inclusions hold and the sets are equal.
\end{proof}

\begin{lemma}%
\label{theo:closure-minkowski}
    Let $J$ be a finite index set.
    Let $X_j \subseteq \PosRealsInfVect$ and let $p_j\in[0,1]$ for every $j\in J$.
    Then
    \[
        \dwc{\conv{\Minkow_{j\in J}p_j\cdot X_j}}
        \eeq
        \dwc{\conv{\Minkow_{j\in J}p_j\cdot\dwc{\conv{X_j}}}}.
    \]
\end{lemma}
\begin{proof}
    We prove both inclusions.

    For $\subseteq$, let
    $y \in \dwc{\conv{\Minkow_{j\in J} p_j\cdot X_j}}.$
    Then there exist $k\in\Nats$, weights $\lambda_1,\ldots,\lambda_k\geq 0$ with $\sum_{\ell=1}^k\lambda_\ell=1$, and elements $x_{j,\ell}\in X_j$ such that
    \[
        y \leq
        \sum_{\ell=1}^k \lambda_\ell \sum_{j\in J}p_jx_{j,\ell}
        \eeq
        \sum_{j\in J}p_j \underbrace{ \left(\sum_{\ell=1}^k \lambda_\ell x_{j,\ell} \right)}_{\in\conv{X_j}}.
    \]
    For every $j\in J$, define
    \[
        z_j \coloneqq \sum_{\ell=1}^k \lambda_\ell x_{j,\ell}.
    \]
    Then $z_j\in\conv{X_j}$ for every $j\in J$.
    Hence
    \[
        y\in \dwc{\Minkow_{j\in J} p_j\cdot\conv{X_j}}.
    \]
    By \Cref{theo:ach-dwc-minkowski}, this implies
    \[
        y \in \Minkow_{j\in J} p_j\cdot\dwc{\conv{X_j}}.
    \]
    Since $\dwc{\conv{X_j}}$ is convex for every $j\in J$, the Minkowski sum on the right-hand side is convex.
    Further, the set remains downward closed.
    Therefore,
    $
        y \in \dwc{\conv{\Minkow_{j\in J}p_j\cdot\dwc{\conv{X_j}}}}.
    $

    For $\supseteq$, let
    $
        y\in\dwc{ \conv{\Minkow_{j\in J}p_j\cdot\dwc{\conv{X_j}}}}.
    $
    Then there exist $n\in\Nats$, weights
    $\lambda_1,\ldots,\lambda_n\geq 0$ with
    $\sum_{\ell=1}^n\lambda_\ell=1$, and
    vectors $z_{j,\ell}\in\dwc{\conv{X_j}}$ such that
    \[
        y \leq\sum_{\ell=1}^n \lambda_\ell \sum_{j\in J} p_jz_{j,\ell}
        \eeq \sum_{j\in J} p_j \left( \sum_{\ell=1}^n \lambda_\ell z_{j,\ell} \right).
    \]
    Since $\dwc{\conv{X_j}}$ is convex, we have
    \[
        z_j \coloneqq \sum_{\ell=1}^n \lambda_\ell z_{j,\ell} \in \dwc{\conv{X_j}}
    \]
    for every $j\in J$.
    Thus
    \[
        y\in \dwc{ \Minkow_{j\in J} p_j\cdot\dwc{\conv{X_j}}}.
    \]
    By \Cref{theo:ach-dwc-minkowski}, this is equivalent to
    \[
        y\in \dwc{ \Minkow_{j\in J} p_j\cdot\conv{X_j} }.
    \]
    Since
    \[
        \Minkow_{j\in J} p_j\cdot\conv{X_j}
        \eeq
        \conv{ \Minkow_{j\in J} p_j\cdot X_j },
    \]
    it follows that
    \[
        y\in \dwc{\conv{ \Minkow_{j\in J} p_j\cdot X_j }}.
    \]
    Hence, both inclusions hold, proving the equality.
\end{proof}

\begin{lemma}%
\label{theo:closure-union}
    Let $I$ be an index set and let $X_i \subseteq \PosRealsInfVect$ for every $i\in I$.
    Then
    \[
        \dwc{\conv{ \bigcup_{i\in I} X_i}}
        \eeq
        \dwc{\conv{\bigcup_{i\in I}\dwc{\conv{X_i}}}}.
    \]
\end{lemma}
\begin{proof}
    $\subseteq$ immediately follows from $X_{i}\subseteq\dwc{\conv{X_{i}}}$ and monotonicity of union, $\dwc{\cdot}$, and $\conv{\cdot}$.
    For $\supseteq$, let
    $
        y \in \dwc{\conv{\bigcup_{i\in I}\dwc{\conv{X_i}}}}.
    $
    Then there exist $n\in\Nats$, weights $\lambda_1,\ldots,\lambda_n\geq0$ with $\sum_{\ell=1}^n\lambda_\ell=1$, and $z_\ell\in\dwc{\conv{X_{i_\ell}}}$ such that
    $
        y \leq \sum_{\ell=1}^n\lambda_\ell z_\ell.
    $
    For every $\ell$, there exists $z_\ell'\in\conv{X_{i_\ell}}$ such that $z_\ell\leq z_\ell'$.
    Hence
    $
        y \leq \sum_{\ell=1}^n\lambda_\ell z_\ell'
    $
    and
    $
        \sum_{\ell=1}^n\lambda_\ell z_\ell' \in \conv{ \bigcup_{i\in I}X_i }.
    $
    Therefore
    $
        y\in \dwc{ \conv{ \bigcup_{i\in I}X_i } }.
    $
\end{proof}

\begin{restatable}{lemma}{clfinclinf}
    \label{theo:clfin-clinf}
       \cl{\AchFin{\MDP}{\OpRew}{\opState}} \quad = \quad  \cl{\Ach{\MDP}{\OpRew}{\opState}}
\end{restatable}
\begin{proof}
    For $\subseteq$, note that we have
    \[
        \AchFin{\MDP}{\OpRew}{\opState} \quad \subseteq \quad  \Ach{\MDP}{\OpRew}{\opState},
    \]
    because
    \begin{align*}
        &\ {\AchFin{\MDP}{\OpRew}{\opState}} \\
        =&\ {\bigcup_{k \in \Nats} \dwc{\set{ \ExpRewN{\MDP}{\OpRew}{\opState}{\sched} \mid \sched \in \Sched{\MDP}}}} \tag{\Cref{def:achievable}}\\
        =&\ {\dwc{\set{ \ExpRewN{\MDP}{\OpRew}{\opState}{\sched} \mid \sched \in \Sched{\MDP}, k \in \Nats}}} \\
        \overset{*}{\subseteq=}&\ {\dwc{\set{ \sup_{k \in \Nats} \ExpRewN{\MDP}{\OpRew}{\opState}{\sched} \mid \sched \in \Sched{\MDP}}}} \\
        =&\ {\dwc{\set{ \ExpRew{\MDP}{\OpRew}{\opState}{\sched} \mid \sched \in \Sched{\MDP}}}} \\
        =&\ {\Ach{\MDP}{\OpRew}{\opState}}. \tag{\Cref{def:achievable}}
    \end{align*}
    Step (*) holds because for all $k\in \Nats$, we have
    \[
        \ExpRewN{\MDP}{\OpRew}{\opState}{\sched} \leq \sup_{k \in \Nats} \ExpRewN{\MDP}{\OpRew}{\opState}{\sched},
    \]
    which implies
    \[
        \ExpRewN{\MDP}{\OpRew}{\opState}{\sched} \in \dwc{\set{\sup_{k \in \Nats} \ExpRewN{\MDP}{\OpRew}{\opState}{\sched}}}.
    \]

    Monotonicity of $\cl{\cdot}$ then immediately gives
    \[
       \cl{\AchFin{\MDP}{\OpRew}{\opState}} \quad \subseteq \quad  \cl{\Ach{\MDP}{\OpRew}{\opState}}.
    \]

    For $\supseteq$, consider:
    \begin{align*}
        &\ \cl{\AchFin{\MDP}{\OpRew}{\opState}} \\
        =&\ \cl{\bigcup_{k \in \Nats} \dwc{\set{ \ExpRewN{\MDP}{\OpRew}{\opState}{\sched} \mid \sched \in \Sched{\MDP}}}} \\
        =&\ \cl{\dwc{\set{ \ExpRewN{\MDP}{\OpRew}{\opState}{\sched} \mid \sched \in \Sched{\MDP}, k \in \Nats}}} \\
        \overset{*}{\supseteq}&\ \cl{\dwc{\set{ \sup_{k \in \Nats} \ExpRewN{\MDP}{\OpRew}{\opState}{\sched} \mid \sched \in \Sched{\MDP}}}} \\
        =&\ \cl{\dwc{\set{ \ExpRew{\MDP}{\OpRew}{\opState}{\sched} \mid \sched \in \Sched{\MDP}}}} \\
        =&\ \cl{\Ach{\MDP}{\OpRew}{\opState}} 
    \end{align*}
    Except (*), this is by definition. 
    Thus, we now show that $\cl{A} \supseteq \cl{B}$, where 
    \[
        A \coloneqq \dwc{\set{ \ExpRewN{\MDP}{\OpRew}{\opState}{\sched} \mid \sched \in \Sched{\MDP}, k \in \Nats}}
    \]
    and 
    \[
        B \coloneqq \dwc{\set{ \sup_{k \in \Nats} \ExpRewN{\MDP}{\OpRew}{\opState}{\sched} \mid \sched \in \Sched{\MDP}}}.
    \]
    
    Note that since $B$ is downward closed, the elements in $\cl{B}$ either are (1) already contained in $B$ itself, or are (2) the supremum of a directed subset $D \subseteq \cl{B}$.
    
    For (1), consider the following:
    For all $\sched \in \Sched{\MDP}$, we have that $\sup_{k \in \Nats} \ExpRewN{\MDP}{\OpRew}{\opState}{\sched}$ is the supremum of the directed set 
    $\set{ \ExpRewN{\MDP}{\OpRew}{\opState}{\sched} \mid k \in \Nats} \subseteq A$.
    So, $\sup_{k \in \Nats} \ExpRewN{\MDP}{\OpRew}{\opState}{\sched} \in \cl{A}$, as the Scott closure of a set by definition is closed under directed suprema.
    Since $\cl{A}$ is downward closed, we get $\cl{A} \supseteq B$.
    
    For (2), let $x \in \cl{B}$ with $x = \sup D$ for some directed set $D \subseteq B$.
    From above, we thus know that $D \subseteq \cl{A}$.
    Since $\cl{A}$ is Scott closed, we get $\sup D \in \cl{A}$, so $x \in \cl{A}$.
    This concludes the proof of $\cl{B} \subseteq \cl{A}$.

\end{proof}

\subsection{Achievable Rewards in MDPs}
\label{sec:bellman-pareto}

Following \Cref{sec:problem-mop}, we develop multiobjective reasoning on MDPs, building on seminal work by~\cite{white1982multi,henig1983vector,chatterjee2006mdp}.
Contrary to that work, we study infinite-state, finite-action MDPs.
The main result in this section concerns achievable rewards, the analogue of multiobjective preexpectations.

\begin{definition}[Achievable Rewards]
\label{def:achievable}
    We define the following:
    \begin{enumerate}
        \item  The \emph{set of achievable rewards from $\mi \in \MI$ after at most $k \in \Nats$ steps} is defined as
        \[
            \AchN{\MDP}{\OpRew}{\mi} = \dwcset{ \ExpRewN{\MDP}{\OpRew}{\mi}{\sched} \mid \sched \in \Sched{\MDP}} \subseteq \PosRealsInfVect.
        \]

        \item  The \emph{set of achievable rewards from $\mi \in \MI$} is defined as
        \[
            \Ach{\MDP}{\OpRew}{\mi} = \dwcset{ \ExpRew{\MDP}{\OpRew}{\mi}{\sched} \mid \sched \in \Sched{\MDP}} \subseteq \PosRealsInfVect.
        \]

        \item  The \emph{set of achievable rewards by memoryless schedulers from $\mi \in \MI$} is defined as
        \[
            \AchMemless{\MDP}{\OpRew}{\mi} = \dwcset{ \ExpRew{\MDP}{\OpRew}{\mi}{\sched} \mid \sched \in \SchedMemless{\MDP}} \subseteq \PosRealsInfVect.
        \]
    \end{enumerate}
    If $\MDP$, $\OpRew$, or $\mi$ are clear from the context, it might be omitted.
\end{definition}
The Pareto front is indeed the boundary of the closure of the achievable rewards, but we always operate on the achievable rewards. %

We demonstrate that memoryless schedulers suffice to characterize the achievable points.
Intuitively, for all points achievable by a memoryful scheduler, we can find a set of memoryless schedulers whose expected reward converges to the reward of the memoryful scheduler, see \Cref{theo:memless-suff}.
In order to prove this, we first show a similar result for \emph{finite}-state MDPs.

\begin{theorem}
\label{theo:memless-suff-fin}
    For a finite-state MDP $\MDP$, a reward function $\OpRew\colon \MT \to \PosRealsInfVect$ and $\mi \in \MI$, for all $\sched \in \Sched{\MDP}$ there exists a \emph{memoryless} scheduler $\sched' \in \SchedMemless{\MDP}$ with 
    \[
        \ExpRew{\MDP}{\OpRew}{\mi}{\sched'} \ggeq \ExpRew{\MDP}{\OpRew}{\mi}{\sched}.
    \]
\end{theorem}
\begin{proof}
    This is well-known for several settings, e.g., reachability probabilities (\cite[Theorem 3.2]{etessami2008multi}), single objectives (\cite[Chapter 10]{baier2008principles}) or bounded rewards (\cite[Proposition 6]{forejt2011quantitative}).
    Note that for general reachability-rewards, we \emph{do} require memory, but not if the target states are absorbing, which they are in our setting.
    
    Since none of the above cited results \emph{precisely} covers our setting, we prove this here.
    We do so in two steps, first proving the sufficiency of memoryless schedulers for reachability probabilities and second by showing that this result transfers to reachability rewards.
    
    For each scheduler $\sched$, define $v^{\sched} \in [0,1]^{\vert \MT \vert}$ by
    \begin{align}
        v_t^{\sched}
        \eeq
        \sup_{k \in \Nats} \sum_{0 \leq i \leq k} \sum_{\substack{\opState_0\ldots\opState_i\in \pathstotarget{\MT} \\ \opState_0 = \mi \\ \opState_i = t}} \pathProb{\MDP}{\sched}(\opState_0\ldots\opState_i)
    \end{align}
    for all $t \in \MT$.
    Thus, $v_t^{\sched}$ is the probability of eventually reaching target $t$ under $\sched$.

    Since $\MDP$ is finite, there are only finitely many pure, memoryless schedulers, say $\hat{\sched}_1,\dots,\hat{\sched}_m$.
    Consider the convex hull of their reachability vectors,
    \[
        C \coloneqq \conv{ \set{ v^{\hat{\sched}_1}, \dots,v^{\hat{\sched}_m}}}.
    \]
    Since $C$ is the convex hull of finitely many points, it is a convex polytope and, in particular, compact.

    We now show that for every mixed, memoryful scheduler $\sched$, there exists some $c \in C$ such that $c \geq v^{\sched}$ componentwise.

    To this end, assume towards a contradiction that no such $c$ exists.
    Equivalently,
    \[
        C \cap \underbrace{ \setcomp{ x \in \Reals^{\vert \MT \vert}}{ x \geq v^{\sched}}}_{\coloneqq\ C'}
        \eeq
        \emptyset.
    \]

    The sets $C$ and $C'$ are convex and topologically closed, and $C$ is additionally compact.
    Hence, by the separating hyperplane theorem \cite[Part III]{rockafellar1970convex}, there exists a nonzero vector $\weight \in \Reals^{\vert \MT \vert}$ that strictly separates $C$ from $C'$.

    Moreover, we may choose $\weight \in \PosReals^{\vert \MT \vert}$, since $C'$ is upward closed: %
    suppose that $\weight_t < 0$ for some $t \in \MT$.
    Starting from any $x \in C'$, we could then increase only the $t$-th component of $x$ arbitrarily.
    Since $\weight_t < 0$, the value $\weight \cdot x$ would thereby become arbitrarily small, which is incompatible with a hyperplane separating the compact set $C$ from the whole upward-closed set $C'$.
    Thus, all components of $\weight$ can be chosen nonnegative.

    Since $v^{\sched} \in C'$, strict separation therefore yields
    \begin{align}
        \sup_{c \in C} \weight \cdot c < \inf_{c' \in C'} \weight \cdot c' \eeq \weight \cdot v^{\sched}.
    \label{eq:assump}
    \end{align}

    We now interpret the separating vector $\weight$ as a scalar reward function.
    Define
    $
        \OpRew_{\weight}(t) \eeq \weight_t
    $
    for all $t \in \MT$.
    This is a single-objective reachability-reward function.

    By definition, $\pathstotarget{\MT}$ contains precisely those finite paths that enter $\MT$ for the first time in their final state.
    Hence, the sets
    \[
        \setcomp{\opState_0\ldots\opState_i \in \pathstotarget{\MT}}{\opState_i = t}, \qquad t \in \MT,
    \]
    partition $\pathstotarget{\MT}$ according to the first target that is reached.
    Consequently, for every scheduler $\sched'$,
    \begin{align*}
        \ExpRew{\MDP}{\OpRew_{\weight}}{\mi}{\sched'}
        &\eeq \sup_{k \in \Nats} \sum_{0 \leq i \leq k} \sum_{\substack{\opState_0\ldots\opState_i\in \pathstotarget{\MT} \\ \opState_0 = \mi }} \pathProb{\MDP}{\sched'}(\opState_0\ldots\opState_i) \cdot \OpRew_{\weight}(\opState_i) \\
        &\eeq \sum_{t \in \MT} \sup_{k \in \Nats} \sum_{0 \leq i \leq k} \sum_{\substack{\opState_0\ldots\opState_i\in \pathstotarget{\MT} \\ \opState_0 = \mi \\ \opState_i = t}} \pathProb{\MDP}{\sched'}(\opState_0\ldots\opState_i) \cdot \OpRew_{\weight}(\opState_i) \\
        &\eeq \sum_{t \in \MT} v_t^{\sched'} \cdot \weight_t \\
        &\eeq v^{\sched'} \cdot \weight.
    \end{align*}
    Thus, maximizing $\weight \cdot v^{\sched'}$ is exactly the single-objective reachability-reward problem induced by $\OpRew_{\weight}$.

    For single-objective reachability rewards, pure, memoryless schedulers suffice, see e.g.\ \cite[Chapter 10]{baier2008principles}.
    Hence, there exists a pure, memoryless scheduler $\hat{\sched}_\ast$ such that
    \[
        v^{\hat{\sched}_\ast} \cdot \weight \ggeq v^{\sched} \cdot \weight.
    \]

    Since $\hat{\sched}_\ast$ is pure and memoryless, its reachability vector $v^{\hat{\sched}_\ast}$ is one of the points whose convex hull defines $C$, i.e., $v^{\hat{\sched}_\ast} \in C.$
    This implies that  %
    \[
        \sup_{c \in C} \weight \cdot c \ggeq \weight \cdot v^{\hat{\sched}_\ast}.    
    \]

    Combining the two inequalities gives
    \[
        \sup_{c \in C} \weight \cdot c
        \ggeq
        \weight \cdot v^{\hat{\sched}_\ast}
        \ggeq
        \weight \cdot v^{\sched},
    \]
    which contradicts \Cref{eq:assump}.

    Therefore, there exists some $c \in C$ such that $c \geq v^{\sched}.$
    Since $c \in C$, there exist coefficients $\weightc_1,\dots,\weightc_m \geq 0$ with $\sum_{j=1}^m \weightc_j = 1$ such that
    \[
        c \eeq \sum_{j=1}^m \weightc_j \cdot v^{\hat{\sched}_j}.
    \]

    Define the mixed, memoryless scheduler $\sched'$ by choosing $\hat{\sched}_j$ with probability $\weightc_j$.
    By definition of mixed schedulers, for every $t \in \MT$,
    \begin{align*}
        v_t^{\sched'}
        \eeq \sum_{j=1}^m \weightc_j \cdot v_t^{\hat{\sched}_j}
        \eeq c_t
        \ggeq v_t^{\sched}.
    \end{align*}
    Thus, $\sched'$ reaches every target state with probability at least as large as $\sched$.

    It remains to transfer the results for the target reachability probabilities to the original reward function $\OpRew$.
    Since all target states are absorbing, for every component $\ell$ of the reward and every scheduler $\varsigma$, for the same reasons as above we have
    \[
        {\left( \ExpRew{\MDP}{\OpRew}{\mi}{\varsigma} \right)}_{\ell} \eeq \sum_{t \in \MT} v_t^{\varsigma} \cdot \OpRew(t)_{\ell}.
    \]

    Since all entries of $\OpRew$ are nonnegative and $v_t^{\sched'} \geq v_t^{\sched}$ for every $t \in \MT$, we obtain componentwise
    \[
        \ExpRew{\MDP}{\OpRew}{\mi}{\sched'} \ggeq  \ExpRew{\MDP}{\OpRew}{\mi}{\sched}.
    \]

\end{proof}

Now, we are in position to prove the main (and last) theorem of this section:
for infinite state MDPs, memoryless schedulers suffice to characterize the closure of the achievable points.

\begin{restatable}{theorem}{memlesssuff}
\label{theo:memless-suff}
    For all $\mi \in \MI$, we have that
    \[
    \cl{\Ach{\MDP}{\OpRew}{\mi} } \eeq \cl{\AchMemless{\MDP}{\OpRew}{\mi} }.
    \]
\end{restatable}

\begin{proof}
Fix an arbitrary initial state $\mi\in\MI$.
By definition, we have
\[
    \Ach{\MDP}{\OpRew}{\mi} 
    \coloneqq
    \dwcset{
        \ExpRew{\MDP}{\OpRew}{\mi}{\sched}
        \mid
        \sched\in\Sched{\MDP}
    }
\]
and
\[
    \AchMemless{\MDP}{\OpRew}{\mi} 
    \coloneqq
    \dwcset{
        \ExpRew{\MDP}{\OpRew}{\mi}{\sched}
        \mid
        \sched\in\SchedMemless{\MDP}
    }~.
\]
Since every memoryless scheduler is a scheduler,
$\AchMemless{\MDP}{\OpRew}{\mi} \subseteq \Ach{\MDP}{\OpRew}{\mi} $.
Monotonicity of the Scott closure therefore yields
\[
    \cl{\AchMemless{\MDP}{\OpRew}{\mi} }
    \subseteq
    \cl{\Ach{\MDP}{\OpRew}{\mi} }~.
\]
For the converse inclusion, it suffices to prove $\Ach{\MDP}{\OpRew}{\mi} \subseteq\cl{\AchMemless{\MDP}{\OpRew}{\mi} }$, because the set $\cl{\AchMemless{\MDP}{\OpRew}{\mi} }$ is Scott closed.

Fix an arbitrary scheduler $\sched\in\Sched{\MDP}$.
For every $k\in\Nats$, we construct a memoryless scheduler $\sched_k''\in\SchedMemless{\MDP}$ whose complete expected reward dominates the expected reward of $\sched$ after at most $k$ steps.
Formally, for every $k\in\Nats$,
\[
    \ExpRewI{\MDP}{\OpRew}{\mi}{\sched}{k}
    \lleq
    \ExpRew{\MDP}{\OpRew}{\mi}{\sched_k''}~.
\]
All inequalities between reward vectors are understood componentwise.
Once this inequality has been established, the right-hand side belongs
to $\AchMemless{\MDP}{\OpRew}{\mi} $, and hence
\[
    \forall k \in \Nats \colon\quad  \ExpRewI{\MDP}{\OpRew}{\mi}{\sched}{k}
    \in
    \dwc{\AchMemless{\MDP}{\OpRew}{\mi} }
    \subseteq
    \cl{\AchMemless{\MDP}{\OpRew}{\mi} }~.
\]
As $k$ varies, nonnegativity of the rewards makes these finite-horizon
expected-reward vectors an increasing chain and, by definition,
\[
    \sup_{k\in\Nats}
        \ExpRewI{\MDP}{\OpRew}{\mi}{\sched}{k}
    \eeq
    \ExpRew{\MDP}{\OpRew}{\mi}{\sched}~.
\]
The Scott closedness of $\cl{\AchMemless{\MDP}{\OpRew}{\mi} }$ then gives
\[
    \ExpRew{\MDP}{\OpRew}{\mi}{\sched}
    \in
    \cl{\AchMemless{\MDP}{\OpRew}{\mi} }~.
\]
Since $\sched\in\Sched{\MDP}$ was arbitrary, this is precisely the
inclusion $\Ach{\MDP}{\OpRew}{\mi} \subseteq\cl{\AchMemless{\MDP}{\OpRew}{\mi} }$ required above.

The remainder of the proof establishes this argument in four steps.
First, for fixed $k$, we construct the finite-state truncation
$\MDP_k^\bot$ containing every state that can be relevant during the
first $k$ transitions from $\mi$; its finiteness allows us to apply the
finite-state result of \Cref{theo:memless-suff-fin}.
Second, we restrict $\sched$ to a scheduler $\sched_k$ on
$\MDP_k^\bot$, preserving its expected reward after at most $k$ steps,
and use that finite-state result to replace $\sched_k$ by a memoryless
scheduler $\sched_k'$ whose complete expected reward dominates that of
$\sched_k$.
Third, we lift $\sched_k'$ to a memoryless scheduler $\sched_k''$ on
$\MDP$.
The lift preserves, with the same probability and reward, every path of
$\MDP_k^\bot$ that reaches $\MT_k^\bot$, while $\MDP$ may have
additional target-reaching paths with nonnegative reward contributions.
Thus, the comparisons from the first three steps form the chain
\begin{align*}
    \ExpRewI{\MDP}{\OpRew}{\mi}{\sched}{k}
    &\eeq
    \ExpRewI{\MDP_k^\bot}{\OpRew_k^\bot}{\mi}{\sched_k}{k}
    \\
    &\lleq
    \ExpRew{\MDP_k^\bot}{\OpRew_k^\bot}{\mi}{\sched_k}
    \\
    &\lleq
    \ExpRew{\MDP_k^\bot}{\OpRew_k^\bot}{\mi}{\sched_k'}
    \\
    &\lleq
    \ExpRew{\MDP}{\OpRew}{\mi}{\sched_k''}~.
\end{align*}
Finally, Step~4 lets $k$ increase: the finite-horizon expected-reward
vectors form a directed chain with supremum
$\ExpRew{\MDP}{\OpRew}{\mi}{\sched}$, and Scott closedness transfers
their membership in $\cl{\AchMemless{\MDP}{\OpRew}{\mi} }$ to that supremum.

\paragraph{Constructing a finite-state truncation.}
Fix a fresh state $\bot\notin\OpStates$.
Starting from $\mi$, define recursively for every $l \in \Nats$,
\[
    \OpStates_0\coloneqq\set{\mi}
    \qquad\text{and}\qquad
    \OpStates_{l+1}
    \coloneqq
    \OpStates_l
    \cup
    \bigcup_{\opState'\in\OpStates_l}\OpSucc(\opState')~.
\]
Thus, $\OpStates_k$ is precisely the set of states reachable from
$\mi$ in at most $k$ transitions.
The set $\OpStates_0$ is finite.
Moreover, if $\OpStates_l$ is finite, then so is
$\OpStates_{l+1}$: the action set is finite and every state-action pair
has finite successor support.
Consequently, every $\OpStates_k$ is finite.

Fix $k\in\Nats$ and let
\[
    \OpStates_k^\bot
    \coloneqq
    \OpStates_k\cup\set{\bot},
    \qquad
    {\MI}_k^\bot
    \coloneqq
    \set{\mi}~.
\]
We define
\[
    \MDP_k^\bot
    \coloneqq
    \bigl(
        \OpStates_k^\bot,
        {\MI}_k^\bot,
        \OpAct,
        \OpP_k^\bot
    \bigr),
\]
where, for all $\opState',\opState''\in\OpStates_k^\bot$ and
$\opAct\in\OpAct$,
\[
    \OpP_k^\bot(\opState',\opAct,\opState'')
    \coloneqq
    \begin{cases}
        \OpP(\opState',\opAct,\opState'')
        & \text{if $\opState',\opState''\in\OpStates_k$,}
        \\[1mm]
        \displaystyle
        \sum_{\widetilde{\opState}\in\OpStates\setminus\OpStates_k}
            \OpP(\opState',\opAct,\widetilde{\opState})
        & \text{if $\opState'\in\OpStates_k$ and $\opState''=\bot$,}
        \\[4mm]
        1
        & \text{if $\opState'=\opState''=\bot$,}
        \\[1mm]
        0
        & \text{otherwise}
    \end{cases}~.
\]
Hence, transitions inside $\OpStates_k$ retain their probabilities,
all probability mass leaving $\OpStates_k$ is redirected to $\bot$,
and $\bot$ is absorbing.

For every $\opState'\in\OpStates_k$ and $\opAct\in\OpAct$, the total
outgoing probability in $\MDP_k^\bot$ equals the total outgoing
probability in $\MDP$.
It follows that the enabled actions of $\MDP_k^\bot$ are
\[
    \OpAct_k^\bot(\opState')
    \coloneqq
    \begin{cases}
        \OpAct(\opState')
        & \text{if $\opState'\in\OpStates_k$,}
        \\
        \OpAct
        & \text{if $\opState'=\bot$}
    \end{cases}~.
\]
In particular, $\MDP_k^\bot$ is a well-defined finite-state MDP.

Its target set and reward function are
\[
    \MT_k^\bot
    \coloneqq
    \MT\cap\OpStates_k
\]
and
\[
    \OpRew_k^\bot
    \colon
    \MT_k^\bot
    \to
    \PosRealsInfVect,
    \qquad
    \OpRew_k^\bot(\opState')
    \coloneqq
    \OpRew(\opState')~.
\]
Since every target in $\MT$ is a sink in $\MDP$, every target in
$\MT_k^\bot$ is also a sink in $\MDP_k^\bot$.

\paragraph{Restricting $\sched$ and eliminating memory.}
Write
\[
    \supp{\sched}
    \eeq
    \set{\hat{\sched}_1,\ldots,\hat{\sched}_m}~.
\]
Thus, each $\hat{\sched}_i$ is a pure scheduler for $\MDP$ and
\[
    \sum_{i=1}^{m}\sched(\hat{\sched}_i)=1~.
\]

Fix a total order on the finite set $\OpAct$.
For every $1\leq i\leq m$, define a pure scheduler
$\hat{\sched}_{i,k}$ for $\MDP_k^\bot$ on every sequence
$\opState_0\ldots\opState_l\in(\OpStates_k^\bot)^+$ by
\[
    \hat{\sched}_{i,k}(\opState_0\ldots\opState_l)
    \coloneqq
    \begin{cases}
        \hat{\sched}_i(\opState_0\ldots\opState_l)
        & \begin{array}{l}
          \text{if $l\leq k$ and $\opState_j\neq\bot$}\\[-1mm]
          \text{for every $0\leq j\leq l$,}
          \end{array}
        \\[4mm]
        \min\bigl(\OpAct_k^\bot(\opState_l)\bigr)
        & \text{otherwise}
    \end{cases}~.
\]
Here, the minimum is taken with respect to the fixed order.
The first case selects an enabled action because
$\OpAct_k^\bot(\opState_l)=\OpAct(\opState_l)$, and the second case is
well-defined because every enabled-action set is nonempty.
Thus, $\hat{\sched}_{i,k}\in\SchedPure{\MDP_k^\bot}$.

Distinct schedulers in $\supp{\sched}$ can have the same restriction.
We therefore define the mixed scheduler $\sched_k$ by adding the mass
of all coincident restrictions: for every
$\hat{\sched}'\in\SchedPure{\MDP_k^\bot}$, let
\[
    \sched_k(\hat{\sched}')
    \coloneqq
    \sum_{\substack{
        1\leq i\leq m\\
        \hat{\sched}_{i,k}=\hat{\sched}'
    }}
        \sched(\hat{\sched}_i)~.
\]
As usual, empty sums evaluate to $0$.
The support of $\sched_k$ is the set of distinct schedulers among
$\hat{\sched}_{1,k},\ldots,\hat{\sched}_{m,k}$, and the fibres in the
preceding definition partition the index set $\set{1,\ldots,m}$.
Consequently,
\[
    \sum_{\hat{\sched}'\in\supp{\sched_k}}
        \sched_k(\hat{\sched}')
    \eeq
    \sum_{i=1}^{m}\sched(\hat{\sched}_i)
    \eeq
    1~.
\]
Thus, $\sched_k\in\Sched{\MDP_k^\bot}$ is a well-defined mixed
scheduler.

Let $\finPath=\opState_0\ldots\opState_l$ be a sequence of states in
$\OpStates_k$ with $l\leq k$.
For every $1\leq i\leq m$, the definitions of
$\hat{\sched}_{i,k}$ and $\OpP_k^\bot$ give
\[
    \pathProb{\MDP_k^\bot}{\hat{\sched}_{i,k}}(\finPath)
    \eeq
    \pathProb{\MDP}{\hat{\sched}_i}(\finPath)~.
\]
Regrouping the finite mixture according to the fibres defining
$\sched_k$ therefore yields
\begin{align*}
    \pathProb{\MDP_k^\bot}{\sched_k}(\finPath)
    &\eeq
    \sum_{i=1}^{m}
        \sched(\hat{\sched}_i)
        \cdot
        \pathProb{\MDP_k^\bot}{\hat{\sched}_{i,k}}(\finPath)
    \\
    &\eeq
    \sum_{i=1}^{m}
        \sched(\hat{\sched}_i)
        \cdot
        \pathProb{\MDP}{\hat{\sched}_i}(\finPath)
    \\
    &\eeq
    \pathProb{\MDP}{\sched}(\finPath)~.
\end{align*}

At a fixed horizon, finite branching and the finite support of a mixed
scheduler leave only finitely many paths with nonzero probability.
Zero-probability paths contribute $0$, also for an infinite terminal
reward, by the convention $0\cdot\infty=0$.
Every path of length at most $k$ with nonzero probability in $\MDP$
from $\mi$ stays in $\OpStates_k$.
Conversely, a path with nonzero probability in $\MDP_k^\bot$ that
reaches $\MT_k^\bot$ cannot contain $\bot$, because $\bot$ is
absorbing and not a target.
Such a path lies in $\OpStates_k$.
Moreover, each of its states before the final one lies outside
$\MT_k^\bot=\MT\cap\OpStates_k$ and hence outside $\MT$.
It therefore reaches $\MT$ for the first time at the same state and
receives the same reward in both MDPs.
The preceding path-probability equality proves
\begin{align}
    \ExpRewI{\MDP}{\OpRew}{\mi}{\sched}{k}
    \eeq
    \ExpRewI{\MDP_k^\bot}{\OpRew_k^\bot}{\mi}{\sched_k}{k}~.
    \label{eq:truncation-simulation}
\end{align}

By the definition of the complete expected reward,
\[
    \ExpRewI{\MDP_k^\bot}{\OpRew_k^\bot}{\mi}{\sched_k}{k}
    \lleq
    \ExpRew{\MDP_k^\bot}{\OpRew_k^\bot}{\mi}{\sched_k}~.
\]
Since $\MDP_k^\bot$ is finite-state,
\Cref{theo:memless-suff-fin} provides a scheduler
$\sched_k'\in\SchedMemless{\MDP_k^\bot}$ such that
\[
    \ExpRew{\MDP_k^\bot}{\OpRew_k^\bot}{\mi}{\sched_k}
    \lleq
    \ExpRew{\MDP_k^\bot}{\OpRew_k^\bot}{\mi}{\sched_k'}~.
\]
Together with \Cref{eq:truncation-simulation}, this gives
\begin{align}
    \ExpRewI{\MDP}{\OpRew}{\mi}{\sched}{k}
    \lleq
    \ExpRew{\MDP_k^\bot}{\OpRew_k^\bot}{\mi}{\sched_k'}~.
    \label{eq:finite-fragment-domination}
\end{align}

\paragraph{Lifting the memoryless scheduler to $\MDP$.}
Write
\[
    \supp{\sched_k'}
    \eeq
    \set{
        \hat{\msched}_{k,1},\ldots,\hat{\msched}_{k,r}
    }~.
\]
The support size $r$ need not equal the support size $m$ used above.
Since $\sched_k'$ is memoryless, each
$\hat{\msched}_{k,j}$ can be identified with a map
$\OpStates_k^\bot\to\OpAct$.

For every $1\leq j\leq r$, define a pure memoryless scheduler
$\hat{\msched}_{k,j}'$ for $\MDP$ by
\[
    \hat{\msched}_{k,j}'(\opState')
    \coloneqq
    \begin{cases}
        \hat{\msched}_{k,j}(\opState')
        & \text{if $\opState'\in\OpStates_k$,}
        \\
        \min\bigl(\OpAct(\opState')\bigr)
        & \text{if $\opState'\in\OpStates\setminus\OpStates_k$}
    \end{cases}~.
\]
The minimum is again taken with respect to the fixed order on
$\OpAct$.
This lift is well-defined because the enabled actions agree on
$\OpStates_k$ and every state of $\MDP$ has at least one enabled
action.

Different schedulers for $\MDP_k^\bot$ can have the same lift; for
example, they may differ only at $\bot$.
For every $\hat{\sched}'\in\SchedPure{\MDP}$, define
\[
    \sched_k''(\hat{\sched}')
    \coloneqq
    \sum_{\substack{
        1\leq j\leq r\\
        \hat{\msched}_{k,j}'=\hat{\sched}'
    }}
        \sched_k'(\hat{\msched}_{k,j})~.
\]
Its support is the set of distinct schedulers among
$\hat{\msched}_{k,1}',\ldots,\hat{\msched}_{k,r}'$, and
\[
    \sum_{\hat{\sched}'\in\supp{\sched_k''}}
        \sched_k''(\hat{\sched}')
    \eeq
    \sum_{j=1}^{r}\sched_k'(\hat{\msched}_{k,j})
    \eeq
    1~.
\]
Thus, $\sched_k''$ is a mixed scheduler for $\MDP$.
All pure schedulers in its support are memoryless, so
$\sched_k''\in\SchedMemless{\MDP}$.

Let $\finPath=\opState_0\ldots\opState_l$ be any sequence of states in
$\OpStates_k$.
For every $1\leq j\leq r$, the scheduler and transition definitions
give
\[
    \pathProb{\MDP}{\hat{\msched}_{k,j}'}(\finPath)
    \eeq
    \pathProb{\MDP_k^\bot}{\hat{\msched}_{k,j}}(\finPath)~.
\]
Hence, we obtain
\begin{align*}
    \pathProb{\MDP}{\sched_k''}(\finPath)
    &\eeq
    \sum_{j=1}^{r}
        \sched_k'(\hat{\msched}_{k,j})
        \cdot
        \pathProb{\MDP}{\hat{\msched}_{k,j}'}(\finPath)
    \\
    &\eeq
    \sum_{j=1}^{r}
        \sched_k'(\hat{\msched}_{k,j})
        \cdot
        \pathProb{\MDP_k^\bot}{\hat{\msched}_{k,j}}(\finPath)
    \\
    &\eeq
    \pathProb{\MDP_k^\bot}{\sched_k'}(\finPath)~.
\end{align*}

Fix any horizon $n\in\Nats$.
Every path of length at most $n$ from $\mi$ that contributes a nonzero summand to the expected-reward sum for $\MDP_k^\bot$ avoids $\bot$ and therefore lies entirely in $\OpStates_k$.
If such a path reaches $\MT_k^\bot$ for the first time, then it also reaches $\MT$ for the first time: all of its states lie in $\OpStates_k$, where $\MT_k^\bot=\MT\cap\OpStates_k$.
The same path hence occurs in the $n$-step expected-reward sum for $\MDP$, with the same probability and terminal reward.
The original MDP may have additional target-reaching paths, whose reward contributions are all nonnegative.
After omitting zero-probability summands using the convention $0\cdot\infty=0$, the truncated sum is therefore a sub-sum of the original sum, and
\[
    \ExpRewI{\MDP_k^\bot}{\OpRew_k^\bot}{\mi}{\sched_k'}{n}
    \lleq
    \ExpRewI{\MDP}{\OpRew}{\mi}{\sched_k''}{n}~.
\]
Taking componentwise suprema over $n\in\Nats$ proves
\begin{align}
    \ExpRew{\MDP_k^\bot}{\OpRew_k^\bot}{\mi}{\sched_k'}
    \lleq
    \ExpRew{\MDP}{\OpRew}{\mi}{\sched_k''}~.
    \label{eq:full-lifting-domination}
\end{align}
Combining
\Cref{eq:finite-fragment-domination,eq:full-lifting-domination}
yields, for every $k\in\Nats$,
\begin{align}
    \ExpRewI{\MDP}{\OpRew}{\mi}{\sched}{k}
    \lleq
    \ExpRew{\MDP}{\OpRew}{\mi}{\sched_k''}~.
    \label{eq:finite-horizon-below-memoryless}
\end{align}

\paragraph{Passing to the Scott closure.}
For every $k\in\Nats$, let
\[
    x_k
    \coloneqq
    \ExpRewI{\MDP}{\OpRew}{\mi}{\sched}{k}~.
\]
By \Cref{eq:finite-horizon-below-memoryless},
\[
    \forall k \in \Nats \colon \quad x_k
    \in
    \dwc{\AchMemless{\MDP}{\OpRew}{\mi} }
    \subseteq
    \cl{\AchMemless{\MDP}{\OpRew}{\mi} }~.
\]
The finite-horizon expected-reward vectors form an increasing chain:
the defining sum for $x_{k+1}$ contains every summand occurring in the
defining sum for $x_k$, together with the nonnegative contributions of
paths that first reach a target after exactly $k+1$ transitions.
Hence,
\[
    \forall k \in \Nats \colon \quad x_k\lleq x_{k+1}~,
\]
so $\set{x_k\mid k\in\Nats}$ is directed.
Since $\cl{\AchMemless{\MDP}{\OpRew}{\mi} }$ is closed under directed suprema and
contains every $x_k$, it also contains their supremum.
By the definition of the complete expected reward,
\[
    \ExpRew{\MDP}{\OpRew}{\mi}{\sched}
    \eeq
    \sup_{k\in\Nats}x_k
    \in
    \cl{\AchMemless{\MDP}{\OpRew}{\mi} }~.
\]
The scheduler $\sched\in\Sched{\MDP}$ was arbitrary.
Thus, $\Ach{\MDP}{\OpRew}{\mi} \subseteq\cl{\AchMemless{\MDP}{\OpRew}{\mi} }$.
By the defining minimality of Scott closure, this implies
$\cl{\Ach{\MDP}{\OpRew}{\mi} }\subseteq\cl{\AchMemless{\MDP}{\OpRew}{\mi} }$.
Together with the reverse inclusion proved at the beginning, this establishes $\cl{\Ach{\MDP}{\OpRew}{\mi} }=\cl{\AchMemless{\MDP}{\OpRew}{\mi} }$.
Since $\mi\in\MI$ was arbitrary, the claim follows.
\end{proof}

\subsection{A Bellman Operator for Computing the Almost Achievable Points}
\label{sec:computing-pareto}
We now characterize the achievable rewards as the least fixpoint of an adequately defined Bellman operator.
We first define the lattice, reusing the Hoare powerdomain $\MVdom$ from \cref{def:hoare-power}:
\begin{definition}[Multivalue Functions]
\label{def:multivalue-function}
    The complete lattice of \emph{multivalue functions} is defined as $(\MV,\MVleq)$, where
    \begin{enumerate}
        \item $\MV = \OpStates \to \MVdom$ is the \emph{set of multivalue functions}, and
        \item $\MVleq$ is the pointwise lifted set inclusion, i.e.\ for all $\mv,\mv' \in \MV$,
        \[
            \mv \MVleq \mv' \quad \text{iff} \quad \forall \opState \in \OpStates\colon \mv(\opState) \subseteq \mv'(\opState).
        \]
    \end{enumerate}
\end{definition}
This is the analogue of multiobjective expectations (see \Cref{def:mexp}).
The least element of $(\MV,\MVleq)$ is the constant zero function $\mylambda{\opState} \zeroHoare$, which we denote by abuse of notation as $\zeroMV$.
Moreover, suprema and infima are the pointwise liftings of suprema and infima in the Hoare powerdomain. %

To ease the adaptation to probabilistic programs, we generalize the notion of rewards:
We let $\OpRew\colon \OpStates \to \MVdom$, so $\OpRew\in \MV$.
We do not adapt the MDP definitions to this new reward type, but instead, when appropriate, use a simple cast function.

\begin{definition}
\label{def:cast-reward}
    Given $\OpRew\colon \MT \to \PosRealsInfVect$, we define $\OpRewCast\colon \MT \to \MVdom$ as
    $\OpRewCast(\opState) = \dwcset{\OpRew(\opState)}.$
\end{definition}

Towards the adaptation of (standard) Bellman operators, we lift the addition using Minkowski sums. %
Replacing the maximum by a union and a convex combination is a direct consequence of the power of our schedulers.
\begin{definition}[Bellman Operator for Downward Closed Pareto Front]
\label{def:multi-bellman}
    Let $\OpRew\colon \MT \to \MVdom$ be a reward assignment to some target set $\MT$.
    We define the function $\BellmanPareto{\OpRew}{\MDP} \colon \MV \to \MV$ as
    \[
        \BellmanPareto{\OpRew}{\MDP}(\mv) = \mylambda{\opState} 
        \begin{dcases}
            \OpRew(\opState) & \text{ if } \opState \in \MT \\
            \bigsqcup_{\opAct \in \OpAct(\opState)} \Minkow_{\opState' \in \OpSucc(\opState,\opAct)} (\OpP(\opState, \opAct, \opState') \cdot \mv(\opState')) & \text{ else }
        \end{dcases}
    \]
\end{definition}
This Bellman operator is $\omega$-continuous.
We are not the first to lift Bellman operators to accommodate multiple objectives---e.g.\ for stochastic games, Bellman operators have been lifted to downward closed convex sets of achievable vectors~\cite{chen2013stochastic,chen2013synthesis,basset2015strategy,ashok2020approximating}.
There are only minor differences between \Cref{def:multi-bellman} and the construction appearing in~\cite{chen2013stochastic}:
our rewards already have the value function's type, so we do not explicitly take a downward closure; the Hoare powerdomain handles this implicitly while also ensuring Scott closure and convexity.

The remainder of this section will be devoted to showing that the least fixpoint of this operator characterizes the achievable rewards, see \Cref{theo:bellman-is-pareto}.
This requires some work first:
we now prove several helper theorems, starting with proving an inductive computation for expected rewards.

\begin{lemma}[Decomposition of expected rewards]
\label{theo:ach-first-step}
    Let $\MDP = \left( \OpStates, \MI, \OpAct, \OpP \right)$ be an MDP and $\OpRew \colon \MT \to \PosRealsInfVect$ be a reward assignment.
    Fix an initial state $\mi \notin \MT$, let $\opAct \in \OpAct(\mi)$, and let $\msched \in \SchedPure{\MDP}$ be a pure scheduler with $\msched(\mi) = \opAct$.

    Then
    \[
        \ExpRewI{\MDP}{\OpRew}{\mi}{\msched}{k+1}
        =
        \sum_{\opState' \in \OpSucc(\mi,\opAct)}
        \OpP(\mi,\opAct,\opState')
        \cdot
        \ExpRewI{\MDP}{\OpRew}{\opState'}{\msched_{\opState'}}{k},
    \]
    where $\msched_{\opState'}$ denotes the scheduler obtained by
    restricting $\msched$ to histories whose first state is $\opState'$.
\end{lemma}
\begin{proof}
    We expand the definition of the expected reward %
    and split every path according to its first successor.
    By definition,
    \[
        \ExpRewI{\MDP}{\OpRew}{\mi}{\sched}{k+1} \eeq
    \sum_{0 \leq i \leq k+1} \sum_{\substack{\opState_0\ldots\opState_i\in \pathstotarget{\MT} \\ \opState_0 = \mi}} \pathProb{\MDP}{\msched}(\mi\opState_1\ldots\opState_i) \cdot \OpRew(\opState_i)~.
    \]
    For a path
    $
        \mi\opState_1\ldots\opState_i
    $
    with $i > 0$, we have
    \[
        \pathProb{\MDP}{\msched}(\mi\opState_1\ldots\opState_i)
        \eeq
        \OpP(\mi,\opAct,\opState_1)
        \cdot
        \pathProb{\MDP}{\msched_{\opState'}}
        (\opState_1\ldots\opState_i).
    \]
    Re-indexing with $j=i-1$ gives $0 \leq j \leq k$.
    Since $\mi \notin \MT$, for every $\opState_1,\ldots,\opState_i \in \OpStates$ with $i > 0$, we have
    \[
        \mi\opState_1\ldots\opState_i \in \pathstotarget{\MT}
        \quad\iff\quad
        \opState_1\ldots\opState_i \in \pathstotarget{\MT}.
    \]
    Indeed, both paths end in the same target state $\opState_i$, and the
    only additional state that has to be checked for the path on the left
    is $\mi$, which is not in $\MT$ by assumption.
    Thus,
    \begin{align*}
        \ExpRewI{\MDP}{\OpRew}{\mi}{\msched}{k+1}
        \eeq
        \sum_{\opState' \in \OpSucc(\mi,\opAct)}
        \OpP(\mi,\opAct,\opState')
        \cdot
        \ExpRewI{\MDP}{\OpRew}{\opState'}{\msched_{\opState'}}{k}.
    \end{align*}
\end{proof}

With \Cref{theo:ach-first-step} in mind, we observe the following consequence for schedulers.

\begin{lemma}[Splicing schedulers]
\label{theo:ach-splice}
    Let $\mi \notin \MT$ and let $\opAct \in \OpAct(\mi)$.
    Then,
    \begin{align*}
        & \setcomp{\ExpRewI{\MDP}{\OpRew}{\mi}{\msched}{k+1}}{\msched \in \SchedPure{\MDP}, \msched(\mi) = \opAct} \\
        \eeq &\ 
        \setcomp{
        \sum_{\opState' \in \OpSucc(\mi,\opAct)}
        \OpP(\mi,\opAct,\opState')
        \cdot
        \ExpRewI{\MDP}{\OpRew}{\opState'}{\msched_{\opState'}}{k}
        }{
        \msched_{\opState'} \in \SchedPure{\MDP}
        }
    \end{align*}
\end{lemma}
\begin{proof}
    We prove both inclusions.

    \noindent
    $\supseteq$:
    For every $\opState' \in \OpSucc(\mi,\opAct)$, let $\msched_{\opState'} \in \SchedPure{\MDP}$ be a pure scheduler.
    Define a pure scheduler $\msched \in \SchedPure{\MDP}$ on histories starting in $\mi$ by
    $
        \msched(\mi) = \opAct
    $
    and
    \[
        \msched(\mi\opState'\opState_2\ldots\opState_n)
        =
        \msched_{\opState'}(\opState'\opState_2\ldots\opState_n).
    \]

    On all other histories, define $\msched$ arbitrarily subject to
    $
        \msched(\dots \opState_k) \in \OpAct(\opState_k).
    $
    Since $\msched(\mi)=\opAct$, the first-step decomposition from \Cref{theo:ach-first-step} gives
    \[
        \ExpRewI{\MDP}{\OpRew}{\mi}{\msched}{k+1}
        =
        \sum_{\opState' \in \OpSucc(\mi,\opAct)}
        \OpP(\mi,\opAct,\opState')
        \cdot
        \ExpRewI{\MDP}{\OpRew}{\opState'}{\msched_{\opState'}}{k}.
    \]
    Hence, for every choice of continuation schedulers $\msched_{\opState'}$, the corresponding expression in the statement is attained by a pure scheduler whose first action at $\mi$ is $\opAct$.
    
    \noindent
    $\subseteq$:
    Conversely, let $\msched \in \SchedPure{\MDP}$ be any pure scheduler satisfying $\msched(\mi)=\opAct$.
    For every $\opState' \in \OpSucc(\mi,\opAct)$, define the continuation scheduler $\msched_{\opState'}$ by
    \[
        \msched_{\opState'}(\opState'\opState_2\ldots\opState_n)
        =
        \msched(\mi\opState'\opState_2\ldots\opState_n).
    \]
    Then \Cref{theo:ach-first-step} again yields
    \[
        \ExpRewI{\MDP}{\OpRew}{\mi}{\msched}{k+1}
        =
        \sum_{\opState' \in \OpSucc(\mi,\opAct)}
        \OpP(\mi,\opAct,\opState')
        \cdot
        \ExpRewI{\MDP}{\OpRew}{\opState'}{\msched_{\opState'}}{k}.
    \]
    Thus every expected reward obtained by a pure scheduler whose first action at $\mi$ is $\opAct$ is contained in the displayed set.

    Therefore, the two sets are equal.
\end{proof}

\Cref{theo:ach-first-step,theo:ach-splice} lead us to the following theorem, stating that we can compute the (closure of) $k$-step achievable points inductively.
This will be the key property used in the soundness proof of the Bellman operator.

\begin{theorem}[Inductive computation of achievable points]
\label{theo:ach-onestep}
    Let $\MDP = \left( \OpStates, \MI, \OpAct, \OpP \right)$ be an MDP and $\OpRew \colon \MT \to \PosRealsInfVect$ be a reward assignment.
    Then, for all $k \in \Nats$ and $\mi \in \MI$, we have
    \[
    \cl{\AchI{\MDP}{\OpRew}{\mi}{k+1}} = \begin{dcases}
            \dwcset{\OpRew(\mi)} & \text{ if } \mi \in \MT \\
            \bigsqcup_{\opAct \in \OpAct(\mi)} \Minkow_{\opState' \in \OpSucc(\mi,\opAct)} \left(\OpP(\mi, \opAct, \opState') \cdot \AchI{\MDP}{\OpRew}{\opState'}{k}\right) & \text{ else }
        \end{dcases}
    \]
\end{theorem}
\begin{proof}
    We assume that all states are initial, i.e., $\MI = \OpStates$.
    This is to avoid ill-definedness when considering achievable values from non-initial states, but is without loss of generality since the values for each initial state are independent of which other states are initial.
    We begin by taking a closer look on $\AchI{\MDP}{\OpRew}{\mi}{k+1}$.
    By definition,
    \begin{align*}
        \AchI{\MDP}{\OpRew}{\mi}{k+1} 
        =\ \dwcsetcomp{ \ExpRewI{\MDP}{\OpRew}{\mi}{\sched}{k+1} }{ \sched \in \Sched{\MDP}} %
    \end{align*}
    We distinguish two cases.

    \paragraph{Case 1: $\mi \in \MT$}
    Since $\mi$ is a target state and target states are, by definition, sinks, there is no path of positive length starting in $\mi$ that belongs to $\pathstotarget{\MT}$.
    Hence, for every $\sched \in \Sched{\MDP}$,
    \[
        \ExpRewI{\MDP}{\OpRew}{\mi}{\sched}{k+1}
        =
        \OpRew(\mi).
    \]
    Therefore,
    \[
        \AchI{\MDP}{\OpRew}{\mi}{k+1}
        =
        \dwcset{\OpRew(\mi)},
    \]
    which implies that
    \[
        \cl{\AchI{\MDP}{\OpRew}{\mi}{k+1}}
        =
        \dwcset{\OpRew(\mi)},
    \]
    proving the first case.

    \paragraph{Case 2: $\mi \notin \MT$}
    For every $i\in \Nats$ and $\opAct \in \OpAct(\mi)$, define
    \[
        Y_{\opAct}^i \coloneqq
        \setcomp{ \ExpRewI{\MDP}{\OpRew}{\mi}{\sched}{i} }{
        \sched \in \SchedPure{\MDP}, \sched(\mi)=\opAct}
    \]
    and for every $\opState \in \OpStates$
    \[
        Y_{\opState}^i
        \coloneqq
        \setcomp{
            \ExpRewI{\MDP}{\OpRew}{\opState}{\msched}{i}
        }{
            \msched \in \SchedPure{\MDP}
        }.
    \]
    By \Cref{theo:ach-splice},
    \begin{align*}
        Y_{\opAct}^{k+1}
        =\ & \setcomp{ \ExpRewI{\MDP}{\OpRew}{\mi}{\sched}{k+1} }{
        \sched \in \SchedPure{\MDP}, \sched(\mi)=\opAct} \\
        =\ & \setcomp{
            \sum_{\opState' \in \OpSucc(\mi,\opAct)} \OpP(\mi,\opAct,\opState') \cdot \ExpRewI{\MDP}{\OpRew}{\opState'}{\msched_{\opState'}}{k}
        }{
            \msched_{\opState'} \in \SchedPure{\MDP}
        } \tag{\Cref{theo:ach-splice}} \\
        =\ &
        \setcomp{
            \sum_{\opState' \in \OpSucc(\mi,\opAct)}
            \OpP(\mi,\opAct,\opState')
            \cdot x_{\opState'}
        }{
            x_{\opState'} \in Y_{\opState'}^k
        } \tag{Def.\ of $Y_{\opState'}^k$} \\
        =\ &
        \Minkow_{\opState' \in \OpSucc(\mi,\opAct)}
        \left(
            \OpP(\mi,\opAct,\opState')
            \cdot Y_{\opState'}^k
        \right).
        \tag{Def.\ of $\minkow$}
    \end{align*}

    Note that a mixed scheduler by definition is a finite distribution, or equivalently a convex combination, of pure schedulers, so we have for all $\opState \in \OpStates$
    \begin{align*}
        \setcomp{
            \ExpRewI{\MDP}{\OpRew}{\opState}{\sched}{k+1}
        }{
            \sched \in \Sched{\MDP}
        }
        \eeq 
        \conv{
            \setcomp{
                \ExpRewI{\MDP}{\OpRew}{\opState}{\msched}{k+1}
            }{
                \msched \in \SchedPure{\MDP}
            }
        }.
    \end{align*}
    Every pure scheduler chooses exactly one action at $\mi$.
    Hence the set of pure-scheduler expected rewards can be partitioned according to this first action:
    \begin{align*}
        &\conv{
            \setcomp{
                \ExpRewI{\MDP}{\OpRew}{\mi}{\msched}{k+1}
            }{
                \msched \in \SchedPure{\MDP}
            }
        }
        \\
        =\ &
        \conv{
            \bigcup_{\opAct\in\OpAct(\mi)}
            \setcomp{
                \ExpRewI{\MDP}{\OpRew}{\mi}{\msched}{k+1}
            }{
                \msched\in\SchedPure{\MDP},
                \ \msched(\mi)=\opAct
            }
        }
        \\
        =\ &
        \conv{
            \bigcup_{\opAct\in\OpAct(\mi)}
            Y_{\opAct}^{k+1}
        }.
    \end{align*}
    Consequently, by the definition of the achievable points,
    \begin{align*}
        \AchI{\MDP}{\OpRew}{\mi}{k+1}
        &=
        \dwc{
            \conv{
                \bigcup_{\opAct \in \OpAct(\mi)}
                Y_{\opAct}^{k+1}
            }
        }.
    \end{align*}
    Inserting the identity for $Y_{\opAct}^{k+1}$ we established above, we get
    \begin{align*}
        \AchI{\MDP}{\OpRew}{\mi}{k+1}
        &=
        \dwc{
            \conv{
                \bigcup_{\opAct \in \OpAct(\mi)}
                \Minkow_{\opState' \in \OpSucc(\mi,\opAct)}
                \left(
                    \OpP(\mi,\opAct,\opState')
                    \cdot Y_{\opState'}^k
                \right)}
        }.
    \end{align*}
    Further, note that we have for every state $\opState'$,
    \[
        \AchI{\MDP}{\OpRew}{\opState'}{k}
        \eeq
        \dwc{
            \conv{
                Y_{\opState'}^k
            }
        },
    \]
    by definition of $Y_{\opState'}^k$ and again the observation that mixed scheduler are convex combinations of pure schedulers.
    Since the successor set $\OpSucc(\mi,\opAct)$ is finite, we can apply \Cref{theo:closure-union,theo:closure-minkowski}, giving us
    \begin{align*}
        &\dwc{
            \conv{
                \bigcup_{\opAct\in\OpAct(\mi)}
                \Minkow_{\opState'\in\OpSucc(\mi,\opAct)}
                \left(
                    \OpP(\mi,\opAct,\opState')
                    \cdot Y_{\opState'}^k
                \right)
            }
        }
        \\
        =\ &
        \dwc{
            \conv{
                \bigcup_{\opAct\in\OpAct(\mi)}
                \Minkow_{\opState'\in\OpSucc(\mi,\opAct)}
                \left(
                    \OpP(\mi,\opAct,\opState')
                    \cdot
                    \dwc{\conv{Y_{\opState'}^k}}
                \right)
            }
        }
        \\
        =\ &
        \dwc{
            \conv{
                \bigcup_{\opAct\in\OpAct(\mi)}
                \Minkow_{\opState'\in\OpSucc(\mi,\opAct)}
                \left(
                    \OpP(\mi,\opAct,\opState')
                    \cdot
                    \AchI{\MDP}{\OpRew}{\opState'}{k}
                \right)
            }
        }.
    \end{align*}
    Thus, we get
    \[
        \cl{\AchI{\MDP}{\OpRew}{\mi}{k+1}}
        =
        \cl{
            \conv{
                \bigcup_{\opAct \in \OpAct(\mi)}
                \Minkow_{\opState' \in \OpSucc(\mi,\opAct)}
                \left(
                    \OpP(\mi,\opAct,\opState')
                    \cdot
                    \AchI{\MDP}{\OpRew}{\opState'}{k}
                \right)
            }
        },
    \]
    which is by \Cref{theo:sup-is-cl} equivalent to
    \[
        \bigsqcup_{\opAct \in \OpAct(\mi)} \Minkow_{\opState' \in \OpSucc(\mi,\opAct)} \left(\OpP(\mi, \opAct, \opState') \cdot \AchI{\MDP}{\OpRew}{\opState'}{k}\right) 
    \]

    This proves the second case and therefore the theorem.
\end{proof}%

Finally, we are in position to show that the Bellman operator is suited for what we are aiming for, namely that it indeed computes the almost achievable rewards in MDPs.

\begin{restatable}[Achievable Rewards via Bellman Operators]{theorem}{bellmanispareto}
\label{theo:bellman-is-pareto}
    Let $\OpRew\colon \MT \to \PosRealsInfVect$ be a reward assignment to some target set $\MT$ and $\mi \in \MI$.
    \begin{enumerate}
        \item $\BellmanPareto{\OpRewCast}{\MDP}^{k}(\zeroMV)(\mi) = \cl{\AchI{\MDP}{\OpRew}{\mi}{k-1}}$ for all $k > 0$
        \item $\lfp \BellmanPareto{\OpRewCast}{\MDP}(\mi) = \cl{\Ach{\MDP}{\OpRew}{\mi}}$
    \end{enumerate}
\end{restatable}

\begin{proof}
We prove (1) by induction on $k$, i.e., for all $k > 0$
\[
\BellmanPareto{\OpRewCast}{\MDP}^{k}(\zeroMV)(\mi) = \cl{\AchI{\MDP}{\OpRew}{\mi}{k-1}}.
\]
We assume that $\MI = \OpStates$ without loss of generality.
For $k=1$, we have
\begin{align*}
    \BellmanPareto{\OpRewCast}{\MDP}^{1}(\zeroMV)(\mi)
    & = \begin{dcases}
            \OpRewCast(\mi) & \text{ if } \mi \in \MT \\
            \bigsqcup_{\opAct \in \OpAct(\mi)} \Minkow_{\opState' \in \OpSucc(\mi,\opAct)} (\OpP(\mi, \opAct, \opState') \cdot \zeroMV(\opState')) & \text{ else }
        \end{dcases} \\
    & = \begin{dcases}
            \OpRewCast(\mi) & \text{ if } \mi \in \MT \\
            0 & \text{ else }
        \end{dcases} \\
\end{align*}

This is precisely equal to the reward collected in zero steps.
Formally:
\begin{align*}
    \cl{\AchI{\MDP}{\OpRew}{\mi}{0} }
    & = \cl{\setcomp{ \ExpRewI{\MDP}{\OpRew}{\mi}{\sched}{0} }{ \sched \in \Sched{\MDP}}} \tag{\Cref{def:achievable}}\\
    & = \cl{\setcomp{ \sum_{0 \leq i \leq 0} \sum_{\substack{\opState_0\ldots\opState_i\in \pathstotarget{\MT} \\ \opState_0 = \mi}} \pathProb{\MDP}{\sched}(\mi\opState_1\ldots\opState_i) \cdot \OpRew(\opState_i) }{ \sched \in \Sched{\MDP}}} \tag{\Cref{def:etr-bound}}\\
    & = \begin{cases}
        \cl{\setcomp{ \pathProb{\MDP}{\sched}(\mi) \cdot \OpRew(\mi) }{ \sched \in \Sched{\MDP}}}  & \text{ if } \mi \in \MT \\
        0 & \text{ else }
    \end{cases} \\
\end{align*}
For the first case assuming $\mi \in \MT$, we get:
\begin{align*}
    & \cl{\setcomp{ \pathProb{\MDP}{\sched}(\mi) \cdot \OpRew(\mi) }{ \sched \in \Sched{\MDP}}} \\
    =\ & \cl{\setcomp{ \left( \sum_{\sched' \in \supp{\sched}}  \sched(\sched') \cdot \pathProb{\MDP}{\sched'}(\mi) \right) \cdot \OpRew(\mi) }{ \sched \in \Sched{\MDP}}} \\  
    =\ & \cl{\setcomp{\left( \sum_{\sched' \in \supp{\sched}}  \sched(\sched') \cdot 1 \right) \cdot \OpRew(\mi) }{ \sched \in \Sched{\MDP}}} \\  
    =\ & \cl{\setcomp{1 \cdot \OpRew(\mi) }{ \sched \in \Sched{\MDP}}} \\  
    =\ & \cl{\set{\OpRew(\mi)}} \\
    =\ & \dwcset{\OpRew(\mi)} \\
    =\ & \OpRewCast(\mi) \tag{\Cref{def:cast-reward}}
\end{align*}

As induction hypothesis, assume that
\[
\BellmanPareto{\OpRewCast}{\MDP}^{k}(\zeroMV)(\mi) = \cl{\AchI{\MDP}{\OpRew}{\mi}{k-1}}.
\]

For $k+1$, we get:
\begin{align*}
    & \BellmanPareto{\OpRewCast}{\MDP}^{k+1}(\zeroMV)(\mi) \\
    =\ & \begin{dcases}
            \OpRewCast(\mi) & \text{ if } \mi \in \MT \\
            \bigsqcup_{\opAct \in \OpAct(\mi)} \Minkow_{\opState' \in \OpSucc(\mi,\opAct)} \left(\OpP(\mi, \opAct, \opState') \cdot \BellmanPareto{\OpRewCast}{\MDP}^{k}(\zeroMV)(\opState') \right) & \text{ else }
        \end{dcases} \\
    =\ & \begin{dcases}
            \OpRewCast(\mi) & \text{ if } \mi \in \MT \\
            \bigsqcup_{\opAct \in \OpAct(\mi)} \Minkow_{\opState' \in \OpSucc(\mi,\opAct)} \left(\OpP(\mi, \opAct, \opState') \cdot \cl{\AchI{\MDP}{\OpRew}{\opState'}{k-1}}\right) & \text{ else }
        \end{dcases} \tag{by IH}\\
    =\ & \begin{dcases}
            \dwcset{\OpRew(\mi)} & \text{ if } \mi \in \MT \\
            \bigsqcup_{\opAct \in \OpAct(\mi)} \Minkow_{\opState' \in \OpSucc(\mi,\opAct)} \left(\OpP(\mi, \opAct, \opState') \cdot \AchI{\MDP}{\OpRew}{\opState'}{k-1}\right) & \text{ else }
        \end{dcases} \tag{$\cl{\cdot}$ is taken upfront}\\
    =\ & \cl{\AchI{\MDP}{\OpRew}{\mi}{k} \tag{\Cref{theo:ach-onestep}}} \\
\end{align*}

The proof of (2) now is straightforward:
\begin{align*}
    \lfp \BellmanPareto{\OpRewCast}{\MDP}(\mi)
    =&\ (\bigsqcup_{k \in \Nats} {\BellmanPareto{\OpRewCast}{\MDP}^{k}}(\zeroMV))(\mi) \tag{by \hyperref[theo:Kleene]{Kleene}} \\
    =&\ \cl{ \bigcup_{k \in \Nats} \BellmanPareto{\OpRewCast}{\MDP}^{k}(\zeroMV)(\mi)} \tag{by Def.\ of $\bigsqcup$}\\
    =&\ \cl{ \bigcup_{k \in \Nats} \AchI{\MDP}{\OpRew}{\mi}{k-1}} \tag{by (1)} \\
    =&\ \cl{\Ach{\MDP}{\OpRew}{\mi}} \tag{\Cref{theo:clfin-clinf}} \\
\end{align*}
\end{proof}

\section{Operational Semantics of Probabilistic Programs}
\label{sec:mop-mdp}

We discussed multiobjective reasoning from two perspectives: at the program-level with the \mopsymbol transformer in the main paper, and at the operational level on MDPs with the Bellman operator $\BellmanParetoNo$ in \Cref{sec:mdp}.
In this section, we establish a tight correspondence between these results, leading up to the proof of \Cref{theo:mop-is-pareto}.

\subsection{Operational MDP}
\label{sec:actualMDPconstruction}

\begin{figure}[t]
    \centering
\begin{gather*}
  \infer{
    \Exec{\SKIP}{\pstate}{\labtau}{1}{\Term}{\pstate}
  }{
  }
  \quad
  \infer{
    \Exec{\ASSIGN{x}{\aexpr}}{\pstate}{\labtau}{1}{\Term}{\pstate\statesubst{x}{\aexpr(\pstate)}}
  }{
  }
\quad
\infer{
	\Exec{\Term}{\pstate}{\labtau}{1}{\Term}{\pstate}
}{
}
  \\[1ex]
  \infer{
    \Exec{\COMPOSE{\program_1}{\program_2}}{\pstate}{\actlab}{\pnum}{\program_2}{\pstate'}
  }{
    \Exec{\program_1}{\pstate}{\actlab}{\pnum}{\Term}{\pstate'}
  }
  \quad
  \infer{
    \Exec{\COMPOSE{\program_1}{\program_2}}{\pstate}{\actlab}{\pnum}{\COMPOSE{\program_1'}{\program_2}}{\pstate'}
  }{
    \Exec{\program_1}{\pstate}{\actlab}{\pnum}{\program_1'}{\pstate'}
  }
  \\[1ex]
  \infer{
    \Exec{\GC{\guard_1}{\program_1}{\guard_2}{\program_2}}{\pstate}{\labalpha}{1}{\program_1}{\pstate}
  }{
    \pstate \models \guard_1
  }
  \quad
  \infer{
	\Exec{\GC{\guard_1}{\program_1}{\guard_2}{\program_2}}{\pstate}{\labbeta}{1}{\program_2}{\pstate}
}{
	\pstate \models \guard_2
}
  \\[1ex]
    \infer{
  	    \Exec{\PCHOICE{\program_1}{\pexpr}{\program_2}}{\pstate}{\labtau}{\pexpr(\pstate)}{\program_1}{\pstate}
  	    }{
          \program_1 \neq \program_2
  	    }
    \quad
    \infer{
  	    \Exec{\PCHOICE{\program_1}{\pexpr}{\program_2}}{\pstate}{\labtau}{1-\pexpr(\pstate)}{\program_2}{\pstate}
  	    }{
          \program_1 \neq \program_2
  	    }
    \\[1ex]
    \infer{
        \Exec{\PCHOICE{\program}{\pexpr}{\program}}{\pstate}{\labtau}{1}{\program}{\pstate}
        }{
        }
    \\[1ex]
  \infer{
    \Exec{\WHILEDO{\guard}{\program}}{\pstate}{\labtau}{1}{\Term}{\pstate}
  }{
    \pstate \models \neg\guard
  }
  \quad
  \infer{
  \Exec{\WHILEDO{\guard}{\program}}{\pstate}{\labtau}{1}{\COMPOSE{\program}{\WHILEDO{\guard}{\program}}}{\pstate}
  }{
    \pstate \models \guard
  }
\end{gather*}
\caption{Rules defining the small-step execution relation $\rightarrow$ \cite[Figure 18]{batz2024programmatic}.}
\label{fig:smallstep}
\Description{
    Inference rules for the small-step execution relations, defining which configuration is reachable from each program.
}
\end{figure}

We define the standard small-step execution relation (cf.~\cite{batz2024programmatic}).
Program configurations are pairs $(\program,\pstate)$, where $\program \in \pGCL \cup \{\Term\}$ is a program (or the termination symbol) and $\pstate \in \States$ is a program state.
Formally, the uncountable set of program configurations is $\OpStates \eeq (\pGCL \cup \{\Term\}) \times \States$.
The execution relation is $\ExecSymbol \subseteq \OpStates \times \{\labtau,\labalpha,\labbeta\} \times [0,1] \times \OpStates$.
Each transition consists of an action label $\actlab \in \{\labtau,\labalpha,\labbeta\}$ and a probability $\pnum \in [0,1]$.
Its formal definition is given in \Cref{fig:smallstep} and follows the standard operational semantics of $\pGCL$.
The label $\labtau$ is used for all transitions except those induced by guarded commands $\GC{\guard_1}{\program_1}{\guard_2}{\program_2}$.
The labels $\labalpha$ and $\labbeta$ distinguish the branch taken by a guarded command.
We write $\ExecAbbr{\mc}{\actlab}{\pnum}{\mc'}$ for $(\mc,\actlab,\pnum,\mc') \in \ExecSymbol$.
We define the successor relation by $\mc \rightharpoonup \mc' \iff \exists\,\actlab,\pnum:\ExecAbbr{\mc}{\actlab}{\pnum}{\mc'}$.
A configuration $\mc'$ is \emph{reachable} from $\mc$ if it is in the reflexive-transitive closure of $\rightharpoonup$.
Based on $\ExecSymbol$, we now define the MDP semantics.

\begin{definition}
\label{def:opsem-pgcl}
    The \emph{operational Markov decision process} $\ProgMDP$ of $\pGCL$ is %
    \begin{align*}
    \ProgMDP \eeq \left( \OpStates, \MI, \OpAct, \OpP \right)~\text{, where}
    \end{align*}
    \begin{enumerate}
        \item $\OpStates = (\pGCL \cup \{\Term\}) \times \States$ is the set of program configurations, %
        \item $\OpAct = \{\labtau,\labalpha,\labbeta\}$ is the set of \emph{action labels},
        \item $\OpP\colon \OpStates \times \OpAct \times \OpStates \to [0,1]$ is the \emph{transition probability function} given by 
        \[
        \OpP(\mc, \actlab,\mc')\eeq
        \begin{cases}
        \pnum & \text{ if } \ExecAbbrSmash{\mc}{\actlab}{\pnum}{\mc'} \\
        0 & \text{else}~,
        \end{cases}
        \]
        \item and $\MI = \{(\program, \pstate) \mid \pstate \in \States, \program \in \pGCL \}$ are the initial states of $\ProgMDP$.
    \end{enumerate}
\end{definition}%

\noindent
Note that this MDP has uncountably many initial states as there are uncountably many \pGCL programs, but from every initial state, only countably many states are reached.

Next, we introduce reward functions corresponding to both single- and multiobjective postexpectations.
Upon reaching a terminal state $(\Term,\pstate)$, we collect reward $\post(\pstate)$ respectively $\mpost(\pstate)$.

\begin{definition}[Rewards from Expectations]
\label{def:rew-mexp}
    Given $\program \in \pGCL$, we define reward structures on the target set $\set{(\Term,\pstate) \mid \pstate \in \States}$ as follows:
    
    For $\post \in \Exp^n$, we define $\ProgR{\post} \colon \OpStates \to \PosRealsInfVect$ as $\ProgR{\post}(\Term,\pstate) = \post(\pstate)$.

    For $\mpost \in \mExp$, we define $\ProgR{\mpost}\colon \OpStates \to \HoarePowerDom$ as $\ProgR{\mpost}(\Term,\pstate) = \mpost(\pstate)$.
\end{definition}

This gives us a means to compare the notions we have defined on program-level to the established notions on MDPs.

\subsection{Relating mop and the Bellman Operator}

In \Cref{sec:mdp}, we introduced a Bellman-style operator for computing Pareto fronts in infinite-state MDPs.
In this section, we establish that the \mopsymbol\ transformer coincides with this Bellman operator when instantiated on the operational MDP semantics of pGCL for an appropriate choice of reward function.
This correspondence is analogous to classical soundness results for weakest preexpectation semantics, where the transformer is shown to agree with its operational model (see, for example, \cite[Theorem 6]{batz2024jp}).

\begin{restatable}[Operational Soundness of \mopsymbol]{theorem}{soundmop}
\label{theo:sound-mop}
    $
        \mop{\program}{\mpost} = \mylambda \pstate \left(\lfp \BellmanPareto{\ProgR{\mpost}}{\ProgMDP}\right)\left(\left(\program,\pstate\right)\right).
    $
\end{restatable}

The proof proceeds in three steps, following the proof of \cite[Theorem 6]{batz2024jp}.
First, \Cref{theo:help-op} characterizes one-step unfolding of $\mopsymbol$ in terms of the small-step operational semantics.
Second, \Cref{theo:omop-compose} establishes that sequential composition is sound for the operational transformer $\omopsymbol$.
Using these two auxiliary results, we then prove the operational soundness of $\mopsymbol$.

\begin{theorem}%
\label{theo:help-op}
    For all $\program \in \pGCL$ and $\mpost \in \mExp$:
    \begin{align*}
        \mop{\program}{\mpost} 
        \eeq
        \mylambda{\pstate} \bigsqcup_{\opAct \in \OpAct((\program,\pstate))}
        \Minkow_{\substack{
            \step{(\program,\pstate)}{\opAct,p}{\opConf'}
        }}
        p \cdot
        \begin{cases}
            \mpost(\tau)
            & \text{if } \opConf' = (\Term,\tau),\\
            \mop{\program'}{\mpost}(\pstate')
            & \text{if } \opConf' = (\program',\pstate').
        \end{cases}
    \end{align*}
\end{theorem}
\begin{proof}
We prove the claim by structural induction on $\program$.
    
\paragraph{$\program = \SKIP$.}
For every $\pstate \in \States$, \Cref{fig:smallstep} gives the unique transition
$\step{(\SKIP,\pstate)}{\labtau,1}{(\Term,\pstate)}$.
By the definition of $\mopsymbol$, we have
\begin{align*}
    \mop{\SKIP}{\mpost}(\pstate)
    \eeq & \mpost(\pstate) \\
    \eeq & 1 \cdot \mpost(\pstate) \\
    \eeq & \bigsqcup_{\opAct \in \OpAct((\SKIP,\pstate))}
    \Minkow_{\substack{
        \step{(\SKIP,\pstate)}{\opAct,p}{\opConf'}
    }}
    p \cdot
    \begin{cases}
        \mpost(\tau)
        & \text{if } \opConf' = (\Term,\tau),\\
        \mop{\program'}{\mpost}(\pstate')
        & \text{if } \opConf' = (\program',\pstate').
    \end{cases}
\end{align*}
where the last equality follows from the unique transition to $(\Term,\pstate)$.

\paragraph{$\program = \ASSIGN{x}{\ee}$.}
For every $\pstate \in \States$, \Cref{fig:smallstep} gives the unique transition
$\step{(\ASSIGN{x}{\ee},\pstate)}{\labtau,1}{(\Term,\pstate[x\mapsto\ee(\pstate)])}$.
By the definition of $\mopsymbol$, we have
\begin{align*}
    \mop{\ASSIGN{x}{\ee}}{\mpost}(\pstate)
    \eeq & \mpost\subst{x}{\ee}(\pstate) \\
    \eeq & \mpost(\pstate[x\mapsto\ee(\pstate)]) \\
    \eeq & 1 \cdot \mpost(\pstate[x\mapsto\ee(\pstate)]) \\
    \eeq & \bigsqcup_{\opAct \in \OpAct((\ASSIGN{x}{\ee},\pstate))}
    \Minkow_{\substack{
        \step{(\ASSIGN{x}{\ee},\pstate)}{\opAct,p}{\opConf'}
    }}
    p \cdot
    \begin{cases}
        \mpost(\tau)
        & \text{if } \opConf' = (\Term,\tau),\\
        \mop{\program'}{\mpost}(\pstate')
        & \text{if } \opConf' = (\program',\pstate').
    \end{cases}
\end{align*}
where the last equality follows from the unique transition to $(\Term,\pstate[x\mapsto\ee(\pstate)])$.

Assume as induction hypotheses that the theorem holds for $\program_1$, $\program_2$, and $\program'$ and for every postexpectation in $\mExp$.

\paragraph{$\program = \COMPOSE{\program_1}{\program_2}$.}
By the definition of $\mopsymbol$, we have
\begin{align*}
    \mop{\COMPOSE{\program_1}{\program_2}}{\mpost}(\pstate)
    \eeq
    \mop{\program_1}{\mop{\program_2}{\mpost}}(\pstate).
\end{align*}
Applying the induction hypothesis to $\program_1$ with postexpectation $\mop{\program_2}{\mpost}$ yields
\begin{align*}
    \mop{\COMPOSE{\program_1}{\program_2}}{\mpost}(\pstate)
    \eeq \hspace{-1em}
    \bigsqcup_{\opAct \in \OpAct((\program_1,\pstate))}
    \Minkow_{\substack{
        \step{(\program_1,\pstate)}{\opAct,p}{\opConf'}
    }}
    p \cdot
    \begin{cases}
        \mop{\program_2}{\mpost}(\tau)
        & \text{if } \opConf' = (\Term,\tau),\\
        \mop{\program'}{\mop{\program_2}{\mpost}}(\pstate')
        & \text{if } \opConf' = (\program',\pstate').
    \end{cases}
\end{align*}
By the sequential composition rules in \Cref{fig:smallstep}, every transition of $(\program_1,\pstate)$ to $(\Term,\tau)$ induces a transition of $(\COMPOSE{\program_1}{\program_2},\pstate)$ to $(\program_2,\tau)$ with the same action and probability.
Likewise, every transition of $(\program_1,\pstate)$ to $(\program',\pstate')$ induces a transition of $(\COMPOSE{\program_1}{\program_2},\pstate)$ to $(\COMPOSE{\program'}{\program_2},\pstate')$ with the same action and probability.
Hence,
\begin{align*}
    & \mop{\COMPOSE{\program_1}{\program_2}}{\mpost}(\pstate) \\
    \eeq\ & \bigsqcup_{\opAct \in \OpAct((\COMPOSE{\program_1}{\program_2},\pstate))}
    \Minkow_{\substack{
        \step{(\COMPOSE{\program_1}{\program_2},\pstate)}{\opAct,p}{\opConf'}
    }}
    p \cdot
    \begin{cases}
        \mop{\program_2}{\mpost}(\tau)
        & \text{if } \opConf' = (\program_2,\tau),\\
        \mop{\COMPOSE{\program'}{\program_2}}{\mpost}(\pstate')
        & \text{if } \opConf' = (\COMPOSE{\program'}{\program_2},\pstate').
    \end{cases} \\
    \eeq\ & \bigsqcup_{\opAct \in \OpAct((\COMPOSE{\program_1}{\program_2},\pstate))}
    \Minkow_{\substack{
        \step{(\COMPOSE{\program_1}{\program_2},\pstate)}{\opAct,p}{(\program', \pstate')}
    }}
    p \cdot \mop{\program'}{\mpost}(\pstate')
\end{align*}
This is precisely the required expression, since the successors of a sequential composition are either of the form $(\program_2,\tau)$ or $(\COMPOSE{\program'}{\program_2},\pstate')$, but never $(\Term,\tau)$.

\paragraph{$\program = \GC{\guard_1}{\program_1}{\guard_2}{\program_2}$.}
By the definition of $\mopsymbol$, we have
\begin{align*}
    & \mop{\GC{\guard_1}{\program_1}{\guard_2}{\program_2}}{\mpost}(\pstate) \\
    \eeq\ & \iverson{\guard_1}(\pstate)\cdot \mop{\program_1}{\mpost}(\pstate) \mecup \iverson{\guard_2}(\pstate)\cdot \mop{\program_2}{\mpost}(\pstate).
\end{align*}

If $\pstate\models\guard_1$, then \Cref{fig:smallstep} gives the transition $\step{(\program,\pstate)}{\labalpha,1}{(\program_1,\pstate)}$.
If $\pstate\models\guard_2$, then \Cref{fig:smallstep} gives the transition
$\step{(\program,\pstate)}{\labbeta,1}{(\program_2,\pstate)}$.

Thus, if only $\guard_1$ (or $\guard_2$) is enabled, then $(\program,\pstate)$ has the unique successor $(\program_1,\pstate)$ (or $(\program_2,\pstate)$).
If both guards are enabled, then $(\program,\pstate)$ has two enabled actions, $\labalpha$ and $\labbeta$, leading to $(\program_1,\pstate)$ and $(\program_2,\pstate)$, respectively.
The two transitions are not combined probabilistically: each action induces a probability $1$ transition.
If neither guard is enabled, then there is no transition from $(\program,\pstate)$, which cannot occur due to the assumption of well-formedness of guarded choice in $\pGCL$.

Hence, the enabled actions $\OpAct((\program,\pstate))$ are precisely the labels of the guards that hold in $\pstate$, and for every $\opAct\in\OpAct((\program,\pstate))$ there is exactly one successor with transition probability $1$.

Consequently, we obtain
\begin{align*}
    \mop{\program}{\mpost}(\pstate)
    \eeq &
    \iverson{\guard_1}(\pstate)\cdot
    \mop{\program_1}{\mpost}(\pstate)
    \mecup
    \iverson{\guard_2}(\pstate)\cdot
    \mop{\program_2}{\mpost}(\pstate) \\
    \eeq &
    \bigsqcup_{\opAct \in \OpAct((\program,\pstate))}
    \Minkow_{\substack{
        \step{(\program,\pstate)}{\opAct,p}{\opConf'}
    }}
    p \cdot
    \begin{cases}
        \mpost(\tau)
        & \text{if } \opConf' = (\Term,\tau),\\
        \mop{\program'}{\mpost}(\pstate')
        & \text{if } \opConf' = (\program',\pstate').
    \end{cases}
\end{align*}

\paragraph{$\program = \PCHOICE{\program_1}{\ps}{\program_2}$.}
By the definition of $\mopsymbol$, we have
\begin{align*}
    \mop{\PCHOICE{\program_1}{\ps}{\program_2}}{\mpost}(\pstate)
    \eeq
    \ps(\pstate)\cdot\mop{\program_1}{\mpost}(\pstate)
    \minkow
    (1-\ps(\pstate))\cdot\mop{\program_2}{\mpost}(\pstate).
\end{align*}
By \Cref{fig:smallstep}, the configuration $(\PCHOICE{\program_1}{\ps}{\program_2},\pstate)$ has the transitions
\[
\step{(\PCHOICE{\program_1}{\ps}{\program_2},\pstate)}{\labtau,\ps(\pstate)}{(\program_1,\pstate)}
\]
and
\[
\step{(\PCHOICE{\program_1}{\ps}{\program_2},\pstate)}{\labtau,1-\ps(\pstate)}{(\program_2,\pstate)}
\]
when $\program_1\neq\program_2$, while for $\program_1=\program_2$ there is the single transition with probability $1$ given in \Cref{fig:smallstep}.
Hence, in either case,
\begin{align*}
    &\mop{\program}{\mpost}(\pstate) \\
    \eeq\ & p(\pstate)\cdot\mop{\program_1}{\mpost}(\pstate) \minkow (1-p)(\pstate)\cdot\mop{\program_2}{\mpost}(\pstate) \\
    \eeq\ &
    \bigsqcup_{\opAct \in \OpAct((\program,\pstate))}
    \Minkow_{\substack{
        \step{(\program,\pstate)}{\opAct,p'}{\opConf'}
    }}
    p' \cdot
    \begin{cases}
        \mpost(\tau)
        & \text{if } \opConf' = (\Term,\tau),\\
        \mop{\program'}{\mpost}(\pstate')
        & \text{if } \opConf' = (\program',\pstate').
    \end{cases}
\end{align*}

\paragraph{$\program = \WHILEDO{\guard}{\program'}$.}
By the definition of $\mopsymbol$, we have
\[
    \mop{\WHILEDO{\guard}{\program'}}{\mpost}
    \eeq
    \lfp \phiMop{}{\mpost}
\]
for the characteristic function
\[
    \phiMop{}{\mpost}(X)
    \eeq
    \iverson{\neg\guard}\cdot\mpost \minkow \iverson{\guard}\cdot\mop{\program'}{X}.
\]
Since a least fixpoint is in particular a fixpoint, we obtain
\[
    \mop{\WHILEDO{\guard}{\program'}}{\mpost}
    \eeq
    \iverson{\neg\guard}\cdot\mpost \minkow \iverson{\guard} \cdot \mop{\program'}{\mop{\WHILEDO{\guard}{\program'}}{\mpost}}.
\]
Let $\pstate\in\States$ be arbitrary.
If $\pstate\models\neg\guard$, then \Cref{fig:smallstep} gives the unique transition
\[
    \step{(\WHILEDO{\guard}{\program'},\pstate)}{\labtau,1}{(\Term,\pstate)}.
\]
Therefore,
\begin{align*}
    & \mop{\WHILEDO{\guard}{\program'}}{\mpost}(\pstate) \\
    \eeq\ & \mpost(\pstate) \\
    \eeq\ & \bigsqcup_{\opAct \in \OpAct((\WHILEDO{\guard}{\program'},\pstate))}
    \Minkow_{\substack{
        \step{(\WHILEDO{\guard}{\program'},\pstate)}{\opAct,p}{\opConf'}
    }}
    p \cdot
    \begin{cases}
        \mpost(\tau)
        & \text{if } \opConf' = (\Term,\tau),\\
        \mop{\program''}{\mpost}(\pstate')
        & \text{if } \opConf' = (\program'',\pstate').
    \end{cases}
\end{align*}
If $\pstate\models\guard$, then \Cref{fig:smallstep} gives the unique transition
\[
    \step{(\WHILEDO{\guard}{\program'},\pstate)}{\labtau,1}{(\COMPOSE{\program'}{\WHILEDO{\guard}{\program'}},\pstate)}.
\]
Further, by the fixpoint equation above, we have
\begin{align*}
    \mop{\WHILEDO{\guard}{\program'}}{\mpost}(\pstate)
    \eeq
    \mop{\program'}{\mop{\WHILEDO{\guard}{\program'}}{\mpost}}(\pstate).
\end{align*}
By the induction hypothesis applied to $\program'$ with postexpectation $\mop{\WHILEDO{\guard}{\program'}}{\mpost}$, we obtain
\begin{align*}
    & \mop{\WHILEDO{\guard}{\program'}}{\mpost}(\pstate) \\
    \eeq\ & \bigsqcup_{\opAct \in \OpAct((\program',\pstate))}
    \Minkow_{\substack{
        \step{(\program',\pstate)}{\opAct,p}{\opConf'}
    }}
    p \cdot
    \begin{cases}
        \mop{\WHILEDO{\guard}{\program'}}{\mpost}(\tau)
        & \text{if } \opConf' = (\Term,\tau),\\
        \mop{\program''}{\mop{\WHILEDO{\guard}{\program'}}{\mpost}}(\pstate')
        & \text{if } \opConf' = (\program'',\pstate').
    \end{cases}
\end{align*}
By the sequential composition rules in \Cref{fig:smallstep}, these transitions correspond exactly to the transitions from $(\COMPOSE{\program'}{\WHILEDO{\guard}{\program'}},\pstate)$, and therefore
\begin{align*}
    & \mop{\WHILEDO{\guard}{\program'}}{\mpost}(\pstate) \\
    \eeq\ & \bigsqcup_{\opAct \in \OpAct((\WHILEDO{\guard}{\program'},\pstate))}
    \Minkow_{\substack{
        \step{(\WHILEDO{\guard}{\program'},\pstate)}{\opAct,p}{\opConf'}
    }}
    p \cdot
    \begin{cases}
        \mpost(\tau)
        & \text{if } \opConf' = (\Term,\tau),\\
        \mop{\program''}{\mpost}(\pstate')
        & \text{if } \opConf' = (\program'',\pstate').
    \end{cases}
\end{align*}
Thus, the required equality holds for the while case, and the induction is complete.
\end{proof}

Let $\program \in \pGCL$ and $\mpost \in \mExp$.
We define
\[
    \omop{\program}{\mpost} \eeq \mylambda \pstate (\lfp \BellmanPareto{\ProgR{\mpost}}{\ProgMDP})(\program,\pstate).
\]

\begin{theorem}%
\label{theo:omop-compose}
For all $\program_1,\program_2 \in \pGCL$ and $\mpost \in \mExp$:
\[
    \omop{\program_1}{\omop{\program_2}{\mpost}}
    \sqsubseteq
    \omop{\COMPOSE{\program_1}{\program_2}}{\mpost}.
\]
\end{theorem}
\begin{proof}
See also \cite[Proof of Lemma 9]{batz2024jp}.
Fix $\program_2 \in \pGCL$ and $\mpost \in \mExp$.
To enhance readability, we write
\[
    \Phi
    \eeq
    \BellmanPareto{
        \ProgR{\omop{\program_2}{\mpost}}
    }{\ProgMDP}
    \quad \text{and} \quad
    \Psi
    \eeq
    \BellmanPareto{\ProgR{\mpost}}{\ProgMDP}.
\]
By the definition of $\omopsymbol$, for every $\program \in \pGCL$ and $\pstate \in \States$,
$
    (\lfp \Psi)(\program,\pstate)
    =
    \omop{\program}{\mpost}(\pstate).
$
Since $\lfp \Psi$ is a fixpoint of $\Psi$, we furthermore have
$
    \Psi(\lfp \Psi) = \lfp \Psi.
$

By the definition of $\omopsymbol$,
\[
    \omop{\program_1}{\omop{\program_2}{\mpost}}(\pstate)
    \eeq
    (\lfp \Phi)(\program_1,\pstate).
\]
Since $\Phi$ is $\omega$-continuous, \Cref{theo:Kleene} yields
\[
    (\lfp \Phi)(\program_1,\pstate)
    \eeq
    \left(
        \bigsqcup_{k \in \Nats}
        \Phi^k(\set{0^n})
    \right)(\program_1,\pstate).
\]
Thus, it suffices to show that for all $\program \in \pGCL$ and $\pstate \in \States$,
\[
    \Phi^k(\set{0^n})(\program,\pstate)
    \subseteq
    (\lfp \Psi)(\COMPOSE{\program}{\program_2},\pstate).
    \tag{$*$}
\]
We do so by induction on $k$.
For $k=0$, we immediately get
\[
    \Phi^0(\set{0^n})(\program,\pstate)
    \eeq
    \set{0^n}
    \subseteq
    (\lfp \Psi)(\COMPOSE{\program}{\program_2},\pstate).
\]

For the induction step, assume that $(*)$ holds for $k$.
Fix arbitrary $\program \in \pGCL$ and $\pstate \in \States$.
Since $(\program,\pstate)$ is not a target configuration, the
non-target clause of the Bellman operator gives
\begin{align*}
    \Phi^{k+1}(\set{0^n})(\program,\pstate)
    &\eeq
    \Phi(\Phi^k(\set{0^n}))(\program,\pstate)
    \\
    &\eeq
    \bigsqcup_{\opAct \in \OpAct((\program,\pstate))}
    \Minkow_{\substack{
        \step{(\program,\pstate)}
        {\opAct,p}
        {(\program',\pstate')}
    }}
    p \cdot
    \Phi^k(\set{0^n})(\program',\pstate').
\end{align*}

We distinguish within each Minkowski sum whether a transition of
$\program$ terminates or leaves a non-terminal residual program.
For every non-terminal successor
$(\program',\pstate')$ with $\program' \neq \Term$, the induction
hypothesis gives
\[
    \Phi^k(\set{0^n})(\program',\pstate')
    \subseteq
    (\lfp \Psi)(\COMPOSE{\program'}{\program_2},\pstate').
\]
For every terminal successor $(\Term,\pstate')$, we have
\[
    \Phi^k(\set{0^n})(\Term,\pstate')
    \subseteq %
    \ProgR{\omop{\program_2}{\mpost}}(\Term,\pstate')
    \eeq
    \omop{\program_2}{\mpost}(\pstate')
    \eeq
    (\lfp \Psi)(\program_2,\pstate').
\]
Indeed, for $k=0$ this follows because $\set{0^n}$ is the least
element, while for $k>0$ it follows directly from the target clause
of $\Phi$, which assigns
$\ProgR{\omop{\program_2}{\mpost}}(\Term,\pstate')$ to the terminal
configuration.

Using monotonicity of scalar multiplication, Minkowski sums, and
$\bigsqcup$, we therefore obtain
\begin{align*}
    &\Phi^{k+1}(\set{0^n})(\program,\pstate)
    \\
    &\subseteq
    \bigsqcup_{\opAct \in \OpAct((\program,\pstate))}
    \Bigg(
    \Minkow_{\substack{
        \step{(\program,\pstate)}
        {\opAct,p}
        {(\program',\pstate')}\\
        \program' \neq \Term
    }}
    p \cdot
    (\lfp \Psi)(\COMPOSE{\program'}{\program_2},\pstate')
    \minkow
    \Minkow_{\substack{
        \step{(\program,\pstate)}
        {\opAct,p}
        {(\Term,\pstate')}
    }}
    p \cdot
    (\lfp \Psi)(\program_2,\pstate')
    \Bigg).
    \tag{IH}
\end{align*}

The two kinds of successors in this expression correspond exactly to
the two small-step rules for sequential composition in
\Cref{fig:smallstep}. Namely, a transition
\[
    \step{(\program,\pstate)}
    {\opAct,p}
    {(\program',\pstate')}
    \qquad
    \text{with } \program' \neq \Term
\]
induces the transition
\[
    \step{(\COMPOSE{\program}{\program_2},\pstate)}
    {\opAct,p}
    {(\COMPOSE{\program'}{\program_2},\pstate')},
\]
whereas a transition
\[
    \step{(\program,\pstate)}
    {\opAct,p}
    {(\Term,\pstate')}
\]
induces
\[
    \step{(\COMPOSE{\program}{\program_2},\pstate)}
    {\opAct,p}
    {(\program_2,\pstate')}.
\]
Conversely, every transition from
$(\COMPOSE{\program}{\program_2},\pstate)$ is obtained by one of these
two rules. Hence,
\begin{align*}
    &\bigsqcup_{\opAct \in \OpAct((\program,\pstate))}
    \Bigg(
    \Minkow_{\substack{
        \step{(\program,\pstate)}
        {\opAct,p}
        {(\program',\pstate')}\\
        \program' \neq \Term
    }}
    p \cdot
    (\lfp \Psi)(\COMPOSE{\program'}{\program_2},\pstate')
    \minkow
    \Minkow_{\substack{
        \step{(\program,\pstate)}
        {\opAct,p}
        {(\Term,\pstate')}
    }}
    p \cdot
    (\lfp \Psi)(\program_2,\pstate')
    \Bigg)
    \\
    &\eeq
    \bigsqcup_{\opAct \in
        \OpAct((\COMPOSE{\program}{\program_2},\pstate))}
    \Minkow_{\substack{
        \step{(\COMPOSE{\program}{\program_2},\pstate)}
        {\opAct,p}
        {(\program',\pstate')}
    }}
    p \cdot
    (\lfp \Psi)(\program',\pstate')
    \\
    &\eeq
    \Psi(\lfp \Psi)(\COMPOSE{\program}{\program_2},\pstate)
    \\
    &\eeq
    (\lfp \Psi)(\COMPOSE{\program}{\program_2},\pstate).
\end{align*}
This proves the induction step.
Consequently, for every $k \in \Nats$,
\[
    \Phi^k(\set{0^n})(\program_1,\pstate)
    \subseteq
    (\lfp \Psi)(\COMPOSE{\program_1}{\program_2},\pstate)
    \eeq
    \omop{\COMPOSE{\program_1}{\program_2}}{\mpost}(\pstate).
\]
Taking the supremum over all $k \in \Nats$ therefore yields
\begin{align*}
    \omop{\program_1}{\omop{\program_2}{\mpost}}(\pstate)
    &\eeq
    \left(
        \bigsqcup_{k \in \Nats}
        \Phi^k(\set{0^n})
    \right)(\program_1,\pstate)
    \\
    &\subseteq
    \omop{\COMPOSE{\program_1}{\program_2}}{\mpost}(\pstate).
\end{align*}
Since $\pstate$ was arbitrary, we conclude
\[
    \omop{\program_1}{\omop{\program_2}{\mpost}}
    \sqsubseteq
    \omop{\COMPOSE{\program_1}{\program_2}}{\mpost}.
\]
\end{proof}

Finally, we are in position to prove \Cref{theo:sound-mop}, which was:

\soundmop*

\begin{proof}
    Recall that
    \[
        \omop{\program}{\mpost}
        \eeq
        \mylambda \pstate
        \left(
            \lfp
            \BellmanPareto{\ProgR{\mpost}}{\ProgMDP}
        \right)(\program,\pstate).
    \]
    We prove
    $
        \mop{\program}{\mpost}
        =
        \omop{\program}{\mpost}
    $
    by showing both inclusions.

    \paragraph{Proof of \ $\sqsupseteq$}
    We define $\mv \in \MV$ by
    \[
        \mv(\opState)
        \eeq
        \begin{cases}
            \mpost(\pstate)
            & \text{if } \opState = (\Term,\pstate),\\
            \mop{\program'}{\mpost}(\pstate)
            & \text{if } \opState = (\program',\pstate).
        \end{cases}
    \]

    We show that $\mv$ is a fixpoint of $\BellmanPareto{\ProgR{\mpost}}{\ProgMDP}$.

    First, consider a target configuration $(\Term,\pstate)$.
    By the target clause of the Bellman operator and
    \Cref{def:rew-mexp},
    \[
        \BellmanPareto{\ProgR{\mpost}}{\ProgMDP}(\mv)(\Term,\pstate)
        \eeq
        \ProgR{\mpost}(\Term,\pstate)
        \eeq
        \mpost(\pstate)
        \eeq
        \mv(\Term,\pstate).
    \]

    Now consider a non-target configuration $(\program',\pstate)$.
    By \Cref{def:multi-bellman},
    \begin{align*}
        \BellmanPareto{\ProgR{\mpost}}{\ProgMDP}(\mv)(\program',\pstate)
        &\eeq
        \bigsqcup_{\opAct \in
            \OpAct((\program',\pstate))}
        \Minkow_{\substack{
            \step{(\program',\pstate)}
            {\opAct,p}
            {\opConf'}
        }}
        p \cdot
        \begin{cases}
            \mpost(\tau)
            & \text{if } \opConf' = (\Term,\tau),\\
            \mop{\program''}{\mpost}(\pstate')
            & \text{if }
            \opConf' = (\program'',\pstate').
        \end{cases}
        \\
        &\eeq
        \mop{\program'}{\mpost}(\pstate)
        \tag{\Cref{theo:help-op}}
        \\
        &\eeq
        \mv(\program',\pstate).
    \end{align*}
    Thus, $\mv$ is a fixpoint of $\BellmanPareto{\ProgR{\mpost}}{\ProgMDP}$.
    In particular, $\BellmanPareto{\ProgR{\mpost}}{\ProgMDP}(\mv) \MVleq \mv$.
    By Park induction,
    \[
        \lfp \BellmanPareto{\ProgR{\mpost}}{\ProgMDP} \MVleq \mv.
    \]
    Hence, for every $\pstate \in \States$,
    \begin{align*}
        \omop{\program}{\mpost}(\pstate)
        &\eeq
        (\lfp \BellmanPareto{\ProgR{\mpost}}{\ProgMDP})(\program,\pstate)
        \\
        &\subseteq
        \mv(\program,\pstate)
        \\
        &\eeq
        \mop{\program}{\mpost}(\pstate).
    \end{align*}
    Therefore,
    $
        \omop{\program}{\mpost}
        \sqsubseteq
        \mop{\program}{\mpost}.
    $

    \paragraph{Proof of \ $\sqsubseteq$}
    We prove the claim by structural induction on $\program$.
    Let
    \[
        \Psi
        \eeq
        \BellmanPareto{\ProgR{\mpost}}{\ProgMDP}.
    \]
    Since $\lfp \Psi$ is a fixpoint of $\Psi$, for every $\program' \in \pGCL$ and $\pstate \in \States$,
    \begin{align}
        \omop{\program'}{\mpost}(\pstate)
        &\eeq
        \bigsqcup_{\opAct \in
            \OpAct((\program',\pstate))}
        \Minkow_{\substack{
            \step{(\program',\pstate)}
            {\opAct,p}
            {\opConf'}
        }}
        p \cdot
        \begin{cases}
            \mpost(\tau)
            & \text{if } \opConf' = (\Term,\tau),\\
            \omop{\program''}{\mpost}(\pstate')
            & \text{if }
            \opConf' = (\program'',\pstate').
        \end{cases}
        \label{eq:omop-unfold}
    \end{align}
    Indeed, this is precisely the non-target clause of
    \Cref{def:multi-bellman} applied to the fixpoint $\lfp \Psi$.

    \paragraph{$\program = \SKIP$.}
    By \Cref{fig:smallstep}, $(\SKIP,\pstate)$ has exactly one
    successor, namely $(\Term,\pstate)$, reached with probability $1$.
    Hence, by \Cref{eq:omop-unfold},
    $
        \omop{\SKIP}{\mpost}(\pstate)
        \eeq
        \mpost(\pstate).
    $
    Therefore,
    \begin{align*}
        \mop{\SKIP}{\mpost}(\pstate)
        \eeq
        \mpost(\pstate)
        \tag{\Cref{def:mop}}
        \eeq
        \omop{\SKIP}{\mpost}(\pstate).
    \end{align*}

    \paragraph{$\program = \ASSIGN{x}{\ee}$.}

    By \Cref{fig:smallstep},
    $(\ASSIGN{x}{\ee},\pstate)$ has exactly one successor,
    $
        (\Term,\pstate\statesubst{x}{\ee(\pstate)}),
    $
    reached with probability $1$.
    Thus,
    \[
        \omop{\ASSIGN{x}{\ee}}{\mpost}(\pstate)
        \eeq
        \mpost(
            \pstate\statesubst{x}{\ee(\pstate)}
        ).
    \]
    Consequently,
    \begin{align*}
        \mop{\ASSIGN{x}{\ee}}{\mpost}(\pstate)
        &\eeq
        \mpost\subst{x}{\ee}(\pstate)
        \tag{\Cref{def:mop}}
        \\
        &\eeq
        \mpost(
            \pstate\statesubst{x}{\ee(\pstate)}
        )
        \\
        &\eeq
        \omop{\ASSIGN{x}{\ee}}{\mpost}(\pstate).
    \end{align*}

    As induction hypothesis, assume that for every subprogram $\program'$ of $\program$ and every $\mpost' \in \mExp$,
    \[
        \mop{\program'}{\mpost'}
        \sqsubseteq
        \omop{\program'}{\mpost'}.
        \tag{IH}
    \]

    \paragraph{$\program = \COMPOSE{\program_1}{\program_2}$}
    \begin{align*}
        \mop{\COMPOSE{\program_1}{\program_2}}{\mpost}
        \eeq\ & \mop{\program_1}{\mop{\program_2}{\mpost}} \tag{\Cref{def:mop}}\\
        \sqsubseteq\ & \mop{\program_1}{\omop{\program_2}{\mpost}} \tag{IH and \Cref{theo:mop-well-cont-monotone}} \\
        \sqsubseteq\ & \omop{\program_1}{\omop{\program_2}{\mpost}} \tag{IH} \\
        \sqsubseteq\ & \omop{\COMPOSE{\program_1}{\program_2}}{\mpost}. \tag{\Cref{theo:omop-compose}}
    \end{align*}

    \paragraph{$\program =
    \GC{\guard_1}{\program_1}{\guard_2}{\program_2}$.}

    The operational rules for guarded choice yield one action for
    each enabled guard. Thus, \Cref{eq:omop-unfold} gives
    \[
        \omop{\program}{\mpost}
        \eeq
        \iverson{\guard_1}
        \cdot
        \omop{\program_1}{\mpost}
        \mmecup
        \iverson{\guard_2}
        \cdot
        \omop{\program_2}{\mpost}.
    \]
    Hence,
    \begin{align*}
        \mop{\program}{\mpost}
        &\eeq
        \iverson{\guard_1}
        \cdot
        \mop{\program_1}{\mpost}
        \mmecup
        \iverson{\guard_2}
        \cdot
        \mop{\program_2}{\mpost}
        \tag{\Cref{def:mop}}
        \\
        &\sqsubseteq
        \iverson{\guard_1}
        \cdot
        \omop{\program_1}{\mpost}
        \mmecup
        \iverson{\guard_2}
        \cdot
        \omop{\program_2}{\mpost}
        \tag{IH}
        \\
        &\eeq
        \omop{\program}{\mpost}.
    \end{align*}

    \paragraph{$\program = \PCHOICE{\program_1}{\pexpr}{\program_2}$.}
    First assume $\program_1 \neq \program_2$.
    By \Cref{fig:smallstep}, both successors belong to the same action and are reached with probabilities $\pexpr(\pstate)$ and $1-\pexpr(\pstate)$.
    Hence, \Cref{eq:omop-unfold} yields
    \[
        \omop{\program}{\mpost}(\pstate)
        \eeq
        \pexpr(\pstate)
        \cdot
        \omop{\program_1}{\mpost}(\pstate)
        \minkow
        (1-\pexpr(\pstate))
        \cdot
        \omop{\program_2}{\mpost}(\pstate).
    \]
    Therefore,
    \begin{align*}
        \mop{\program}{\mpost}(\pstate)
        &\eeq \pexpr(\pstate) \cdot \mop{\program_1}{\mpost}(\pstate) \minkow (1-\pexpr(\pstate)) \cdot \mop{\program_2}{\mpost}(\pstate)\tag{\Cref{def:mop}}\\
        &\subseteq \pexpr(\pstate) \cdot \omop{\program_1}{\mpost}(\pstate) \minkow (1-\pexpr(\pstate)) \cdot \omop{\program_2}{\mpost}(\pstate) \tag{IH} \\
        &\eeq \omop{\program}{\mpost}(\pstate).
    \end{align*}

    If $\program_1=\program_2$, the operational semantics instead gives a single transition with probability $1$. In this case,
    \Cref{eq:omop-unfold} yields
    \[
        \omop{\program}{\mpost}(\pstate)
        \eeq
        \omop{\program_1}{\mpost}(\pstate).
    \]
    Moreover, since elements of $\HoarePowerDom$ are convex,
    $
        \pexpr(\pstate) \cdot A \minkow (1-\pexpr(\pstate)) \cdot A \eeq A
    $
    for every $A \in \HoarePowerDom$.
    Thus, the desired inclusion again follows from the induction hypothesis.

    \paragraph{$\program = \WHILEDO{\guard}{\program'}$.}
    By the definition of $\mopsymbol$,
    $
        \mop{\program}{\mpost}
        \eeq
        \lfp \phiMop{}{\mpost}
    $
    for the characteristic function
    \[
        \phiMop{}{\mpost}(X)
        \eeq
        \iverson{\neg\guard} \cdot \mpost
        \mminkow
        \iverson{\guard} \cdot
        \mop{\program'}{X}.
    \]
    We prove
    \[
        \phiMop{}{\mpost}
        \left(
            \omop{\program}{\mpost}
        \right)
        \sqsubseteq
        \omop{\program}{\mpost}.
    \]
    Park induction then gives
    \[
        \mop{\program}{\mpost}
        \eeq
        \lfp \phiMop{}{\mpost}
        \sqsubseteq
        \omop{\program}{\mpost}.
    \]
    By the induction hypothesis for $\program'$,
    \begin{align*}
        & \phiMop{}{\mpost}\left(\omop{\program}{\mpost}\right) \\
        =\ & \iverson{\neg\guard} \cdot \mpost \mminkow \iverson{\guard} \cdot \mop{\program'}{\omop{\program}{\mpost}} \\
        \sqsubseteq\ & \iverson{\neg\guard} \cdot \mpost \mminkow \iverson{\guard} \cdot \omop{\program'}{\omop{\program}{\mpost}} \tag{IH} \\
        \sqsubseteq\ & \iverson{\neg\guard} \cdot \mpost \mminkow \iverson{\guard} \cdot \omop{\COMPOSE{\program'}{\program}}{\mpost}. \tag{\Cref{theo:omop-compose}}
    \end{align*}

    It remains to compare the last expression with
    $\omop{\program}{\mpost}$.
    We do this pointwise.
    If $\pstate \models \neg\guard$, the small-step semantics gives
    $
        \step{(\program,\pstate)}
        {\labtau,1}
        {(\Term,\pstate)}.
    $
    Therefore, by \Cref{eq:omop-unfold},
    $
        \omop{\program}{\mpost}(\pstate)
        \eeq
        \mpost(\pstate).
    $

    If $\pstate \models \guard$, the small-step semantics gives
    $
        \step{(\program,\pstate)}
        {\labtau,1}
        {(\COMPOSE{\program'}{\program},\pstate)}.
    $
    Hence,
    $
        \omop{\program}{\mpost}(\pstate)
        \eeq
        \omop{
            \COMPOSE{\program'}{\program}
        }{\mpost}(\pstate).
    $

    Consequently,
    \[
        \iverson{\neg\guard} \cdot \mpost
        \mminkow
        \iverson{\guard} \cdot
        \omop{
            \COMPOSE{\program'}{\program}
        }{\mpost}
        \eeq
        \omop{\program}{\mpost}.
    \]
    Thus,
    \[
        \phiMop{}{\mpost}
        \left(
            \omop{\program}{\mpost}
        \right)
        \sqsubseteq
        \omop{\program}{\mpost},
    \]
    and Park induction proves the loop case.

    This concludes the structural induction and therefore proves
    \[
        \mop{\program}{\mpost}
        \sqsubseteq
        \omop{\program}{\mpost}.
    \]

    Together with the converse inclusion and antisymmetry, we obtain
    \[
        \mop{\program}{\mpost}
        \eeq
        \omop{\program}{\mpost}
        \eeq
        \mylambda \pstate
        \left(
            \lfp
            \BellmanPareto{\ProgR{\mpost}}{\ProgMDP}
        \right)(\program,\pstate).
    \]
\end{proof}

\subsection{Operational Soundness of mop}

From \Cref{theo:sound-mop}, we know that \mopsymbol computes the same as the Bellman operator on its operational MDP.
It remains to show that the achievable points on the MDP are equivalent to the achievable points on program level.
To this end, we first note that each memoryless mixed scheduler resolving nondeterminism in the operational MDP of \pGCL can be lifted to a mixed determinization of a nondeterministic program, in that the produced weakest preexpectations is at least the expected reward.
For the remainder of this section, let $\program \in \pGCL$ and $\post \in \Exp^n$.

\begin{restatable}[]{lemma}{schedleqdeterm}
\label[lemma]{theo:sched-leq-determ}
    For all $\pstate \in \States$, $\AchMemless{\ProgMDP}{\ProgR{\post}}{(\program,\pstate)} \subseteq \AchExp{\program}{\post}{\pstate}$.
\end{restatable}
\begin{proof}
    Let $\sched$ be a mixed memoryless scheduler on $\ProgMDP$.
    We prove that there exists a mixed determinization $\program_\sched$
    such that
    \[
        \ExpRew{\ProgMDP}{\ProgR{\post}}{(\program,\pstate)}{\sched}
        \leq \wpVec{\program_\sched}{\post}(\pstate).
    \]
    where the inequality holds componentwise.

    We construct $\program_\sched$ to follow the choices made by $\sched$.
    Since $\sched$ is mixed memoryless, it induces a finite distribution
    over pure memoryless schedulers.
    For each pure scheduler $\sched_i$ in its support, we construct
    a pure determinization $\program_i$ by strengthening the guards
    of nondeterministic choices according to $\sched_i$.
    The construction is largely straightforward.
    Care is needed under sequential composition, where the scheduler
    may resolve a choice at the configuration of the enclosing program
    rather than at the guarded choice in isolation.
    We therefore need to take the surrounding program context
    into account.
    
    The determinization rules ensure that $\program_i\determ\program$ for every such $\program_i$.
    By induction on the program structure, we obtain that and
    \begin{align}
        \ExpRew{\ProgMDP}{\ProgR{\post}}{(\program,\pstate)}{\sched_i}
        \leq \wpVec{\program_i}{\post}(\pstate).
        \label{eq:pure-eq}
    \end{align}

    Let $\determDist(\program_i)$ be the probability assigned to $\sched_i$ by $\sched$.
    Then
    \[
        \program_\sched
        \coloneqq
        \PCHOICEFINITE{\program_i}{}{\determDist(\program_i)}{}
    \]
    is a mixed determinization of $\program$.
    Using \Cref{eq:pure-eq} for each $\program_i$, we get
    \begin{align*}
        \wpVec{\program_\sched}{\post}(\pstate)
        &=
        \sum_i \determDist(\program_i)
        \cdot \wpVec{\program_i}{\post}(\pstate)
        \\
        &\geq
        \sum_i \determDist(\program_i)
        \cdot
        \ExpRew{\ProgMDP}{\ProgR{\post}}
        {(\program,\pstate)}{\sched_i}
        \\
        &=
        \ExpRew{\ProgMDP}{\ProgR{\post}}
        {(\program,\pstate)}{\sched}.
    \end{align*}
    Since achievable sets are downward closed, the claimed inclusion follows.
\end{proof}

Note that the converse direction does not hold.
This is due to the operational semantics (see \Cref{def:opsem-pgcl}), where in a probabilistic choice, if the subprograms are equal, the branches are merged in the MDP.
Consequently, nondeterministic choices in these branches must be resolved equally in the MDP, but can differ in determinizations on program level.
The schedulers would need infinite support to cover this, which we do not allow.

We continue by showing that the achievable points of the program are included in \mopsymbol, a property that we will need later on.

\begin{restatable}{lemma}{mopincludesach}
\label[lemma]{theo:mop-includes-ach}
   For all $\pstate \in \States$, $\AchExp{\program}{\post}{\pstate} \subseteq \mop{\program}{\dwcset{\post}}(\pstate)$.
\end{restatable}
\begin{proof}
    Let $\program' \determmixed \program$.
    We argue that $\wpVec{\program'}{\post}(\pstate) \in \mop{\program}{\dwcset{\post}}(\pstate)$.
    First, since $\program'$ is fully probabilistic, i.e., there are no nondeterministic choices, \mopsymbol conservatively extends \wpsymbol:
    \[
        \mop{\program'}{\dwcset{\post}}(\pstate) = \dwc{\wpVec{\program'}{\post}(\pstate)}
    \]
    by straightforward induction on $\program'$. %
    Second, $\program'$ is a determinization of $\program$, i.e., it only restricts nondeterministic choices and we have
    \[
        \mop{\program'}{\dwcset{\post}}(\pstate) \, \subseteq \, \mop{\program}{\dwcset{\post}}(\pstate),
    \]
    again by induction on the program structure.
    So, we have that $\wpVec{\program'}{\post}(\pstate) \in \mop{\program}{\dwcset{\post}}(\pstate)$ and we get
    \begin{align*}
        & \AchExp{\program}{\post}{\pstate} \\
        \eeq& \dwcsetcomp{\wpVec{\program'}{\post}(\pstate) }{ \program' \determmixed \program} \\
        \subseteq\ & \mop{\program}{\dwcset{\post}}(\pstate). \\
    \end{align*}
\end{proof}

Combining the program-level and MDP-level correspondences established in the preceding sections finally yields the central semantic result of the paper:
\mopsymbol exactly characterizes the closure of the achievable set.
We now state this result formally, as introduced in \Cref{sec:sound}.

\mopispareto*

\begin{proof} %
    From \Cref{theo:mop-includes-ach},
    \[
        \AchExp{\program}{\post}{\pstate} \subseteq \mop{\program}{\dwcset{\post}}(\pstate),
    \]
    which by Scott closedness of \mopsymbol yields that
    \[
        \cl{\AchExp{\program}{\post}{\pstate}} \subseteq \mop{\program}{\dwcset{\post}}(\pstate).
    \]
    
    For the reverse inclusion, note that we have
    \begin{align*}
        & \mop{\program}{\dwcset{\post}}(\pstate) \\
        =\ & \left(\lfp \BellmanPareto{\ProgR{\dwcset{\post}}}{\ProgMDP}\right)((\program,\pstate)) \tag{\Cref{theo:sound-mop}} \\
        =\ & \cl{\Ach{\ProgMDP}{\ProgR{\post}}{(\program,\pstate)}} \tag{\Cref{theo:bellman-is-pareto}} \\
        =\ & \cl{\AchMemless{\ProgMDP}{\ProgR{\post}}{(\program,\pstate)}} \tag{\Cref{theo:memless-suff}} \\
    \end{align*}

    From \Cref{theo:sched-leq-determ}, we get that
    \[
        \AchMemless{\ProgMDP}{\ProgR{\post}}{(\program,\pstate)} \subseteq \AchExp{\program}{\post}{\pstate},
    \]
    so since \mopsymbol is already Scott closed, this yields together with the above
    \[
        \mop{\program}{\dwcset{\post}}(\pstate) \subseteq \cl{\AchExp{\program}{\post}{\pstate}}.
    \]
    Both inclusions yield the desired equality.
\end{proof}

As a direct consequence of \Cref{theo:mop-is-pareto}, we get the following correspondence between \mopsymbol and pure determinizations, that was also already stated in the main paper.

\mopwppure*

\begin{proof}
    We have that
    \begin{align*}
        & \mop{\program}{\dwcset{\post}}(\pstate) \\
        =\ & \cl{\AchExp{\program}{\post}{\pstate}} \tag{\Cref{theo:mop-is-pareto}} \\
        =\ & \cl{\dwcset{\wpVec{\program'}{\post}(\pstate) \mid \program' \determmixed \program}} \tag{\Cref{def:achievable-exp}} \\
        =\ & \cl{\set{ \wpVec{\program'}{\post}(\pstate) \mid \program' \determmixed \program }}. \tag{Def.\ of $\cl{\cdot}$}
    \end{align*}
    Thus it suffices to show that 
    \[
        \convsetcomp{ \wpVec{\program'}{\post}(\pstate) }{ \program' \determ \program } \eeq \setcomp{ \wpVec{\program'}{\post}(\pstate)}{\program' \determmixed \program },
    \]
    since the Scott closure of the right-hand side is equivalent to $\mop{\program}{\dwcset{\post}}(\pstate)$ by the above equations.

    For $\subseteq$, let $x \in \convsetcomp{ \wpVec{\program'}{\post}(\pstate)}{ \program' \determ \program }$.
    Then, there exist $\weightc_1, \dots, \weightc_m \in [0,1]$ with $\sum_i \weightc_i = 1$ and $\program_1, \dots, \program_m \determ \program$ such that
    \[
        x = \sum_i \weightc_i \cdot \wpVec{\program_i}{\post}(\pstate).
    \]
    Let $\program''$ be the program executing $\program_i$ with probability $\weightc_i$.
    Then, $\program'' \determmixed \program$ and $\wpVec{\program''}{\post}(\pstate) = x$, proving $\subseteq$.

    For $\supseteq$, let $x = \wpVec{\program'}{\post}(\pstate)$ for some $\program' \determmixed \program$.
    Then, $\program'$ by definition executes programs $\program_i$ with probability $\weightc_i$ for $i \in \set{1,\dots,m}$.
    Thus, $x$ can be written as
    \[
        x = \sum_i \weightc_i \cdot \wpVec{\program_i}{\post}(\pstate),
    \]
    and therefore is an element of $\convsetcomp{ \wpVec{\program'}{\post}(\pstate) }{ \program' \determ \program }$.
\end{proof}

From this, we immediately get the following two corollaries, proving that \mopsymbol is a conservative extension of \wpsymbol.

\begin{corollary}
\label{theo:mop-wp-determ}
    If $\program$ is deterministic,
    $
        \mop{\program}{\dwcset{\post}} = \dwcset{ \wpVec{\program}{\post} }.
    $
\end{corollary}
\begin{proof}
    Follows immediately from \Cref{theo:mop-wp-pure} since a deterministic program $\program$ only has one determinization, which is $\program$.
\end{proof}
\begin{corollary}
\label{theo:mop-wp-single}
    If $\post$ is single-objective,
    $
        \mop{\program}{\dwcset{\post}} = \dwcset{ \wp{\program}{\post} }.
    $
\end{corollary}
\begin{proof}
    Since $\post$ is single objective, we have by \Cref{theo:wp-sup-determ} that
    \[
        \sup \set{\wp{\program'}{\post} \mid \program' \in \pGCL, \program' \determmixed \program} \eeq \wp{\program}{\post},
    \]
    which together with \Cref{theo:mop-wp-pure} implies the property that is to be shown.
\end{proof}

\section{Omitted Proofs of \titleCref{sec:mop}}
\label{proofs:mop}

This section contains proofs that were omitted in \Cref{sec:mop}, with the exception of the proof of  \Cref{theo:mop-is-pareto}, which is developed in the preceding \Cref{sec:mdp,sec:mop-mdp}.

\subsection{Proof of {\titleCref{theo:sup-is-cl}}}
\label{proof:sup-is-cl}

\supiscl*

\begin{proof}
    Let $A = \conv{\bigcup X}$.  
    Since every $H\in X$ is downward closed and convex, $A$ is convex.  
    Moreover, since each $H\in X$ is downward closed, $A$ is downward closed.  
    We claim that $\cl{A}$ is an element of $\HoarePowerDom$.  
    By definition of Scott closure, $\cl{A}$ is Scott closed (thus also downward closed).  
    Furthermore, the Scott closure of a convex subset of $\PosRealsInfVect$ is convex, see e.g., \cite[Theorem 3.45]{soltan2019lectures}. %
    Hence $\cl{A}$ is non-empty, downward closed, convex, and Scott closed, and therefore $\cl{A}\in\HoarePowerDom$.  
    Since $A$ contains $\bigcup X$, we have $H\subseteq\cl{A}$ for every $H\in X$, so $\cl{A}$ is an upper bound of $X$.  
    Now let $H\in\HoarePowerDom$ be any upper bound of $X$.  
    Then $\bigcup X\subseteq H$.  
    Since $H$ is convex, we obtain
    \[
        A=\conv{\bigcup X}\subseteq H.
    \]
    Since $H$ is Scott closed and contains $A$, this implies $\cl{A}\subseteq H$.

    Thus $\cl{A}$ is the least upper bound of $X$, and consequently
    \[
        \bighoaresup X = \cl{\conv{\bigcup X}}.
    \]
\end{proof}

\subsection{Proof of \titleCref{theo:mop-well-cont-monotone}}
\label{proof:mop-well-cont-monotone}

\mopwellcontmonotone*

\begin{proof}
    We prove all three properties simultaneously by structural induction on $\program$. More precisely, we show that $\mopC{\program}$ maps $\mExp$ to $\mExp$, preserves suprema of $\omega$-chains, and is monotone:
    Let $\meset = \{\me_0 \sqsubseteq \me_1 \sqsubseteq \cdots\} \subseteq \mExp$ be an $\omega$-chain. For $\omega$-continuity, we show
    \[
        \mop{\program}{\bigmecup \meset} \eeq \bigmecup_{\me \in \meset} \mop{\program}{\me}.
    \]
    For monotonicity, let $\me,\meg \in \mExp$ with $\me \sqsubseteq \meg$ and show
    \[
        \mop{\program}{\me} \sqsubseteq \mop{\program}{\meg}.
    \]

    \paragraph{$\program = \SKIP$.}
    Well-definedness follows immediately since $\mop{\SKIP}{\mpost}=\mpost\in\mExp$.
    Moreover,
    \begin{align*}
        \mop{\SKIP}{\bigmecup \meset} 
        \eeq \bigmecup \meset 
        \eeq \bigmecup_{\me\in\meset}\mop{\SKIP}{\me},
    \end{align*}
    which proves $\omega$-continuity.
    Monotonicity follows immediately from
    \[
        \mop{\SKIP}{\me} \eeq \me \sqsubseteq \meg \eeq \mop{\SKIP}{\meg}.
    \]

    \paragraph{$\program = \ASSIGN{x}{\ee}$.}

    For every $\pstate\in\States$,
    \[
        \mop{\ASSIGN{x}{\ee}}{\mpost}(\pstate)
        \eeq \mpost(\pstate\statesubst{x}{\ee(\pstate)}) \in \HoarePowerDom,
    \]
    since $\mpost\in\mExp$. Thus, $\mop{\ASSIGN{x}{\ee}}{\mpost}\in\mExp$.

    Since suprema in $\mExp$ are defined pointwise, we have
    \begin{align*}
        & \mop{\ASSIGN{x}{\ee}}{\bigmecup\meset}(\pstate) \\
        =\ & \left(\bigmecup\meset\right)(\pstate\statesubst{x}{\ee(\pstate)}) \\
        =\ & \bigmecup_{\me\in\meset}\me(\pstate\statesubst{x}{\ee(\pstate)}) \\
        =\ & \bigmecup_{\me\in\meset}\mop{\ASSIGN{x}{\ee}}{\me}(\pstate).
    \end{align*}

    Finally, $\me\sqsubseteq\meg$ implies
    \[
        \me(\pstate\statesubst{x}{\ee(\pstate)})
        \subseteq
        \meg(\pstate\statesubst{x}{\ee(\pstate)})
    \]
    for every $\pstate$, and therefore
    $\mop{\ASSIGN{x}{\ee}}{\me}\sqsubseteq\mop{\ASSIGN{x}{\ee}}{\meg}$.

    For the induction step, we assume that all three properties hold for $\program_1$, $\program_2$, and $\program'$.

    \paragraph{$\program = \COMPOSE{\program_1}{\program_2}$.}

    By the induction hypothesis,
    $\mop{\program_2}{\mpost}\in\mExp$, and hence another application of the induction hypothesis gives
    $
        \mop{\program_1}{\mop{\program_2}{\mpost}}\in\mExp.
    $
    Thus, $\mopC{\program}$ is well-defined.

    For $\omega$-continuity, observe first that
    $\{\mop{\program_2}{\me}\mid\me\in\meset\}$ is an $\omega$-chain by monotonicity of $\mopC{\program_2}$. We can therefore apply the continuity induction hypothesis for $\program_1$ to this chain:
    \begin{align*}
        & \mop{\COMPOSE{\program_1}{\program_2}}{\bigmecup\meset} \\
        =\ & \mop{\program_1}{\mop{\program_2}{\bigmecup\meset}} \\
        =\ & \mop{\program_1}{\bigmecup_{\me\in\meset}\mop{\program_2}{\me}} \tag{IH on $\program_2$} \\
        =\ & \bigmecup_{\me\in\meset}\mop{\program_1}{\mop{\program_2}{\me}} \tag{IH on $\program_1$} \\
        =\ & \bigmecup_{\me\in\meset}\mop{\COMPOSE{\program_1}{\program_2}}{\me}.
    \end{align*}

    Monotonicity follows by two applications of the induction hypothesis:
    \begin{align*}
        \me \sqsubseteq \meg
        &\implies \mop{\program_2}{\me}\sqsubseteq\mop{\program_2}{\meg} \\
        &\implies \mop{\program_1}{\mop{\program_2}{\me}}
        \sqsubseteq \mop{\program_1}{\mop{\program_2}{\meg}}.
    \end{align*}

    \paragraph{$\program = \GC{\guard_1}{\program_1}{\guard_2}{\program_2}$.}

    By the induction hypothesis,
    $\mop{\program_1}{\mpost},\mop{\program_2}{\mpost}\in\mExp$.
    Scalar multiplication by an Iverson bracket preserves elements of $\HoarePowerDom$, and $\mecup$ is the join in $\HoarePowerDom$. Hence,
    \[
        \iverson{\guard_1}\cdot\mop{\program_1}{\mpost}
        \mmecup
        \iverson{\guard_2}\cdot\mop{\program_2}{\mpost}
        \in\mExp.
    \]
    Thus, $\mopC{\program}$ is well-defined.

    Scalar multiplication and $\mecup$ preserve suprema of $\omega$-chains. Using the induction hypothesis,
    \begin{align*}
        & \mop{\program}{\bigmecup\meset} \\
        =\ & \iverson{\guard_1}\cdot\mop{\program_1}{\bigmecup\meset}
        \mmecup \iverson{\guard_2}\cdot\mop{\program_2}{\bigmecup\meset} \\
        =\ & \iverson{\guard_1}\cdot\bigmecup_{\me\in\meset}\mop{\program_1}{\me}
        \mmecup \iverson{\guard_2}\cdot\bigmecup_{\me\in\meset}\mop{\program_2}{\me} \tag{IH} \\
        =\ & \bigmecup_{\me\in\meset}
        \left(\iverson{\guard_1}\cdot\mop{\program_1}{\me}
        \mmecup \iverson{\guard_2}\cdot\mop{\program_2}{\me}\right) \tag{*} \\
        =\ & \bigmecup_{\me\in\meset}\mop{\program}{\me}.
    \end{align*}
    For (*), we use that the two families are indexed by the same chain; monotonicity of the subprogram transformers (IH) makes both image families chains.

    Monotonicity follows directly from the induction hypothesis and monotonicity of scalar multiplication and $\mecup$:
    \begin{align*}
        & \mop{\program}{\me} \\
        =\ & \iverson{\guard_1}\cdot\mop{\program_1}{\me}
        \mmecup \iverson{\guard_2}\cdot\mop{\program_2}{\me} \\
        \sqsubseteq\ & \iverson{\guard_1}\cdot\mop{\program_1}{\meg}
        \mmecup \iverson{\guard_2}\cdot\mop{\program_2}{\meg} \tag{IH} \\
        =\ & \mop{\program}{\meg}.
    \end{align*}

    \paragraph{$\program = \PCHOICE{\program_1}{\ps}{\program_2}$.}

    By the induction hypothesis, $\mop{\program_1}{\mpost},\mop{\program_2}{\mpost}\in\mExp$.
    Since scalar multiplication and Minkowski sum are well-defined on $\HoarePowerDom$, %
    \[
        \ps\cdot\mop{\program_1}{\mpost}
        \mminkow
        (1-\ps)\cdot\mop{\program_2}{\mpost}
        \in\mExp.
    \]
    Thus, $\mopC{\program}$ is well-defined.

    For $\omega$-continuity, the induction hypothesis and preservation of chain suprema by scalar multiplication and Minkowski sum give
    \begin{align*}
        & \mop{\program}{\bigmecup\meset} \\
        =\ & \ps\cdot\mop{\program_1}{\bigmecup\meset}
        \mminkow (1-\ps)\cdot\mop{\program_2}{\bigmecup\meset} \\
        =\ & \ps\cdot\bigmecup_{\me\in\meset}\mop{\program_1}{\me}
        \mminkow (1-\ps)\cdot\bigmecup_{\me\in\meset}\mop{\program_2}{\me} \tag{IH} \\
        =\ & \bigmecup_{\me\in\meset}
        \left(\ps\cdot\mop{\program_1}{\me}
        \mminkow (1-\ps)\cdot\mop{\program_2}{\me}\right) \\
        =\ & \bigmecup_{\me\in\meset}\mop{\program}{\me}.
    \end{align*}
    As in the guarded choice case, the penultimate equality uses that $\meset$ is a chain and that both subprogram transformers are monotone.

    Monotonicity follows from the induction hypothesis and monotonicity of scalar multiplication and Minkowski sum:
    \begin{align*}
        & \mop{\program}{\me} \\
        =\ & \ps\cdot\mop{\program_1}{\me}
        \mminkow(1-\ps)\cdot\mop{\program_2}{\me} \\
        \sqsubseteq\ & \ps\cdot\mop{\program_1}{\meg}
        \mminkow(1-\ps)\cdot\mop{\program_2}{\meg} \tag{IH} \\
        =\ & \mop{\program}{\meg}.
    \end{align*}

    \paragraph{$\program = \WHILEDO{\guard}{\program'}$.}

    For $\mpost\in\mExp$, denote the characteristic function as
    \[
        \phiMop{}{\mpost}(X)
        \eeq \iverson{\neg\guard}\cdot\mpost \mminkow \iverson{\guard}\cdot\mop{\program'}{X}.
    \]
    By the induction hypothesis and well-definedness of scalar multiplication and Minkowski sum, $X\in\mExp$ implies $\phiMop{}{\mpost}(X)\in\mExp$.
    
    Moreover, $\phiMop{}{\mpost}$ is $\omega$-continuous in $X$: for every $\omega$-chain $\{X_i\}_{i\in\Nats}$,
    \begin{align*}
        & \phiMop{}{\mpost}\left({\bigmecup}_i X_i\right) \\
        =\ & \iverson{\neg\guard}\cdot\mpost \mminkow \iverson{\guard}\cdot\mop{\program'}{{\bigmecup}_i X_i} \\
        =\ & \iverson{\neg\guard}\cdot\mpost \mminkow \iverson{\guard}\cdot{\bigmecup}_i\mop{\program'}{X_i} \tag{IH} \\
        =\ & \bigmecup_i {\phiMop{}{\mpost}}(X_i).
    \end{align*}
    Since $\mExp$ is a complete lattice, $\lfp\phiMop{}{\mpost}$ therefore exists and belongs to $\mExp$. Hence, $\mopC{\program}$ is well-defined.

    It remains to show $\omega$-continuity.
    Let $\meset=\{\me_0\sqsubseteq\me_1\sqsubseteq\cdots\}$ be an $\omega$-chain and write $\me_\infty\eeq\bigmecup_{i\in\Nats}\me_i$. Recall that
    \[
        \phiMop{}{\me}(X) \eeq \iverson{\neg\guard}\cdot\me \mminkow \iverson{\guard}\cdot\mop{\program'}{X}.
    \]

    First, $\phiMop{}{\me}$ is $\omega$-continuous in $X$ by the induction hypothesis on $\program'$.
    Moreover, for every fixed $X\in\mExp$, it is $\omega$-continuous in the parameter $\me$:
    \begin{align*}
        & \phiMop{}{\me_\infty}(X) \\
        =\ & \iverson{\neg\guard}\cdot\left(\bigmecup_{i\in\Nats}\me_i\right) \mminkow \iverson{\guard}\cdot\mop{\program'}{X} \\
        =\ & \bigmecup_{i\in\Nats}\left(\iverson{\neg\guard}\cdot\me_i \mminkow \iverson{\guard}\cdot\mop{\program'}{X}\right) \\
        =\ & \bigmecup_{i\in\Nats} {\phiMop{}{\me_i}}(X).
    \end{align*}

    We now show by induction on $k$ that
    \[
        \left(\phiMop{}{\me_\infty}\right)^k(\zeroME)
        \eeq
        \bigmecup_{i\in\Nats}\left(\phiMop{}{\me_i}\right)^k(\zeroME).
        \tag{$*$}
    \]
    For $k=0$, both sides are $\zeroME$.

    For the induction step, note first that
    $
        \left\{\left(\phiMop{}{\me_i}\right)^k(\zeroME)\right\}_{i\in\Nats}
    $
    is an $\omega$-chain.
    This follows by induction on $k$ from $\me_i\sqsubseteq\me_{i+1}$ and monotonicity of $\phiMop{}{\me}$ in both $\me$ and $X$. Hence, using the induction hypothesis and $\omega$-continuity in both arguments,
    \begin{align*}
        & \left({\phiMop{}{\me_\infty}}\right)^{k+1}(\zeroME) \\
        =\ & {\phiMop{}{\me_\infty}} \left(\bigmecup_{i\in\Nats}\left({\phiMop{}{\me_i}}\right)^k(\zeroME)\right) \tag{IH} \\
        =\ & \bigmecup_{i\in\Nats} {\phiMop{}{\me_i}}\left(\left({\phiMop{}{\me_i}}\right)^k(\zeroME)\right) \\
        =\ & \bigmecup_{i\in\Nats} \left({\phiMop{}{\me_i}}\right)^{k+1}(\zeroME).
    \end{align*}
    The second equality holds because both $\{\me_i\}_{i\in\Nats}$ and $\{(\phiMop{}{\me_i})^k(\zeroME)\}_{i\in\Nats}$ are increasing:
    for any two indices $i,j$, choosing $\ell\geq i,j$ bounds $\phiMop{}{\me_j}((\phiMop{}{\me_i})^k(\zeroME))$ by the corresponding term with index $\ell$.

    Using \Cref{theo:Kleene} and $(*)$, we conclude
    \begin{align*}
        \mop{\program}{\me_\infty}
        =\ & \lfp\phiMop{}{\me_\infty} \\
        =\ & \bigmecup_{k\in\Nats}\left(\phiMop{}{\me_\infty}\right)^k(\zeroME) \\
        =\ & \bigmecup_{k\in\Nats}\bigmecup_{i\in\Nats}\left(\phiMop{}{\me_i}\right)^k(\zeroME) \\
        =\ & \bigmecup_{i\in\Nats}\bigmecup_{k\in\Nats}\left(\phiMop{}{\me_i}\right)^k(\zeroME) \\
        =\ & \bigmecup_{i\in\Nats}\lfp\phiMop{}{\me_i} \\
        =\ & \bigmecup_{i\in\Nats} \mop{\program}{\me_i}.
    \end{align*}
    Thus, $\mopC{\program}$ is $\omega$-continuous.

    Finally, suppose $\me\sqsubseteq\meg$. For every $X\in\mExp$, monotonicity of scalar multiplication and Minkowski sum gives
    \[
        \phiMop{}{\me}(X) \sqsubseteq \phiMop{}{\meg}(X).
    \]
    In particular,
    \[
        \phiMop{}{\me}(\lfp\phiMop{}{\meg})
        \sqsubseteq \phiMop{}{\meg}(\lfp\phiMop{}{\meg})
        \eeq \lfp\phiMop{}{\meg}.
    \]
    By Park induction,
    $
        \lfp\phiMop{}{\me} \sqsubseteq \lfp\phiMop{}{\meg}.
    $
    Thus,
    $
        \mop{\program}{\me}\sqsubseteq\mop{\program}{\meg},
    $
    proving monotonicity and completing the induction.
\end{proof}

\subsection{Proof of \titleCref{theo:mop-health}}
\label{proof:mop-health}

\mophealth*

\begin{proof}
First, note that strictness follows from feasibility by setting $b = \zeroME$.
We show the other properties separately by induction over the program structure.

\paragraph{\textbf{Positive homogeneity.}}

We prove that for every $r\in\PosReals$ and $\me\in\mExp$,
\[
    \mop{\program}{r\cdot\me} \eeq r\cdot\mop{\program}{\me}.
\]
\paragraph{$\program=\SKIP$.}

\begin{align*}
    & \mop{\SKIP}{r\cdot\me} \\
    =\ & r\cdot\me  \\
    =\ & r\cdot\mop{\SKIP}{\me}. %
\end{align*}

\paragraph{$\program=\ASSIGN{x}{\ee}$.}

For every $\pstate\in\States$,
\begin{align*}
    & \mop{\ASSIGN{x}{\ee}}{r\cdot\me}(\pstate) \\
    =\ & (r\cdot\me)(\pstate\statesubst{x}{\ee(\pstate)})  \\
    =\ & r\cdot\me(\pstate\statesubst{x}{\ee(\pstate)}) \\
    =\ & r\cdot\mop{\ASSIGN{x}{\ee}}{\me}(\pstate). 
\end{align*}

Assume as induction hypothesis that positive homogeneity holds for $\program_1$, $\program_2$, and $\program'$.

\paragraph{$\program=\COMPOSE{\program_1}{\program_2}$.}

\begin{align*}
    & \mop{\program}{r\cdot\me} \\
    =\ & \mop{\program_1}{\mop{\program_2}{r\cdot\me}}  \\
    =\ & \mop{\program_1}{r\cdot\mop{\program_2}{\me}} \tag{IH on $\program_2$} \\
    =\ & r\cdot\mop{\program_1}{\mop{\program_2}{\me}} \tag{IH on $\program_1$} \\
    =\ & r\cdot\mop{\program}{\me}. 
\end{align*}

\paragraph{$\program=\GC{\guard_1}{\program_1}{\guard_2}{\program_2}$.}

Using that scalar multiplication distributes over $\mecup$,
\begin{align*}
    & \mop{\program}{r\cdot\me} \\
    =\ & \iverson{\guard_1}\cdot\mop{\program_1}{r\cdot\me}
    \mmecup \iverson{\guard_2}\cdot\mop{\program_2}{r\cdot\me} \\
    =\ & \iverson{\guard_1}\cdot r\cdot\mop{\program_1}{\me}
    \mmecup \iverson{\guard_2}\cdot r\cdot\mop{\program_2}{\me} \tag{IH} \\
    =\ & r\cdot\left(
    \iverson{\guard_1}\cdot\mop{\program_1}{\me}
    \mmecup \iverson{\guard_2}\cdot\mop{\program_2}{\me}\right) \\
    =\ & r\cdot\mop{\program}{\me}. 
\end{align*}

\paragraph{$\program=\PCHOICE{\program_1}{\ps}{\program_2}$.}

\begin{align*}
    & \mop{\program}{r\cdot\me} \\
    =\ & \ps\cdot\mop{\program_1}{r\cdot\me}
    \mminkow (1-\ps)\cdot\mop{\program_2}{r\cdot\me}  \\
    =\ & \ps\cdot r\cdot\mop{\program_1}{\me}
    \mminkow (1-\ps)\cdot r\cdot\mop{\program_2}{\me} \tag{IH} \\
    =\ & r\cdot\left(
    \ps\cdot\mop{\program_1}{\me}
    \mminkow (1-\ps)\cdot\mop{\program_2}{\me}\right) \\
    =\ & r\cdot\mop{\program}{\me}. 
\end{align*}

\paragraph{$\program=\WHILEDO{\guard}{\program'}$.}

Recall that
\[
    \mop{\program}{\me} \eeq \lfp\phiMop{}{\me},
    \qquad
    \phiMop{}{\me}(X)
    \eeq \iverson{\neg\guard}\cdot\me
    \mminkow \iverson{\guard}\cdot\mop{\program'}{X}.
\]

First observe that, for every $X\in\mExp$,
\begin{align*}
    & \phiMop{}{r\cdot\me}(r\cdot X) \\
    =\ & \iverson{\neg\guard}\cdot r\cdot\me
    \mminkow \iverson{\guard}\cdot\mop{\program'}{r\cdot X} \\
    =\ & \iverson{\neg\guard}\cdot r\cdot\me
    \mminkow \iverson{\guard}\cdot r\cdot\mop{\program'}{X} \tag{IH} \\
    =\ & r\cdot\left(
    \iverson{\neg\guard}\cdot\me
    \mminkow \iverson{\guard}\cdot\mop{\program'}{X}\right) \\
    =\ & r\cdot\phiMop{}{\me}(X).
\end{align*}

Since $r\cdot \zeroME=\zeroME$, induction on $k$ therefore yields
\[
    \left(\phiMop{}{r\cdot\me}\right)^k(\zeroME)
    \eeq r\cdot\left(\phiMop{}{\me}\right)^k(\zeroME)
\]
for every $k\in\Nats$.
Using Kleene's fixed-point theorem and the fact that scalar multiplication preserves suprema of $\omega$-chains,
\begin{align*}
    & \mop{\program}{r\cdot\me} \\
    =\ & \lfp\phiMop{}{r\cdot\me} \\
    =\ & \bigsqcup_{k\in\Nats}\left(\phiMop{}{r\cdot\me}\right)^k(\zeroME) \tag{\Cref{theo:Kleene}} \\
    =\ & \bigsqcup_{k\in\Nats}r\cdot\left(\phiMop{}{\me}\right)^k(\zeroME) \\
    =\ & r\cdot\bigsqcup_{k\in\Nats}\left(\phiMop{}{\me}\right)^k(\zeroME) \\
    =\ & r\cdot\lfp\phiMop{}{\me} \tag{\Cref{theo:Kleene}} \\
    =\ & r\cdot\mop{\program}{\me}.
\end{align*}

This proves positive homogeneity.

\paragraph{\textbf{Sublinearity.}}

By positive homogeneity, it suffices to prove subadditivity:
\[
    \mop{\program}{\me \mminkow \meg}
    \sqsubseteq
    \mop{\program}{\me} \mminkow \mop{\program}{\meg}.
\]

\paragraph{$\program=\SKIP$.}

\begin{align*}
    & \mop{\SKIP}{\me \mminkow \meg} \\
    =\ & \me \mminkow \meg \\
    =\ & \mop{\SKIP}{\me} \mminkow \mop{\SKIP}{\meg}.
\end{align*}

\paragraph{$\program=\ASSIGN{x}{\ee}$.}

For every $\pstate\in\States$,
\begin{align*}
    & \mop{\ASSIGN{x}{\ee}}{\me \mminkow \meg}(\pstate) \\
    =\ & (\me \mminkow \meg)(\pstate\statesubst{x}{\ee(\pstate)}) \\
    =\ & \me(\pstate\statesubst{x}{\ee(\pstate)}) \mminkow \meg(\pstate\statesubst{x}{\ee(\pstate)}) \\
    =\ & \mop{\ASSIGN{x}{\ee}}{\me}(\pstate) \mminkow \mop{\ASSIGN{x}{\ee}}{\meg}(\pstate).
\end{align*}

For the remaining cases, assume as induction hypothesis that subadditivity holds for $\program_1$, $\program_2$, and $\program'$.

\paragraph{$\program=\COMPOSE{\program_1}{\program_2}$.}

Using the induction hypothesis and monotonicity of $\mopC{\program_1}$,
\begin{align*}
    & \mop{\program}{\me \mminkow \meg} \\
    =\ & \mop{\program_1}{\mop{\program_2}{\me \mminkow \meg}}  \\
    \sqsubseteq\ & \mop{\program_1}{\mop{\program_2}{\me} \mminkow \mop{\program_2}{\meg}} \tag{IH on $\program_2$ and monotonicity} \\
    \sqsubseteq\ & \mop{\program_1}{\mop{\program_2}{\me}} \mminkow \mop{\program_1}{\mop{\program_2}{\meg}} \tag{IH on $\program_1$} \\
    =\ & \mop{\program}{\me} \mminkow \mop{\program}{\meg}. 
\end{align*}

\paragraph{$\program=\GC{\guard_1}{\program_1}{\guard_2}{\program_2}$.}

By the induction hypothesis,
\begin{align*}
    & \mop{\program}{\me \mminkow \meg} \\
    =\ & \iverson{\guard_1}\cdot\mop{\program_1}{\me \mminkow \meg}
    \mmecup \iverson{\guard_2}\cdot\mop{\program_2}{\me \mminkow \meg} \\
    \sqsubseteq\ & \iverson{\guard_1}\cdot\left(\mop{\program_1}{\me}\mminkow\mop{\program_1}{\meg}\right)
    \mmecup \iverson{\guard_2}\cdot\left(\mop{\program_2}{\me}\mminkow\mop{\program_2}{\meg}\right) \tag{IH} \\
    \sqsubseteq\ & \left(\iverson{\guard_1}\cdot\mop{\program_1}{\me}
    \mmecup \iverson{\guard_2}\cdot\mop{\program_2}{\me}\right)
    \mminkow
    \left(\iverson{\guard_1}\cdot\mop{\program_1}{\meg}
    \mmecup \iverson{\guard_2}\cdot\mop{\program_2}{\meg}\right) \\
    =\ & \mop{\program}{\me} \mminkow \mop{\program}{\meg}. 
\end{align*}
The second inclusion holds because each branchwise Minkowski sum is contained in the Minkowski sum of the two joins.

\paragraph{$\program=\PCHOICE{\program_1}{\ps}{\program_2}$.}

Using the induction hypothesis and distributivity of scalar multiplication over Minkowski sum,
\begin{align*}
    & \mop{\program}{\me \mminkow \meg} \\
    =\ & \ps\cdot\mop{\program_1}{\me \mminkow \meg}
    \mminkow(1-\ps)\cdot\mop{\program_2}{\me \mminkow \meg}  \\
    \sqsubseteq\ & \ps\cdot\left(\mop{\program_1}{\me}\mminkow\mop{\program_1}{\meg}\right)
    \mminkow(1-\ps)\cdot\left(\mop{\program_2}{\me}\mminkow\mop{\program_2}{\meg}\right) \tag{IH} \\
    =\ & \left(\ps\cdot\mop{\program_1}{\me}\mminkow(1-\ps)\cdot\mop{\program_2}{\me}\right)
    \mminkow
    \left(\ps\cdot\mop{\program_1}{\meg}\mminkow(1-\ps)\cdot\mop{\program_2}{\meg}\right) \\
    =\ & \mop{\program}{\me} \mminkow \mop{\program}{\meg}.
\end{align*}

\paragraph{$\program=\WHILEDO{\guard}{\program'}$.}

Recall that
\[
    \mop{\program}{\me}\eeq\lfp\phiMop{}{\me},
    \qquad
    \phiMop{}{\me}(X)
    \eeq \iverson{\neg\guard}\cdot\me
    \mminkow\iverson{\guard}\cdot\mop{\program'}{X}.
\]

We first show that, for all $X,Y\in\mExp$,
\[
    \phiMop{}{\me\mminkow\meg}(X\mminkow Y)
    \sqsubseteq
    \phiMop{}{\me}(X)\mminkow\phiMop{}{\meg}(Y).
    \tag{$*$}
\]
Indeed,
\begin{align*}
    & \phiMop{}{\me\mminkow\meg}(X\mminkow Y) \\
    =\ & \iverson{\neg\guard}\cdot(\me\mminkow\meg)
    \mminkow\iverson{\guard}\cdot\mop{\program'}{X\mminkow Y} \\
    \sqsubseteq\ & \iverson{\neg\guard}\cdot(\me\mminkow\meg)
    \mminkow\iverson{\guard}\cdot\left(\mop{\program'}{X}\mminkow\mop{\program'}{Y}\right) \tag{IH} \\
    =\ & \left(\iverson{\neg\guard}\cdot\me
    \mminkow\iverson{\guard}\cdot\mop{\program'}{X}\right)
    \mminkow
    \left(\iverson{\neg\guard}\cdot\meg
    \mminkow\iverson{\guard}\cdot\mop{\program'}{Y}\right) \\
    =\ & \phiMop{}{\me}(X)\mminkow\phiMop{}{\meg}(Y).
\end{align*}

We now show by induction on $k$ that
\[
    \left(\phiMop{}{\me\mminkow\meg}\right)^k(\set{0^n})
    \; \sqsubseteq \;
    \left(\phiMop{}{\me}\right)^k(\set{0^n})
    \mminkow
    \left(\phiMop{}{\meg}\right)^k(\set{0^n}).
    \tag{$**$}
\]
For $k=0$, this holds since
\[
    \set{0^n}=\set{0^n}\mminkow\set{0^n}.
\]
For the induction step, monotonicity of $\phiMop{}{\me\mminkow\meg}$ and $(*)$ give
\begin{align*}
    & \left(\phiMop{}{\me\mminkow\meg}\right)^{k+1}(\set{0^n}) \\
    \sqsubseteq\ & \phiMop{}{\me\mminkow\meg}
    \left(
        \left(\phiMop{}{\me}\right)^k(\set{0^n})
        \mminkow
        \left(\phiMop{}{\meg}\right)^k(\set{0^n})
    \right) \tag{IH} \\
    \sqsubseteq\ & \left(\phiMop{}{\me}\right)^{k+1}(\set{0^n})
    \mminkow
    \left(\phiMop{}{\meg}\right)^{k+1}(\set{0^n}). \tag{$*$}
\end{align*}

Finally, using Kleene's fixed-point theorem, $(**)$, and $\omega$-continuity of the Minkowski sum,
\begin{align*}
    & \mop{\program}{\me\mminkow\meg} \\
    =\ & \bigsqcup_{k\in\Nats}\left(\phiMop{}{\me\mminkow\meg}\right)^k(\set{0^n}) \tag{\Cref{theo:Kleene}} \\
    \sqsubseteq\ & \bigsqcup_{k\in\Nats}
    \left(
        \left(\phiMop{}{\me}\right)^k(\set{0^n})
        \mminkow
        \left(\phiMop{}{\meg}\right)^k(\set{0^n})
    \right) \\
    =\ & \left(\bigsqcup_{k\in\Nats}\left(\phiMop{}{\me}\right)^k(\set{0^n})\right)
    \mminkow
    \left(\bigsqcup_{k\in\Nats}\left(\phiMop{}{\meg}\right)^k(\set{0^n})\right) \\
    =\ & \mop{\program}{\me}\mminkow\mop{\program}{\meg}. \tag{\Cref{theo:Kleene}}
\end{align*}

Thus, $\mopsymbol$ is subadditive. Together with positive homogeneity,
\begin{align*}
    \mop{\program}{r\cdot\me\mminkow\meg}
    \sqsubseteq\ & \mop{\program}{r\cdot\me}\mminkow\mop{\program}{\meg} \\
    =\ & r\cdot\mop{\program}{\me}\mminkow\mop{\program}{\meg},
\end{align*}
which proves sublinearity.

\paragraph{\textbf{Feasibility}}

Let $\mpost \in \mExp$ be bounded by $b \in \PosReals^n$, i.e.,
$
    \mpost \ssqsubseteq b.
$
We prove that
\[
    \mop{\program}{\mpost} \ssqsubseteq b.
\]

\paragraph{$\program = \SKIP$.}
By \Cref{def:mop},
\[
    \mop{\SKIP}{\mpost} \eeq \mpost \ssqsubseteq b.
\]

\paragraph{$\program = \ASSIGN{x}{\ee}$.}

For every $\pstate \in \States$,
\begin{align*}
    & \mop{\ASSIGN{x}{\ee}}{\mpost}(\pstate) \\
    =\ & \mpost(\pstate\statesubst{x}{\ee(\pstate)})  \\
    \ssqsubseteq\ & b.
\end{align*}
The last step follows from the assumption $\mpost \ssqsubseteq b$, which holds at every state and hence in particular at $\pstate\statesubst{x}{\ee(\pstate)}$.

For the remaining cases, assume as induction hypothesis that feasibility holds for $\program_1$, $\program_2$, and $\program'$ and for every postexpectation bounded by $b$.

\paragraph{$\program = \COMPOSE{\program_1}{\program_2}$.}

By the induction hypothesis on $\program_2$,
\[
    \mop{\program_2}{\mpost} \ssqsubseteq b.
\]
We may therefore apply the induction hypothesis on $\program_1$ to the postexpectation $\mop{\program_2}{\mpost}$, which gives
\begin{align*}
    \mop{\program}{\mpost}
    =\ & \mop{\program_1}{\mop{\program_2}{\mpost}} 
    \ssqsubseteq\ b.
\end{align*}

\paragraph{$\program = \GC{\guard_1}{\program_1}{\guard_2}{\program_2}$.}

By the induction hypothesis,
\[
    \mop{\program_1}{\mpost} \ssqsubseteq b
    \qquad\text{and}\qquad
    \mop{\program_2}{\mpost} \ssqsubseteq b.
\]
For every state $\pstate$, each enabled branch therefore contributes only values bounded by $b$. Taking their join cannot introduce a value exceeding $b$. Hence,
\begin{align*}
    \mop{\program}{\mpost}
    =\ & \iverson{\guard_1}\cdot\mop{\program_1}{\mpost}
    \mmecup \iverson{\guard_2}\cdot\mop{\program_2}{\mpost} 
    \ssqsubseteq\ b.
\end{align*}

\paragraph{$\program = \PCHOICE{\program_1}{\ps}{\program_2}$.}

By the induction hypothesis,
\[
    \mop{\program_1}{\mpost} \ssqsubseteq b
    \qquad\text{and}\qquad
    \mop{\program_2}{\mpost} \ssqsubseteq b.
\]
Thus, for every $\pstate \in \States$, all $x$ and $y$ from the respective sets satisfy $x \leq b$ and $y \leq b$. Since $\ps(\pstate)\in[0,1]$,
\[
    \ps(\pstate) \cdot x + (1-\ps(\pstate)) \cdot y \leq \ps(\pstate) \cdot b + (1-\ps(\pstate)) \cdot b = b.
\]
Consequently,
\begin{align*}
    \mop{\program}{\mpost}
    =\ & \ps\cdot\mop{\program_1}{\mpost}
    \mminkow (1-\ps)\cdot\mop{\program_2}{\mpost}
    \ssqsubseteq\ b.
\end{align*}

\paragraph{$\program = \WHILEDO{\guard}{\program'}$.}

Recall that
\[
    \mop{\program}{\mpost} \eeq \lfp\phiMop{}{\mpost},
    \qquad
    \phiMop{}{\mpost}(X)
    \eeq \iverson{\neg\guard}\cdot\mpost
    \mminkow \iverson{\guard}\cdot\mop{\program'}{X}.
\]

We first observe that $\phiMop{}{\mpost}$ preserves the bound $b$. Suppose $X\ssqsubseteq b$. By the induction hypothesis on $\program'$,
$
    \mop{\program'}{X} \ssqsubseteq b.
$
Since also $\mpost\ssqsubseteq b$, for every $\pstate$ the characteristic functional selects either $\mpost(\pstate)$ or $\mop{\program'}{X}(\pstate)$ according to the value of $\guard$. Hence,
\[
    X\ssqsubseteq b
    \quad\Longrightarrow\quad
    \phiMop{}{\mpost}(X)\ssqsubseteq b.
    \tag{$*$}
\]

By \Cref{theo:Kleene},
\[
    \mop{\program}{\mpost}
    \eeq \bigsqcup_{k\in\Nats}\left(\phiMop{}{\mpost}\right)^k(\zeroME).
\]
Since $\zeroME\ssqsubseteq b$, repeated application of $(*)$ yields
\[
    \left(\phiMop{}{\mpost}\right)^k(\zeroME) \ssqsubseteq b
\]
for every $k\in\Nats$. The supremum of an increasing chain whose elements are all bounded by $b$ is again bounded by $b$. Therefore,
$
    \mop{\program}{\mpost} \ssqsubseteq b.
$
This proves feasibility and we are done.

\end{proof}

\section{Additional Material on Invariant-Based Reasoning for Loops}
\label{proofs:loops}

\Cref{sec:loops} presented several rules for reasoning about loops using invariants.
This section provides additional material, starting with relevant theorems and omitted proofs, and ending with some additional examples.

\subsection{Auxiliary Results on Fixed-Points }

We begin with stating the very classic fixpoint theorem by Kleene:
\begin{theorem}[Kleene's Fixpoint Theorem]
\label{theo:Kleene}
    For every $\omega$-continuous function $\Phi$ on a complete lattice with least element $\bot$, we have
    \[
        \lfp \Phi \eeq \bigsqcup \set{\Phi^n(\bot) \mid n \in \Nats}~.
    \]
\end{theorem}

Next, we recap a classic consequence of compactness, which applies in our setting since every element of $\HoarePowerDom$ is a closed subset of the compact space $\PosRealsInfVect$ (see the proof of \Cref{theo:wp-as-weighted-mop}).

\begin{theorem}[Cantor's Intersection Theorem {\cite[cf.\ Theorem~26.9]{munkres2000topology}}]
\label{theo:cantor}
    In a compact space, every decreasing sequence of nonempty closed sets has a nonempty intersection.
\end{theorem}

Finally, we state a theorem about the greatest fixpoint iteration of \molpsymbol.
This is required for the proofs of \Cref{theo:residual,theo:induction_lower}.

\begin{restatable}[Descending Fixpoint Iteration]{theorem}{cokleene}
\label{theo:co-kleene}
    $\gfp \phiMolpNo = \bigsqcap \set{(\phiMolpNo)^n(\topME) \mid n \in \Nats}$.
\end{restatable}

\begin{proof}
    The inclusion $\sqsubseteq$ is pure order theory: we have 
    $\gfp \phiMolpNo \sqsubseteq \topME$, and applying the monotone function $\phiMolpNo$ repeatedly preserves this bound while fixing the left-hand side.
    For $\sqsupseteq$, it suffices that the right-hand side is a fixpoint of $\phiMolpNo$, i.e., that $\phiMolpNo$ commutes with infima of descending sequences.
    In contrast to the $\omega$-continuity of $\phiMopNo$ (\Cref{theo:mop-well-cont-monotone}), this $\omega$-\emph{co}-continuity is not routine:
    infima in $\mExp$ are statewise intersections, and the Minkowski sums and suprema in the rules of \Cref{table:mop} must be commuted with them.
    For that, we leverage \Cref{theo:cantor}.
    As an example, consider the fixed-probability Minkowski combination with $p \in [0,1]$ and let $z \in \bigcap_n \bigl(p \cdot X_n(\pstate) \minkow (1-p) \cdot Y_n(\pstate)\bigr)$ for descending sequences $X_n, Y_n$:
    the witness sets $W_n = \setcomp{(u,v) \in X_n(\pstate) \times Y_n(\pstate)}{z \leq p \cdot u + (1-p) \cdot v}$ are nonempty, closed, and decreasing in the compact space $\PosRealsInfVect \times \PosRealsInfVect$, so \Cref{theo:cantor} yields a single witness pair that works for \emph{all} $n$ simultaneously, proving $z \in p \cdot \bigl(\bigcap_n X_n(\pstate)\bigr) \minkow (1-p) \cdot \bigl(\bigcap_n Y_n(\pstate)\bigr)$, where downward closedness of the two intersections absorbs the remaining slack in $z \leq p \cdot u + (1-p) \cdot v$.
    The suprema of guarded choices require more care:
    due to the convention $0 \cdot \infty = 0$, closure points involving $\infty$-components need not be dominated by exact convex combinations, and one first reduces to the bounded sublattices via finite caps, where the compactness argument applies.
    Lastly, for nested loops, co-continuity of the body transformer propagates through inner greatest fixpoints by a Park-style argument.
    The claim then follows by induction on the program structure. 
\end{proof}

\subsection{Proof of \titleCref{theo:park-upper-mop}}
\label{proof:park-upper-mop}

\parkupper*

\begin{proof}
    Immediate by Knaster-Tarski \cite{tarski1955lattice} since $\lfp \phiMopNo = \bigsqcap \{\meg ~|~ \phiMopNo(\meg) \sqsubseteq \meg\}$. 
\end{proof}

\subsection{Proof of \titleCref{theo:residual}}
\label{proof:residual}

\residual*

\begin{proof}
    By induction on the program structure. The generalization from $\meg = \zeroME$ to arbitrary $\meg$ is what makes the induction go through.
    For $\SKIP$ and assignments, the claim holds with equality.
    For $\COMPOSE{\program_1}{\program_2}$, apply the induction hypothesis for $\program_2$, monotonicity of $\molpC{\program_1}$, and the induction hypothesis for $\program_1$ at the intermediate postexpectations $\mop{\program_2}{\mpost}$ and $\molp{\program_2}{\meg}$.
    For the probabilistic choice, the claim follows from the induction hypotheses by rearranging Minkowski sums; for the guarded choice, we additionally use
    $(X_1 \mminkow Y_1) \mecup (X_2 \mminkow Y_2) \ssqsubseteq (X_1 \mecup X_2) \mminkow (Y_1 \mecup Y_2)$, which holds since the right-hand side is an upper bound on both joinands.
    For a loop, we compare the three fixpoints through their iterates:
    by \Cref{theo:co-kleene}, $\molp{\WHILEDO{\guard}{\program'}}{\mpost \mminkow \meg}$ and $\molp{\WHILEDO{\guard}{\program'}}{\meg}$ are the infima of the descending iterates $B_n = (\phiMolpNo)^n(\topME)$ and $R_n$ of their respective liberal characteristic functions, and by \Cref{theo:Kleene}, $\mop{\WHILEDO{\guard}{\program'}}{\mpost}$ is the supremum of the ascending iterates $A_n = (\phiMopNo)^n(\zeroME)$.
    An inner induction on $n$ establishes the claim stagewise, i.e., $B_n \sqsubseteq A_n \mminkow R_n$:
    the base case is $\topME \sqsubseteq \zeroME \mminkow \topME$, and the step follows from the outer induction hypothesis for the loop body, applied at the stage expectations $A_n$ and $R_n$, together with monotonicity.
    It remains to transfer this to the respective limits.
    Fix $\pstate$ and $z \in \bigcap_n B_n(\pstate)$.
    The stagewise inclusions provide witnesses for every $n$, where the strict witness lies in the closed set $\bigl(\bigsqcup_m A_m\bigr)(\pstate)$ since $A_n \sqsubseteq \bigsqcup_m A_m$; so the sets
    \[
        W_n \eeq \setcomp{(u,v) \in \Bigl(\textstyle\bigsqcup_m A_m\Bigr)(\pstate) \times R_n(\pstate)}{z \leq u + v}
    \]
    are nonempty; they are closed and decreasing in the compact space $\PosRealsInfVect \times \PosRealsInfVect$, so \Cref{theo:cantor} yields a single pair $(u,v)$ that works for all $n$, witnessing
    $z \in \bigl(\bigsqcup_m A_m\bigr)(\pstate) \minkow \bigl(\bigcap_n R_n(\pstate)\bigr)$.
\end{proof}

\subsection{Proof of \titleCref{theo:induction_lower}}
\label{proof:induction_lower}

\inductionlower*

\begin{proof}
    As outlined further above, it suffices to show that the respective assumptions imply $\mop{\program}{\mpost}(\pstate) = \molp{\program}{\mpost}(\pstate)$. For that, it suffices to show that $\molp{\program}{\zeroME}(\pstate) = \{0\}$ because \Cref{theo:residual} then yields (with $\meg = \zeroME$)
    \[
        \molp{\program}{\mpost \mminkow \zeroME}(\pstate) 
        = 
        \molp{\program}{\mpost}(\pstate)
        \subseteq
        \mop{\program}{\mpost}(\pstate) \minkow \{0\} = \mop{\program}{\mpost}(\pstate)
    \]
    and thus $\mop{\program}{\mpost}(\pstate) = \molp{\program}{\mpost}(\pstate)$ by antisymmetry of $\subseteq$. For \Cref{theo:induction_lower}.\ref{theo:induction_lower:dast}, $\molp{\program}{\zeroME}(\pstate) = \{0\}$ holds by assumption. \Cref{theo:induction_lower}.\ref{theo:induction_lower:dct} is more interesting. We first prove that $\molpsymbol$ is $\omega$-co-continuous, which implies that (cf.\ \Cref{theo:co-kleene})
    \[
        \gfp \phiMolpNo = \bigsqcap \set{(\phiMolpNo)^n(\topME) \mid n \in \Nats}
        \quad\text{and thus}\quad
        \forall k\in\Nats\colon \gfp \phiMolpNo \sqsubseteq (\phiMolpNo)^k(\topME)~.
    \]
    By assumption, we have $(\phiMolpNo)^k(\topME)(\pstate) = \{0\}$ for some $k$ and thus $(\gfp \phiMolpNo)(\pstate) = \{0\}$.
\end{proof}

\subsection{Proof of \titleCref{theo:cousot-mop}}
\label{proof:cousot-mop}

\cousotmop*

\begin{proof}
    By induction on $k$, using monotonicity of $\phiMopNo$, we get $\invariant_k \sqsubseteq (\phiMopNo)^{k+1}(\zeroME)$, and each iterate is below $\lfp \phiMopNo$, again by induction on $k$.
    Hence $\lfp \phiMopNo$ is an upper bound on all $\invariant_k$, and thus on their supremum.
\end{proof}

\subsection{Additional Examples on Loop Rules}

The main paper introduced \Cref{ex:split-upper}, for which we now provide the details left open.

\begin{example}[Park Induction]
\label{ex:split-upper-full}
    Recall that we guessed the invariant
    \[
        \invariant \eeq \dwc{\conv{\set{\vvechorizon{x + k}{y},\ \vvechorizon{x}{y + k}}}}~.
    \]
    To verify that $\invariant$ is a superinvariant, first note that for $k = 0$, we have $\phiMopNo(\invariant) = \dwcset{\vvechorizon{x}{y}} = \invariant$.
    For $k > 0$, the rules of \Cref{table:mop} yield
    \[
        \mop{\GC{\true}{\ASSIGN{x}{x+1}}{\true}{\ASSIGN{y}{y+1}}}{\invariant}
        \eeq
        \invariant\subst{x}{x+1} \mmecup \invariant\subst{y}{y+1}~,
    \]
    and substituting $k-1$ for $k$ in this supremum gives, by \Cref{theo:sup-is-cl},
    \[
        \phiMopNo(\invariant)
        \eeq
        \dwc{\conv{\set{\vvechorizon{x + k}{y},\ \vvechorizon{x+1}{y + k - 1},\ \vvechorizon{x + k - 1}{y+1},\ \vvechorizon{x}{y + k}}}}~.
    \]
    The two middle points lie on the segment between the two outer ones, so $\phiMopNo(\invariant) = \invariant$, i.e.,
    the invariant $\invariant$ is a fixpoint, and in particular a superinvariant.
    By \Cref{theo:park-upper-mop} we get
    \[
    \mop{\program_{\text{split}}}{\dwcset{\vvechorizon{x}{y}}} \sqsubseteq \invariant~.
    \]
\end{example}

For reasoning about lower bounds, we presented two loop rules in \Cref{theo:induction_lower}.
The next example illustrates the former of them, applicable to dAST programs.

\begin{example}[Lower Bounds under dAST]
\label{ex:geo-lower}
    Consider the following (non-dCT) loop
    \[
        \program_{\text{geo}} \eeq \WHILEDO{c = 1}{\ \PCHOICE{\ASSIGN{c}{0}}{\tfrac{1}{2}}{\ASSIGN{x}{x+1}}\,}
    \]
    with postexpectation $\dwcset{\vvechorizon{\iverson{x \geq 1}}{\iverson{x = 0}}}$, bounded by $\vvechorizon{1}{1}$:
    the two objectives are the probabilities of terminating with an incremented resp.\ untouched $x$.
    The program is not dCT (flipping tails forever is a diverging execution, albeit of probability $0$), but it is dAST (for every $\pstate$), so \Cref{theo:induction_lower}.\ref{theo:induction_lower:dast} applies.
    We use the invariant
    \[
        \invariant
        \eeq
        [c \neq 1] \cdot \dwcset{\vvechorizon{\iverson{x \geq 1}}{\iverson{x = 0}}}
        \mminkow
        [c = 1] \cdot \dwcset{\vvechorizon{\iverson{x \geq 1} + \tfrac{1}{2} \cdot \iverson{x = 0}}{\tfrac{1}{2} \cdot \iverson{x = 0}}}.
    \]
    For $c \neq 1$, we have $\phiMopNo(\invariant) = \dwcset{\vvechorizon{\iverson{x \geq 1}}{\iverson{x = 0}}} = \invariant$; for $c = 1$, the rules of \Cref{table:mop} give
    \[
        \phiMopNo(\invariant)
        \eeq \tfrac{1}{2} \cdot \invariant\subst{c}{0} \mminkow \tfrac{1}{2} \cdot \invariant\subst{x}{x+1}
        \eeq \dwcset{\vvechorizon{\tfrac{1}{2} \cdot \iverson{x \geq 1} + \tfrac{1}{2}}{\tfrac{1}{2} \cdot \iverson{x = 0}}}
        \eeq \invariant,
    \]
    where the last step uses $\iverson{x \geq 1} + \iverson{x = 0} = 1$ on $\Nats$.
    So $\invariant$ is an exact fixpoint, and \Cref{theo:park-upper-mop} and \Cref{theo:induction_lower}.\ref{theo:induction_lower:dast} together yield $\mop{\program_{\text{geo}}}{\dwcset{\vvechorizon{\iverson{x \geq 1}}{\iverson{x = 0}}}} = \invariant$.
    For the initial state $c = 1$, $x = 0$, we obtain $\dwcset{\vvechorizon{\tfrac{1}{2}}{\tfrac{1}{2}}}$:
    the program cannot do better than a fair split between the two objectives.
\end{example}

Finally, the following example illustrates the application of \Cref{theo:cousot-mop} using lower $\omega$-invariants.

\begin{example}[Lower $\omega$-Invariants]
\label{ex:geo-omega}
    Revisiting \Cref{ex:geo-lower}, we prove the same lower bound with \Cref{theo:cousot-mop}, i.e., without appealing to termination.
    The $k$-th stage must account for what is achievable within $k$ unrollings of the loop. One choice is 
    \[\invariant_0 = [c \neq 1] \cdot \dwcset{\vvechorizon{\iverson{x \geq 1}}{\iverson{x = 0}}} \mminkow [c=1] \cdot \zeroHoare \]
    and, for $k \geq 1$,
    \begin{align*}
        \invariant_k
        \eeq{} & [c \neq 1] \cdot \dwcset{\vvechorizon{\iverson{x \geq 1}}{\iverson{x = 0}}} \\
        & \mminkow~ [c = 1] \cdot \dwcset{\vvechorizon{(1 - 2^{-k}) \cdot \iverson{x \geq 1} + \bigl(\tfrac{1}{2} - 2^{-k}\bigr) \cdot \iverson{x = 0}}{\ \tfrac{1}{2} \cdot \iverson{x = 0}}}.
    \end{align*}
    Computations as in \Cref{ex:geo-lower} verify $\invariant_0 = \phiMopNo(\zeroME)$ and $\invariant_{k+1} = \phiMopNo(\invariant_k)$, so $\invariant$ is a lower $\omega$-invariant.
    The stages form an ascending chain whose coefficients converge to those of the invariant from \Cref{ex:geo-lower}, and since suprema in $\mExp$ are Scott closed (\Cref{theo:sup-is-cl}), $\bigsqcup_k \invariant_k$ is exactly that invariant. We thus obtain the same lower bound as before, at the price of carrying the unrolling index $k$ through every step.
\end{example}

\section{Additional Material on Weighted Objectives and Synthesized Determinizations}
\label{proofs:wsum}

\Cref{sec:wsum} discussed the relation of weighted objectives and how this gives rise to an algorithm for synthesizing determinizations.
Here we provide some additional material on the briefly mentioned scalarization method, the concrete example for a program generating nonexposed extreme points, as well as the omitted proofs.

\begin{figure}[t]%
        \centering
\begin{tikzpicture}[
    >=stealth
]

\coordinate (V1) at (2,3);
\coordinate (V2) at (3,2);

\filldraw[
    fill=mygreen!30,
    draw=mygreen,
    line join=round
]
    (0,0)
    -- (0,3)
    -- (V1)
    -- (V2)
    -- (3,0)
    -- cycle;

\draw[->] (0,0) -- (4,0);
\draw[->] (0,0) -- (0,4.7);

\filldraw[
    fill=myred!40,
    draw=none,
    line join=round
]
(0,{11/3})
-- (4,{7/3})
-- (4,4.5)
-- (0,4.5)
-- cycle;

\draw[draw=myred]
(0,{11/3})
-- (4,{7/3});

\draw[->, thick, myblue]
    (0,0) -- (1.4,4.2)
    node[right] {$\weight$};

\end{tikzpicture}
\caption{Illustration of the points excluded by a single weight vector $\weight$, whose projection onto $\weight$ exceeds $\wp{\program}{\weight \cdot \post}(\pstate)$ (in red).
These points are certainly unachievable.
}
\label{fig:mop-as-weighted-wp}

\Description{
    A two-dimensional geometric illustration shows the points excluded by a single weight vector $w$ from the achievable multiobjective expectation set.
    The horizontal and vertical axes represent two objectives.
    The achievable region is shaded light green and forms a polygon bounded by the coordinate axes, a horizontal upper segment, a descending segment, and a vertical right boundary.
    A blue arrow labeled $w$ extends from the origin upward and to the right, indicating the direction of the weight vector.
    A descending red line passes through the upper-right vertex of the horizontal boundary of the green region and represents the supporting hyperplane associated with $w$.
    The region above this line is shaded light red, indicating points whose projection onto $w$ exceeds the weighted weakest preexpectation $\operatorname{wp}\llbracket C \rrbracket(w \cdot f)(\sigma)$.
    These points are excluded because their weighted objective values exceed the maximum achievable value in the direction of $w$.
    The illustration demonstrates how a single weight vector defines a supporting halfspace that can be used to identify a region of certainly unachievable points.
}

\end{figure}

\subsection{Scalarization of Multiobjective Preexpectations}

We rephrase \Cref{theo:wp-as-weighted-mop} by looking at the supporting halfspace given by each weight vector
\[
    \set{ x \in \PosRealsInfVect \mid \weight \cdot x \leq \wp{\program}{\weight \cdot \post}(\pstate)}.
\]
This set contains all vectors whose weighted sum  $\weight \cdot x$ is no larger than the optimal weighted sum achievable by the program $\program$.
\Cref{theo:wp-as-weighted-mop} implies that for any $y \in \PosRealsInfVect$ which is \emph{not} included in the halfspace, we have
$y \not \in \mop{\program}{\dwcset{\post}}(\pstate)$.
These points $y$ are illustrated in red in \Cref{fig:mop-as-weighted-wp}.

Considering a second weight vector yields a second supporting halfspace and therefore may exclude more points, progressively refining the approximation towards the actually achievable points.
By considering all possible weight vectors, we recover the entire \mopsymbol as the intersection of all resulting halfspaces, as similarly observed by \citet{watanabe2026posterior} (cf.\ Proposition~1) in the special case of indicator expectations.

\begin{restatable}{theorem}{mopasweightedwp}
\label{theo:mop-as-weighted-wp}
    For all $\pstate \in \States$, we have
    \[
        \mop{\program}{\dwcset{\post}}(\pstate) \eeq \bighoareinf_{\weight \in \Weights} \setcomp{ x \in \PosRealsInfVect }{ \weight \cdot x \leq \wp{\program}{\weight \cdot \post}(\pstate)}~.
    \]
\end{restatable}
\begin{proof}
    Fix $\pstate \in \States$ and write
    \[
        \HDel \eeq \mop{\program}{\dwcset{\post}}(\pstate)
        \qquad \text{and} \qquad
        \HDel_\weight \eeq \set{ x \mid \weight \cdot x \leq \wp{\program}{\weight \cdot \post}(\pstate)}~.
    \]
    We use the two observations from the proof of \Cref{theo:wp-as-weighted-mop}:
    for every $\HDel' \in \HoarePowerDom$ and $\weight \in \Weights$, the supremum $\sup\, (\weight \cdot \HDel')$ is attained, so that we may write $\max$ throughout,
    and every supporting halfspace, in particular every $\HDel_\weight$, is itself an element of $\HoarePowerDom$.
    Since infima in $\HoarePowerDom$ are plain intersections (see the discussion after \Cref{def:hoare-power}), the claim amounts to the set equality $\HDel = \bigcap_{\weight \in \Weights} \HDel_\weight$.

    The inclusion $\subseteq$ is immediate:
    for all $x \in \HDel$ and $\weight \in \Weights$, we have
    \[
        \weight \cdot x
        ~\leq~ \max\, (\weight \cdot \HDel)
        \eeq \wp{\program}{\weight \cdot \post}(\pstate)
        \tag{\Cref{theo:wp-as-weighted-mop}}
    \]
    and hence $x \in \HDel_\weight$.

    For $\supseteq$, we show that every point $x \not\in \HDel$ is excluded by at least one halfspace, so that $\bigcap_{\weight} \HDel_\weight$ contains no point outside of $\HDel$.
    Thus, let $x \not\in \HDel$. We construct a weight vector $\weight \in \Weights$ with
    \[
        \wp{\program}{\weight \cdot \post}(\pstate)
        \eeq
        \max\, (\weight \cdot \HDel)
        ~<~ \weight \cdot x~,
        \tag{\Cref{theo:wp-as-weighted-mop}}
    \]
    witnessing $x \not\in \HDel_\weight$.
    Call a component $i$ \emph{unbounded} if $z_i = \infty$ for some $z \in \HDel$, and \emph{bounded} otherwise.
    Unbounded components constrain the choice of $\weight$:
    if $\weight_i > 0$ for an unbounded component $i$ with witness $z \in \HDel$, $z_i = \infty$, then $\max\, (\weight \cdot \HDel) \geq \weight \cdot z \geq \weight_i \cdot \infty = \infty$, so $\HDel_\weight$ is all of $\PosRealsInfVect$ and excludes nothing.
    Any separating weight vector must therefore vanish on the unbounded components.
    We proceed in four steps: (i) we show that $\HDel$ is a full \emph{cylinder} along its unbounded components, so that discarding these components loses no information; (ii) we pass to the \emph{bounded part} $\HDel^0$ of $\HDel$, a finitely valued element of $\HoarePowerDom$, together with the point $x^0 \not\in \HDel^0$ obtained from $x$ by setting all unbounded components to $0$; (iii) we \emph{separate} $x^0$ from $\HDel^0$ by a weight vector; (iv) and we \emph{transfer} the separation back to $\HDel$ by zeroing that weight vector on the unbounded components and renormalizing.

    For step (i), let $i$ be an unbounded component and $y \in \HDel$. We claim that every point $y'$ with $y'_j \leq y_j$ for all $j \neq i$ and an \emph{arbitrary} value $y'_i$ lies in $\HDel$ as well, i.e., membership in $\HDel$ places no constraint whatsoever on any unbounded component.
    This is forced by the three closure properties of $\HoarePowerDom$:
    pick $z \in \HDel$ with $z_i = \infty$ and define, for every $\weightc \in (0,1]$, the point $p_\weightc \in \PosRealsInfVect$ with components
    \[
        (p_\weightc)_i \eeq y'_i
        \qquad \text{and} \qquad
        (p_\weightc)_j \eeq (1-\weightc) \cdot y_j \quad \text{for } j \neq i~.
    \]
    Each $p_\weightc$ lies in $\HDel$:
    the convex combination $\weightc \cdot z + (1-\weightc) \cdot y$ is in $\HDel$ by convexity, its $i$-th component is $\weightc \cdot \infty + (1-\weightc) \cdot y_i = \infty \geq y'_i$ since $\weightc > 0$, its $j$-th component is at least $(1-\weightc) \cdot y_j$, and $\HDel$ is downward closed.
    The set $D = \set{p_\weightc \mid \weightc \in (0,1]}$ is a chain, and hence directed:
    for $\weightc' \leq \weightc$, we have $p_\weightc \leq p_{\weightc'}$ componentwise, i.e., the points increase as $\weightc$ decreases.
    The supremum of $D$ is computed componentwise:
    at $i$, all points share the constant value $y'_i$,
    and at $j \neq i$, we have $\sup_{\weightc \in (0,1]}\, (1-\weightc) \cdot y_j = y_j$, since $1-\weightc$ ranges over $[0,1)$:
    for finite $y_j$, the values $(1-\weightc) \cdot y_j$ approach $y_j$, and for $y_j = \infty$, every $\weightc < 1$ already yields $(1-\weightc) \cdot \infty = \infty$ (for $\weightc = 1$, the convention $0 \cdot \infty = 0$ applies, affecting neither case).
    By Scott closedness, $\sup D \in \HDel$, and since $y' \leq \sup D$, with equality at component $i$, downward closedness yields $y' \in \HDel$.

    For step (ii), let $\HDel^0$ consist of all points of $\HDel$ whose unbounded components are all $0$, and let $x^0$ be obtained from $x$ by setting all unbounded components to $0$.
    Then $\HDel^0$ is again an element of $\HoarePowerDom$:
    we have $\HDel^0 = \HDel \cap \dwcset{u}$ for the single point $u$ with value $0$ at all unbounded and $\infty$ at all bounded components;
    the downward closure $\dwcset{u} = \set{y \mid y \leq u}$ contains $0^n$, is downward closed by construction, and, being defined by componentwise upper bounds, is closed under convex combinations and under suprema of directed sets,
    and intersections of elements of $\HoarePowerDom$ are again elements, being their infima, where nonemptiness holds since $0^n$ lies in every nonempty downward closed set.
    Moreover, $\HDel^0$ is \emph{finitely valued}, i.e.\ $\HDel^0 \subseteq \PosRealsVect$ because 
    its points are $0$ at the unbounded components and finite at the bounded ones, by the definition of boundedness.
    Finally, we have $x^0 \not\in \HDel^0$:
    suppose to the contrary that $x^0 \in \HDel^0$, enumerate the unbounded components as $i_1, \dots, i_m$ (with $m = 0$ if there are none), and let $x^{(k)}$ denote the point that agrees with $x$ on $i_1, \dots, i_k$ and with $x^0$ on all other components.
    In particular, $x^{(0)} = x^0$ and $x^{(m)} = x$, as $x$ and $x^0$ differ at most on unbounded components.
    By induction on $k$, every $x^{(k)}$ lies in $\HDel$:
    we have $x^{(0)} = x^0 \in \HDel^0 \subseteq \HDel$ by assumption, and $x^{(k+1)}$ arises from $x^{(k)} \in \HDel$ by changing only the value of the \emph{unbounded} component $i_{k+1}$, which the cylinder property, applied with $y = x^{(k)}$ and $y' = x^{(k+1)}$, leaves unconstrained.
    Hence $x = x^{(m)} \in \HDel$, contradicting $x \not\in \HDel$.

    For step (iii), we produce a weight vector $\weight' \in \Weights$ separating $x^0$ from $\HDel^0$, i.e., with $\max\, (\weight' \cdot \HDel^0) < \weight' \cdot x^0$, by distinguishing whether $x^0$ is finitely valued.
    If $x^0$ has an infinite component $x^0_j = \infty$, then $j$ is necessarily bounded, as the unbounded components of $x^0$ are $0$, and the $j$-th unit vector $\weight' = e_j \in \Weights$ suffices:
    the convention $0 \cdot \infty = 0$ discards all other components of $x^0$, so $\weight' \cdot x^0 = x^0_j = \infty$, whereas $\max\, (\weight' \cdot \HDel^0) < \infty$, since the maximum is attained at a point of the finitely valued set $\HDel^0$.
    If instead $x^0$ is finitely valued i.e., $x^0 \in \PosRealsVect$, then both $x^0$ and $\HDel^0$ are contained in $\PosRealsVect \subseteq \Reals^n$, where weighted sums coincide with the usual inner products, and we invoke convex geometry.
    The set $\HDel^0$ is a \emph{compact convex} subset of $\Reals^n$:
    convexity holds by \Cref{def:hoare-power} and, on finitely valued sets, coincides with convexity in $\Reals^n$;
    as established in the proof of \Cref{theo:wp-as-weighted-mop}, $\HDel^0$ is a closed subset of the compact space $\PosRealsInfVect$ and hence compact there,
    and since compactness is intrinsic to the subspace \cite[Section~26]{munkres2000topology} and the topology of $\PosRealsInfVect$ restricted to $\PosRealsVect$ is the Euclidean one (established, e.g., by the homeomorphism $x \mapsto \frac{x}{x+1}$, $\infty \mapsto 1$), $\HDel^0$ is compact as a subset of $\Reals^n$, too.
    Moreover, the upward closed set $U = \set{y \in \Reals^n \mid x^0 \leq y}$ is closed, convex, and disjoint from $\HDel^0$, since any point of $\HDel^0$ above $x^0$ would, by downward closedness, imply $x^0 \in \HDel^0$.
    The separating hyperplane theorem for disjoint closed convex sets, one of which is compact \cite[Part III]{rockafellar1970convex}, thus yields a normal vector $c \in \Reals^n$ with
    \[
        \sup\, \set{c \cdot y \mid y \in \HDel^0} ~<~ \inf\, \set{c \cdot z \mid z \in U}~,
    \]
    where the left-hand side is a real number, since $\HDel^0$ is nonempty and compact.
    The vector $c$ is nonnegative:
    if $c_i < 0$, then moving from $x^0$ arbitrarily far along the $i$-th coordinate, which stays inside $U$, would drive the right-hand side to $-\infty$, contradicting that it strictly exceeds the real left-hand side.
    Further, $c \neq 0$, as otherwise both sides would equal $0$, since $\HDel^0$ and $U$ are nonempty, contradicting strictness of the separation.
    Consequently, $\weight' = c / \sum_i c_i \in \Weights$ is a weight vector, and dividing the separation inequality by $\sum_i c_i > 0$ yields $\max\, (\weight' \cdot \HDel^0) < \weight' \cdot x^0$, where we use that $x^0 \in U$.

    For step (iv), we transfer the separation from $\HDel^0$ back to $\HDel$.
    Since the points of $\HDel^0$ as well as $x^0$ are $0$ on all unbounded components, both sides of the inequality $\max\, (\weight' \cdot \HDel^0) < \weight' \cdot x^0$ depend only on the bounded components of $\weight'$.
    In particular, $s = \sum_{j \text{ bounded}} \weight'_j$ is strictly positive, as $s = 0$ would force both sides to be $0$. That $\weight' \cdot x^0$ vanishes in this case again uses $0 \cdot \infty = 0$ if $x^0$ has an infinite component.
    We may thus zero and renormalize: define $\weight \in \Weights$ by $\weight_i = 0$ for unbounded $i$ and $\weight_j = \weight'_j / s$ for bounded $j$, whose components indeed sum to $s / s = 1$.
    For $y \in \HDel$, let $y^0$ denote the point obtained by setting all unbounded components of $y$ to $0$. Then $y^0 \in \HDel$ by downward closedness, and hence $y^0 \in \HDel^0$.
    Using the convention $0 \cdot \infty = 0$ at the unbounded components, we have $\weight \cdot y = \frac{1}{s} \cdot (\weight' \cdot y^0)$, and likewise $\weight \cdot x = \frac{1}{s} \cdot (\weight' \cdot x^0)$.
    Since moreover every point of $\HDel^0$ arises as $y^0$ from some $y \in \HDel$, namely from itself, we conclude
    \[
        \max\, (\weight \cdot \HDel)
        \eeq \frac{1}{s} \cdot \max\, (\weight' \cdot \HDel^0)
        ~<~ \frac{1}{s} \cdot (\weight' \cdot x^0)
        \eeq \weight \cdot x~,
    \]
    where strictness is preserved as $0 < \frac{1}{s} < \infty$.
    This exhibits the required excluding halfspace $\HDel_\weight$ and completes the proof of $\supseteq$.
\end{proof}

As discussed before, a point is excluded from a halfspace whenever it is too optimistic in the direction of some $\weight$, i.e., whenever its weighted sum exceeds the value achievable by the program.
Conversely, a point belongs to the intersection of all these halfspaces exactly when no weighted sum of its coordinates exceeds the value that the program can achieve under the corresponding weighted objective.

\begin{figure}[t]
    \begin{subfigure}[c]{0.5\textwidth}
\centering
\begin{lstlisting}[mathescape]
$\ASSIGN{q}{1} \fatsemi \ASSIGN{p}{0} \fatsemi \ASSIGN{c}{1} \fatsemi$
$\WHILE{c=1} \{$
  $\COMPOSE{\ASSIGN{x}{\frac{q^2-p^2}{q^2+p^2}}}{\ASSIGN{y}{\sqrt{1-x^2} + 1}} \fatsemi$
  $\GCFSTOPEN{~\blue{\trueBold}}$
    $\gray{\text{\% stop}}$
    $\ASSIGN{c}{0}$ 
  $\GCSECCLOSE\GCSECOPEN{\blue{\trueBold}}{}$ 
    $\gray{\text{\% choose next rational}}$
    $\GCFSTOPEN{~q \leq p}{\ASSIGN{q}{q+1} \fatsemi \ASSIGN{p}{0}} \GCSECCLOSE$
    $\GCSECOPEN{q > p}{\ASSIGN{p}{p+1}} \GCSECCLOSE$
  $\GCSECCLOSE\GCSECOPEN{\boldsymbol{\blue{x=1}}}{}$
    $\gray{\text{\% stop and drop y}}$
    $\ASSIGN{y}{0} \fatsemi \ASSIGN{c}{0}$
  $\GCSECCLOSE$
$\}$
\end{lstlisting}
\caption{$\program_{\text{circ}}$. 
    The ternary guarded choice is syntactic sugar for a nested binary guarded choice. %
}
\label{fig:nonexposed-face}
    \end{subfigure}
    \hfill
        \begin{subfigure}[c]{0.45\textwidth}
\centering
\begin{tikzpicture}[
    scale=2,
    >=stealth
]

\filldraw[fill=mygreen!30,draw=mygreen]
    (0,2)
    arc[start angle=90,end angle=0,radius=1]
    -- (1,0)
    -- (0,0)
    -- cycle;

\draw[->] (0,0) -- (1.6,0) node[below] {$x$};
\draw[->] (0,0) -- (0,2.3) node[left] {$y$};

\foreach \x in {1}
    \draw (\x,0.01) -- (\x,-0.01) node[below] {\x};

\foreach \y in {1,2}
    \draw (0.01,\y) -- (-0.01,\y) node[left] {\y};

\coordinate (t2) at (1,1);
\fill[myred] (t2) circle (1pt);
\node[myred,right] at ($(t2)$) {$\target$};

\end{tikzpicture}
\caption{Illustration of $\mop{\program_{\text{circ}}}{\dwcset{\vvechorizon{x}{y}}}$, including a target $\target$ for synthesis.}
\label{fig:circ-pareto-mop}
\Description{
    Two panels are shown side by side, illustrating the nondeterministic program $C_{\mathrm{circ}}$ and its corresponding multiobjective preexpectation set.

    The left panel presents the program $C_{\mathrm{circ}}$, which initializes the variables $q$ and $c$ to one and $p$ to zero.
    While $c=1$, the program computes $x$ as the ratio $(q^2-p^2)/(q^2+p^2)$ and assigns $y$ the value $\sqrt{1-x^2}+1$.
    The program then makes a nondeterministic choice between three branches.
    The first branch stops execution by setting $c$ to zero.
    The second branch advances to the next rational number by incrementing $q$ and resetting $p$ to zero when $q \leq p$, or incrementing $p$ when $q>p$.
    The third branch, enabled when $x=1$, stops execution by setting both $y$ and $c$ to zero.
    The three-way guarded choice is syntactic sugar for a nested binary guarded choice.

    The right panel shows the multiobjective preexpectation set of $C_{\mathrm{circ}}$ with respect to the objectives $x$ and $y$.
    The horizontal axis represents $x$, while the vertical axis represents $y$.
    The achievable region is shaded light green and is bounded by the coordinate axes, a curved upper boundary, and a vertical segment at $x=1$.
    The curved boundary follows a quarter-circle arc from $(0,2)$ to $(1,1)$.
    The vertical boundary extends from $(1,1)$ to $(1,0)$.
    A red point labeled $\tau$ is located at $(1,1)$, where the curved boundary meets the vertical segment.
    This point represents the target for strategy synthesis.
    Together, the panels illustrate how the nondeterministic program generates a convex multiobjective expectation set with a curved boundary and a non-exposed extreme point $\tau$.
}

    \end{subfigure}
    \caption{Nondeterministic program $\program_{\text{circ}}$ and the corresponding $\mop{\program_{\text{circ}}}{\dwcset{\vvechorizon{x}{y}}}$.
    }
\label{fig:nonexposed-face-full-1}
\Description{The circle program, producing the Pareto front with non-exposed extreme points.}
\end{figure}

\subsection{Non-exposed Extreme Points}

The possible existence of non-exposed extreme points was mentioned in \Cref{sec:synthesis}.
We now give a concrete program for which the almost achievable points contain such a point.

\begin{example}
\label{ex:circle}
    Consider the program $\program_{\text{circ}}$ in \Cref{fig:nonexposed-face}.
    The program enumerates all rational points of the upper right quarter of the unit circle centered at $(0,1)$ by looping through values $q=1,2,\dots$ and $p=0,1,\dots,q$ and assigning to $x$ the rational $\frac{q^2-p^2}{q^2+p^2}$ whose corresponding $y$ coordinate on the unit circle is rational as well.
    In each iteration, the program can either stop and return the current point, continue to the next rational number, or, if $x=1$, stop and drop the $y$ coordinate to $0$.
    For post $\vvechorizon{x}{y}$, we get
    \begin{align*}
    & \mop{\program_{\text{circ}}}{\dwcset{\vvechorizon{x}{y}}} 
    =\ \dwc{\conv{\set{\vvechorizon{1}{0}} \cup \set{\vvechorizon{a}{b+1} \mid a^2+b^2 = 1, a,b \geq 0} }}, %
    \end{align*}
    which is visualized in \Cref{fig:circ-pareto-mop}, having the same shape as \Cref{fig:non-exposed}.

\begin{figure}[t]
    \begin{subfigure}[l]{0.45\textwidth}
\centering
\begin{lstlisting}[mathescape]
$\ASSIGN{q}{1} \fatsemi \ASSIGN{p}{0} \fatsemi \ASSIGN{c}{1} \fatsemi$
$\WHILE{c=1} \{$
  $\COMPOSE{\ASSIGN{x}{\frac{q^2-p^2}{q^2+p^2}}}{\ASSIGN{y}{\sqrt{1-x^2} + 1}} \fatsemi$
  $\GCFSTOPEN{~\boldsymbol{\red{x=1}}}$
    $\gray{\text{\% stop}}$
    $\ASSIGN{c}{0}$ 
  $\GCSECCLOSE\GCSECOPEN{~\boldsymbol{\red{x \neq 1}}}{}$ 
    $\gray{\text{\% choose next rational}}$
    $\GCFSTOPEN{~q \leq p}{\ASSIGN{q}{q+1} \fatsemi \ASSIGN{p}{0}} \GCSECCLOSE$
    $\GCSECOPEN{q > p}{\ASSIGN{p}{p+1}} \GCSECCLOSE$
  $\GCSECCLOSE\GCSECOPEN{~\boldsymbol{\red{\falseBold}}}{}$
    $\gray{\text{\% stop and drop y}}$
    $\ASSIGN{y}{0} \fatsemi \ASSIGN{c}{0}$
  $\GCSECCLOSE$
$\}$
\end{lstlisting}
\caption{Determinization $\program_{\text{circ}}^{\target}$ achieving $\target=\vvechorizon{1}{1}$.}
\label{fig:nonexposed-face-ach-t}
    \end{subfigure}
    \hfill
    \begin{subfigure}[l]{0.45\textwidth}
\centering
\begin{lstlisting}[mathescape]
$\ASSIGN{q}{1} \fatsemi \ASSIGN{p}{0} \fatsemi \ASSIGN{c}{1} \fatsemi$
$\WHILE{c=1} \{$
  $\COMPOSE{\ASSIGN{x}{\frac{q^2-p^2}{q^2+p^2}}}{\ASSIGN{y}{\sqrt{1-x^2} + 1}} \fatsemi$
  $\GCFSTOPEN{~\boldsymbol{\red{\falseBold}}}$
    $\gray{\text{\% stop}}$
    $\ASSIGN{c}{0}$ 
  $\GCSECCLOSE\GCSECOPEN{~\boldsymbol{\red{x \neq 1}}}{}$ 
    $\gray{\text{\% choose next rational}}$
    $\GCFSTOPEN{~q \leq p}{\ASSIGN{q}{q+1} \fatsemi \ASSIGN{p}{0}} \GCSECCLOSE$
    $\GCSECOPEN{q > p}{\ASSIGN{p}{p+1}} \GCSECCLOSE$
  $\GCSECCLOSE\GCSECOPEN{~\boldsymbol{\red{x=1}}}{}$
    $\gray{\text{\% stop and drop y}}$
    $\ASSIGN{y}{0} \fatsemi \ASSIGN{c}{0}$
  $\GCSECCLOSE$
$\}$
\end{lstlisting}
\caption{Determinization $\program_{\text{circ}}^{x}$ optimizing for $x$.}
\label{fig:nonexposed-face-opt-x}
    \end{subfigure}
        
    \caption{Determinizations of $\program_{\text{circ}}$ from \Cref{fig:nonexposed-face}, both optimal for $\weight \cdot \vvechorizon{x}{y}$.}
\label{fig:nonexposed-face-full-2}

\Description{
    Two program listings are shown side by side, illustrating two different determinizations of the nondeterministic program $C_{\mathrm{circ}}$.
    Both programs initialize $q$ and $c$ to one and $p$ to zero.
    Inside a loop that continues while $c=1$, both programs compute $x$ as $(q^2-p^2)/(q^2+p^2)$ and $y$ as $\sqrt{1-x^2}+1$.
    The nondeterministic choices of the original program are replaced by deterministic guards, highlighted in red.

    The left listing shows the determinization $C_{\mathrm{circ}}^{\tau}$, which achieves the target $\tau=(1,1)$.
    The first branch is selected when $x=1$ and stops execution by setting $c$ to zero.
    The second branch is selected when $x\neq 1$ and advances to the next rational number by updating $q$ and $p$.
    The third branch, which would stop execution and reset $y$ to zero, is disabled by the guard $\texttt{false}$.
    This determinization therefore stops at the target point without discarding the value of $y$.

    The right listing shows the determinization $C_{\mathrm{circ}}^{x}$, which optimizes the objective $x$.
    The first branch is disabled by the guard $\texttt{false}$.
    The second branch advances to the next rational number whenever $x\neq 1$.
    The third branch is selected when $x=1$ and stops execution by setting both $y$ and $c$ to zero.
    This determinization therefore reaches the point $(1,0)$ instead of the target $\tau=(1,1)$.

    Both determinizations are optimal for the weighted objective $w\cdot(x,y)$ with a weight vector directed along the positive $x$-axis.
    However, they produce different outcomes for the second objective, illustrating that optimizing a weighted objective does not necessarily identify a unique target point.
}

\end{figure}

    Assume we want to synthesize a determinization achieving the point
    $
        \target=\vvechorizon{1}{1}.
    $
    This point is achievable, as witnessed by the determinization $\program_{\text{circ}}^{\target}$ in \Cref{fig:nonexposed-face-ach-t}, achieving $\target$ by stopping as soon as $x=1$, so the previous assignment must have set $y$ to $1$. %
    The face containing $\target$ is exposed by the weight vector $\weight=\vvechorizon{1}{0}$, suggesting that we maximize the objective $1\cdot x + 0\cdot y = x$.
    And in fact, the determinization $\program_{\text{circ}}^{\target}$ is optimal for $x$.
    However, it is not the only one: the determinization shown in \Cref{fig:nonexposed-face-opt-x} is optimal for $x$ as well.
    This program stops as soon as $x=1$, sacrificing $y$ by setting it to $0$, thus yielding the value
    $
        \wpVec{\program_{\mathrm{circ}}^{\weight}}{\vvechorizon{x}{y}} = \vvechorizon{1}{0},
    $
    missing the target $\target=\vvechorizon{1}{1}$.

    The reason is that $\target$ is not exposed:
    the face containing $\target$ also contains all points on the line from $\vvechorizon{1}{0}$ to $\vvechorizon{1}{1}$.
    Consequently, the determinization optimal for the respective weight vector is only guaranteed to achieve a point on the same face as $\target$, but not necessarily $\target$ itself.
\end{example}

\subsection{Proof of \titleCref{theo:wp-as-weighted-mop}}
\label{proof:wp-as-weighted-mop}

\wpasweightedmop*

\begin{proof}
    Throughout the proof, fix a weight vector $\weight \in \Weights$.
    We first collect two observations on how weighted sums interact with the Hoare powerdomain.

    First, we show that for every $\HDel \in \HoarePowerDom$, the supremum $\sup\, (\weight \cdot \HDel)$ is \emph{attained} by some point of $\HDel$; in particular, the maximum in the statement of the theorem is well-defined.
    We equip $\PosRealsInf$ with its order topology \cite[Section~14]{munkres2000topology}, under which it is compact \cite[Theorem~27.1]{munkres2000topology}.
    Accordingly, we equip $\PosRealsInfVect$ with the product topology, which is again compact \cite[Theorem~26.7]{munkres2000topology}.
    The map $x \mapsto \weight \cdot x = \sum_{i} \weight_i \cdot x_i$ from $\PosRealsInfVect$ to $\PosRealsInf$ is continuous with respect to these topologies:
    addition is continuous on $\PosRealsInf$, and so is multiplication by the constant $\weight_i \in [0,1]$.
    Moreover, every $\HDel \in \HoarePowerDom$ is a closed subset of $\PosRealsInfVect$, as downward closedness and Scott closedness together imply topological closedness.\footnote{%
        Let $a \not\in \HDel$ and let $D$ be the set of all $c \leq a$ with $c_i < a_i$ for every $i$ with $a_i \neq 0$.
        Then $D$ is directed with $\sup D = a$, so Scott closedness of $\HDel$ yields some $c \in D$ with $c \not\in \HDel$.
        The open set $\set{z \mid c_i < z_i \text{ for all } i \text{ with } a_i \neq 0}$ is a neighborhood of $a$ disjoint from $\HDel$:
        every $z$ in this set satisfies $c \leq z$, so $z \in \HDel$ would imply $c \in \HDel$ by downward closedness.
        More abstractly, Scott closed subsets of a continuous lattice are closed in its Lawson topology, which on $\PosRealsInfVect$ coincides with the product topology considered here; see \cite[Chapter~III]{gierz2003continuous}.}
    Consequently, $\HDel$ is a nonempty (\Cref{def:hoare-power}) closed subset of a compact space and hence itself compact \cite[Theorem~26.2]{munkres2000topology}.
    The extreme value theorem \cite[Theorem~27.4]{munkres2000topology}, which applies since the codomain $\PosRealsInf$ carries the order topology, thus yields that the continuous map $\weight \cdot (\cdot)$ attains a maximum on the nonempty compact set $\HDel$, i.e., $\sup\, (\weight \cdot \HDel) = \max\, (\weight \cdot \HDel)$ is attained by some point of $\HDel$.

    Second, for every $c \in \PosRealsInf$, the supporting halfspace $\set{x \mid \weight \cdot x \leq c}$ is itself an element of $\HoarePowerDom$:
    it contains $0^n$, it is downward closed since $\weight \geq 0$, it is convex since $\weight \cdot (\cdot)$ is linear, and it is Scott closed since $\weight \cdot (\cdot)$ preserves suprema of directed sets.
    Now observe that
    \[
        \max\, (\weight \cdot \HDel) \leq c
        \quad \text{iff} \quad
        \HDel \subseteq \set{x \mid \weight \cdot x \leq c}~.
    \]
    Since the supporting halfspace lies in $\HoarePowerDom$ and $\bighoaresup X$ is the \emph{least} element of $\HoarePowerDom$ containing every $\HDel \in X$, this equivalence shows that, for every $X \subseteq \HoarePowerDom$, the quantities $\max\, (\weight \cdot \bighoaresup X)$ and $\sup_{\HDel \in X}\, \max\, (\weight \cdot \HDel)$ have the same upper bounds $c$, and are therefore equal, i.e., we have
    \[
        \max\, \bigl(\weight \cdot \bighoaresup X\bigr) \eeq \sup_{\HDel \in X}~ \max\, (\weight \cdot \HDel)~.
        \tag{$\ast$}
    \]
    Notice that the outer supremum on the right-hand side ranges over the possibly infinite set $X$ and need not be attained.
    Intuitively, ($\ast$) states that the convex and Scott closures taken by $\bighoaresup$ (\Cref{theo:sup-is-cl}) do not add points with larger weighted sums.

    Next, in order to enable a proof by induction on the program structure, we generalize the claim from the principal postexpectation $\dwcset{\post}$ to \emph{arbitrary} $\mpost \in \mExp$. For that,
    define the \emph{scalarization of $\mpost \in \mExp$ along $\weight$} as the classic expectation
    \[
        \scal{\weight}{\mpost} \eeq \mylambda \pstate \max\, \bigl(\weight \cdot \mpost(\pstate)\bigr) \morespace{\in} \Exp~.
    \]
    We claim that scalarization commutes with the two transformers, i.e.,
    \[
        \forall \program \in \pGCL \colon \forall \mpost \in \mExp\colon \quad
        \scal{\weight}{\mop{\program}{\mpost}} \eeq \wp{\program}{\scal{\weight}{\mpost}}.
        \tag{$\dagger$}
    \]
    The theorem is the instance $\mpost = \dwcset{\post}$ of ($\dagger$) evaluated at $\pstate$:
    by monotonicity of $\weight \cdot (\cdot)$, the maximum of $\weight \cdot x$ over $\dwcset{\post(\pstate)} = \set{x \mid x \leq \post(\pstate)}$ is attained at $x = \post(\pstate)$ itself, so $\scal{\weight}{\dwcset{\post}} = \weight \cdot \post$.
    The generalization to arbitrary $\mpost$ is needed because
    sequential composition and loops apply \mopsymbol\ to intermediate postexpectations which are in general not of the form $\dwcset{\postg}$ for any $\postg \in \Exp^n$, and the induction hypothesis must be available for these cases as well.

    We prove ($\dagger$) by induction on the program structure, going over the rules in \Cref{table:mop,tab:wp}.
    For $\SKIP$ and $\ASSIGN{x}{\ee}$, and sequential composition, this is straightforward.
    For $\PCHOICE{\program_1}{\ps}{\program_2}$, it suffices that weighted maxima are linear in convex Minkowski combinations, i.e., for all $\HDel, \HDel' \in \HoarePowerDom$ and $p \in [0,1]$,
    \[
        \max\, \Bigl(\weight \cdot \bigl(p \cdot \HDel \minkow (1-p) \cdot \HDel'\bigr)\Bigr)
        \eeq
        p \cdot \max\, (\weight \cdot \HDel) + (1-p) \cdot \max\, (\weight \cdot \HDel')~.
    \]
    Here, $\leq$ holds because $\weight \cdot (p \cdot x + (1-p) \cdot y) = p \cdot (\weight \cdot x) + (1-p) \cdot (\weight \cdot y)$, and $\geq$ holds by combining the two maximizers, which exist by the first observation above.
    For the guarded choice, multiplication by the guards commutes with scalarization statewise, and the binary suprema taken in the respective domains match by ($\ast$), i.e.,
    \[
        \max\, \bigl(\weight \cdot (\HDel \hoaresup \HDel')\bigr) \eeq \max \set{\max\, (\weight \cdot \HDel),\, \max\, (\weight \cdot \HDel')}~.
    \]

    Finally, consider $\WHILEDO{\guard}{\program'}$, for which both sides of ($\dagger$) are least fixpoints.
    We have $\mop{\WHILEDO{\guard}{\program'}}{\mpost} = \lfp \phiMopNo$ with $\phiMopNo(X) = [\neg \guard] \cdot \mpost \mminkow [\guard] \cdot \mop{\program'}{X}$, and $\wp{\WHILEDO{\guard}{\program'}}{\scal{\weight}{\mpost}} = \lfp \phiWpNo$ for the corresponding classic characteristic function w.r.t.\ the postexpectation $\scal{\weight}{\mpost}$.
    Both characteristic functions are continuous (\Cref{theo:mop-well-cont-monotone}; for \wpsymbol\ see \cite{kaminski2019advanced}), so Kleene's fixpoint theorem (\Cref{theo:Kleene}) expresses both least fixpoints as suprema of their iterations from the respective bottom element.
    An inner induction on $k \in \Nats$ shows that scalarization maps the $k$-th \mopsymbol\ iteration to the $k$-th \wpsymbol\ iteration, i.e., $\scal{\weight}{(\phiMopNo)^k(\zeroME)} = (\phiWpNo)^k(0)$:
    for $k=0$, we have $\scal{\weight}{\zeroME} = 0$ since $\zeroME(\pstate) = \set{0^n}$.
    For $k+1$, we push the scalarization through one application of $\phiMopNo$:
    statewise, if $\guard$ holds, both characteristic functions apply their respective body transformer, and the outer induction hypothesis ($\dagger$) for $\program'$ applies at the intermediate postexpectation $(\phiMopNo)^k(\zeroME)$. Conversely, if $\guard$ does not hold, both return $\mpost$ resp.\ $\scal{\weight}{\mpost}$.
    In summary, we conclude:
    \begin{align*}
        & \scal{\weight}{\mop{\WHILEDO{\guard}{\program'}}{\mpost}} \\
        =\ & \scal{\weight}{\textstyle\bigsqcup_{k}\, (\phiMopNo)^k(\zeroME)} \tag{\Cref{theo:Kleene}} \\
        =\ & \textstyle\bigsqcup_{k}\, \scal{\weight}{(\phiMopNo)^k(\zeroME)} \tag{$\ast$, applied statewise} \\
        =\ & \textstyle\bigsqcup_{k}\, (\phiWpNo)^k(0) \tag{inner induction} \\
        =\ & \wp{\WHILEDO{\guard}{\program'}}{\scal{\weight}{\mpost}} \tag{\Cref{theo:Kleene} for $\phiWpNo$}
    \end{align*}
\end{proof}

\subsection{Proof of \titleCref{theo:wpofdeterm}}
\label{proof:wpofdeterm}

\wpofdeterm*

\begin{proof}
    Since $\program_\weight$ is deterministic, we know that $\wpsymbol$ is linear (\cite[Theorem 4.21]{kaminski2019advanced}).
    We get:
    \begin{align*}
        & \weight \cdot \wpVec{\program_\weight}{\post} \\ 
        =\ & \weight \cdot \vvvechorizon{\wp{\program_\weight}{\post_1}}{\dots}{\wp{\program_\weight}{\post_n}} \\
        =\ & \sum \weight_i \cdot \wp{\program_\weight}{\post_i} \\
        =\ & \wp{\program_\weight}{\sum \weight_i \cdot \post_i} \tag{Linearity}\\
        =\ & \wp{\program_\weight}{\weight \cdot \post} \\
        =\ & \wp{\program}{\weight \cdot \post} \tag{$\program_\weight$ is optimal for $\weight \cdot \post$}\\
    \end{align*}
\end{proof}

\subsection{Proof of \titleCref{theo:determ-wp-as-mop}}
\label{proof:determ-wp-as-mop}

\determwpasmop*

\begin{proof}
    First, note that $\wpVec{\program_\weight}{\post}(\pstate) \in \mop{\program}{\dwcset{\post}}(\pstate)$ due to \Cref{theo:mop-is-pareto}.
    By definition, $\face{\weight}{\mop{\program}{\dwcset{\post}}(\pstate)}$ is equal to
    \[
        \set{x \in \mop{\program}{\dwcset{\post}}(\pstate) \mid
        \weight \cdot x =
        \sup \set{\weight \cdot y \mid y \in \mop{\program}{\dwcset{\post}}(\pstate)}}.
    \]
    Since we have
    \begin{align*}
        & \weight \cdot \wpVec{\program_\weight}{\post}(\pstate) \\
        =\ & \wp{\program}{\weight \cdot \post}(\pstate) \tag{\Cref{theo:wpofdeterm}}\\
        =\ & \sup \set{\weight \cdot y \mid y \in \mop{\program}{\dwcset{\post}}(\pstate)} \tag{\Cref{theo:wp-as-weighted-mop}}
    \end{align*}
    the property holds.
\end{proof}

\subsection{Proof of \titleCref{theo:strasz}}
\label{proof:strasz}

\strasz*

\begin{proof}
    Let $\HDel = \mop{C}{\dwcset{\post}}(\pstate)$.
    By \Cref{proof:wp-as-weighted-mop}, $\HDel$ is compact (and convex since it is an element of $\HoarePowerDom$).

    Let
    \[
        K = \HDel \cap \PosRealsVect
        \qquad\text{and}\qquad
        \widehat K = \setcomp{y \in \Reals^n}{\exists x \in K\colon y \leq x}.
    \]
    Thus, $\widehat K$ is the downward extension of $K$ to $\mathbb{R}^n$.
    In particular, $\widehat K$ is a closed convex subset of $\mathbb{R}^n$.

    We first show that $p$ is an extreme point of $\widehat K$.
    Suppose that
    $
        p = \lambda x + (1-\lambda)y
    $
    for $x,y \in \widehat K$ and $0<\lambda<1$.
    Write $x = x' - r$ and $y = y' - s$ for $x',y' \in K$ and $r,s \in [0,\infty)^n$.
    Then
    \[
        \lambda x' + (1-\lambda)y'
        =
        p + \lambda r + (1-\lambda)s
        \geq p.
    \]
    Since $\lambda x' + (1-\lambda)y' \in K$ and $p$ is Pareto optimal, we must have $\lambda r + (1-\lambda)s = 0.$
    Hence $r=s=0$, and therefore
    $
        p = \lambda x' + (1-\lambda)y'.
    $
    Since $p$ is an extreme point of $K$, this implies $x'=y'=p$, and thus $x=y=p$.
    Hence $p$ is an extreme point of $\widehat K$.

    By Straszewicz's theorem~\cite[Theorem 18.6]{rockafellar1970convex}, the exposed points of a finite-dimensional closed convex set are dense in its extreme points.
    Hence there exists a sequence
    $
        (p_k)_{k \geq 1}
    $
    of exposed points of $\widehat K$ such that
    $
        \lim_{k\to\infty} p_k = p.
    $

    It remains to observe that every $p_k$ is exposed by a nonnegative weight vector.
    Let $p_k$ be exposed by $\weight \in \mathbb{R}^n$.
    If $\weight_i<0$ for some $i$, then, since $\widehat K$ is downward closed in $\mathbb{R}^n$, decreasing the $i$-th coordinate of $p_k$ by any $t>0$ yields another point of $\widehat K$ with strictly larger $\weight$-value, contradicting that $\weight$ exposes $p_k$.
    Thus
    $
        \weight \in [0,\infty)^n.
    $

    Moreover, $p_k \in K$: write $p_k=x-r$ for some $x\in K$ and $r\in[0,\infty)^n$.
    Since $\weight\geq0$,
    \[
        \weight\cdot x
        \geq
        \weight\cdot p_k.
    \]
    As $x\in K\subseteq\widehat K$ and $\weight$ exposes $p_k$ uniquely, this implies $x=p_k$, and hence $p_k\in K$.
    The same weight vector therefore also exposes $p_k$ in $K$.

    Finally, normalizing $\weight$ to $\weight'$ satisfying $\sum_i \weight'_i=1$ gives a weight vector $\weight'\in\Weights$ exposing $p_k$ as required.
\end{proof}

\section{Details on the Case Studies}
\label{app:examples}

\subsection{The Robot}
\label{sec:inv-robot}

As discussed in \Cref{sec:robotics}, the route ``commit to $k$ fast steps'' for each $k = 1, \dots, \lceil \gap/2 \rceil$ is optimal.
We now discuss which points are achieved by these routes.
We write $\qf \coloneqq 0.7$ and $\qs \coloneqq 0.95$ for the two probabilities of not breaking down.

It takes at most $R_k \coloneqq k + (\gap \monus 2k)$ rounds to reach the goal.
The probability of reaching the goal is
\[
    S_k(\gap) \coloneqq \qf^{\,k-1} \cdot \qs^{\,\gap \monus 2k},
\]
capturing that we start by taking $k-1$ fast steps, then $\gap \monus 2k$ slow steps, and the final fast step reaches the goal with probability $1$.

Second, we discuss the expected remaining time budget.
Let $C_k(r)$ be the probability of surviving the first $r$ relevant coins, which is exactly the probability that round $r+1$ is played at all, given by
$
    C_k(r) \coloneqq \qf^{\,\min(r,\, k-1)} \cdot \qs^{\,r \monus (k-1)}.
$
The remaining time budget corresponds to $t$ minus the number of rounds played, which is therefore given by
\[
    L_k(t, \gap) \coloneqq \sum_{r=0}^{\min(t,\, R_k) - 1} C_k(r), \qquad \text{ where (as before) } R_k = k + (\gap \monus 2k).
\]
The truncation at $\min(t, R_k)$ is precisely the budget being truncated at zero.
Consequently, the route taking $k$ fast steps achieves the point
\[
    \kappa_k(t, \gap) \coloneqq \vvechorizon{\, t - L_k(t, \gap)}{\ S_k(\gap)}.
\]

We propose the following invariant, combining all optimal routes:
\[
    \invariant \eeq \iverson{f\neq 0 \lor p\geq g}\cdot \dwcset{\vvechorizon{t}{\iverson{p \geq g}}} \minkow \iverson{f=0 \land p<g} \cdot \bigsqcup_{k=1}^{\lceil (g \monus p)/2 \rceil} \dwcset{\,\kappa_k(t, g \monus p)\,}.
\]

As mentioned above, we show that this invariant is precisely the least fixpoint of the \mopsymbol\ characteristic function, that is, $\lfp \phiMopNo = \invariant$.
By \Cref{theo:induction_lower}.\ref{theo:induction_lower:dct} this reduces to showing that $\phiMopNo(\invariant) = \invariant$.
Since this follows by a straightforward application of the rules in \Cref{table:mop}, we omit the calculations.
For the entire program including initializations, we can simplify this to
\[
    \mop{\program_{\text{robot}}}{\dwcset{\vvechorizon{t}{\iverson{p \geq g}}}}(\pstate)
    \;=\;
    \bigsqcup_{k=1}^{\lceil g/2 \rceil} \dwcset{\, \vvechorizon{\, g - L_k(g)}{\ 0.7^{\,k-1} \cdot 0.95^{\,g \monus 2k} }\, };
\]
with 
\[
    L_k(g) \coloneqq \frac{1-0.7^k}{0.3} + 19\cdot 0.7^{\,k-1} \left(1-0.95^{\,g\monus 2k}\right), 
\]
with $L_k(g)$ abbreviating $L_k(g, g)$ and depending only on the initial value of $g$.

\begin{figure}[t]

\scalebox{0.7}{
\begin{tikzpicture}[
    x={4cm},
    y={4cm},
    >=stealth
]

\foreach \x in {1,2,3}
    \draw (\x,0.01) -- (\x,-0.01) node[below] {\x};

\foreach \y/\label in {0.5/{0.5},1/{1}}
    \draw (0.01,\y) -- (-0.01,\y) node[left] {\label};

\coordinate (V1) at (1.290125,0.857375);
\coordinate (V2) at (2.635,0.665);
\coordinate (V3) at (2.81,0.49);

\filldraw[
    fill=mygreen!30,
    draw=mygreen,
    line join=round
]
    (0,0)
    -- (0,0.857375)
    -- (V1)
    -- (V2)
    -- (V3)
    -- (2.81,0)
    -- cycle;

\draw[->] (0,0) -- (3.2,0) node[below] {$t$};
\draw[->] (0,0) -- (0,1.1) node[left] {$\iverson{p\geq g}$};

\draw[->, ultra thick, myblue]
  (V1) -- ($(V1)+0.03*({1},{9})$)
  node[right] {$\weight_1$};

\draw[myblue, thick]
  ($(V1)-0.03*({-9},{1})$) -- 
  ($(V1)+0.03*({-9},{1})$);

\draw[->, ultra thick, myblue]
  (V2) -- ($(V2)+0.12*({1},{2})$)
  node[right] {$\weight_2$};

\draw[myblue, thick]
  ($(V2)-0.12*({-2},{1})$) -- 
  ($(V2)+0.12*({-2},{1})$);

\coordinate (P) at ($(V1)!0.6!(V2)$);
\fill[myblue] (P) circle (2.2pt);
\node[myblue,right] at ($(P)+(0,-0.06)$) {$\target$};

\fill[mygreen!80!black] (V1) circle (2.2pt);
\node[mygreen!80!black] at ($(V1)+(-0,-0.1)$) {$\kappa_1$};

\fill[mygreen!80!black] (V2) circle (2.2pt);
\node[mygreen!80!black] at ($(V2)+(-0.08,-0.07)$) {$\kappa_2$};

\fill[mygreen!80!black] (V3) circle (2.2pt);
\node[mygreen!80!black] at ($(V3)+(-0.1,-0.03)$) {$\kappa_3$};

\end{tikzpicture}}
\caption{Illustration of the weight vectors $\weight_1 = \vvechorizon{\frac{1}{10}}{\frac{9}{10}}$ exposing $\kappa_1$ and
$\weight_2 = \vvechorizon{\frac{1}{3}}{\frac{2}{3}}$ exposing $\kappa_2$, as required to achieve $\target$.
Note that apart from the scaling of the axes, this plot is equivalent to \Cref{fig:robot-pareto-mop-2}.}
\label{fig:robot-pareto-weights}

\Description{
    A two-dimensional plot illustrates two weight vectors that expose vertices of the Pareto front of the robot program.
    The horizontal axis represents the remaining time $t$, while the vertical axis represents the probability of reaching the goal $[p \geq g]$.
    The achievable region is shaded light green and is bounded above by a piecewise-linear Pareto front.
    Three dark green points labeled $\kappa_1$, $\kappa_2$, and $\kappa_3$ mark vertices of the Pareto front.
    A blue point labeled $\tau$ lies on the descending segment connecting $\kappa_1$ and $\kappa_2$, representing the target for strategy synthesis.

    Two blue arrows labeled $w_1$ and $w_2$ indicate the weight vectors used to expose the vertices $\kappa_1$ and $\kappa_2$, respectively.
    The first weight vector $w_1=(\frac{1}{10},\frac{9}{10})$ points almost vertically upward from $\kappa_1$, emphasizing the probability of reaching the goal.
    The second weight vector $w_2=(\frac{1}{3},\frac{2}{3})$ points diagonally upward and to the right from $\kappa_2$, assigning relatively more weight to the remaining time.
    Thin blue supporting lines pass through the corresponding vertices and are perpendicular to their respective weight vectors.
    Each supporting line touches the Pareto front at its associated vertex, illustrating how the weight vectors expose these points.

    The illustration demonstrates that the target $\tau$ lies between two exposed vertices and can be achieved using the strategies associated with $\kappa_1$ and $\kappa_2$.
    Apart from the scaling of the axes, the achievable region corresponds to the Pareto front shown previously for $g=5$.
}

\end{figure}

For the initial state $\pstate_5 = (g \mapsto 5)$ as fixed in the main part, this amounts to the multiobjective preexpectation displayed in \Cref{fig:robot-pareto-weights}.
The figure also shows the weight vectors $\weight_1$ and $\weight_2$ required for the synthesis of a determinization achieving $\target$.

\subsection{The Casino}
\label{sec:inv-casino}
\begin{figure}[t]
\begin{lstlisting}[mathescape, basicstyle=\small]
$\annotate{R\subst{i}{i+1}}$
$\ASSIGN{i}{i+1} \fatsemi$
$\annotate{R \coloneqq \dwcset{\vvechorizon{i}{n}} \mecup \frac{2}{3} \cdot ([c=0] \cdot \dwcset{\vvechorizon{i}{n+1}} \mminkow [c=1] \cdot \dots}$
$\GCFSTOPEN{~\boldsymbol{\blue{\trueBold}}}$
    $\annotate{\dwcset{\vvechorizon{i}{n}}}$
    $\gray{\text{\% stop and cash out}}$
    $\ASSIGN{c}{0} \fatsemi$
    $\annotate{[c=0] \cdot \dwcset{\vvechorizon{i}{n}} \mminkow [c=1] \cdot {\displaystyle \bigmecup_{m \in \Nats}} \dwcset{\vvechorizon{i + 3-\frac{2^{m+1}}{3^m}}{\left(\frac{2}{3}\right)^m \cdot (n+m)}}}$
$\GCSECCLOSE \GCSECOPEN{\boldsymbol{\blue{\trueBold}}}{}$
    $\annotate{\frac{2}{3} \cdot \left([c=0] \cdot \dwcset{\vvechorizon{i}{n+1}} \mminkow  [c=1] \cdot {\displaystyle \bigmecup_{m \in \Nats}} \dwcset{\vvechorizon{i + 3-\frac{2^{m+1}}{3^m}}{\left(\frac{2}{3}\right)^m \cdot (n+1+m)}}\right)}$
    $\annotateNo{\mminkow \frac{1}{3} \cdot \dwcset{\vvechorizon{i}{0}}}$
    $\gray{\text{\% gamble}}$
    $\PCHOICE{\ASSIGN{n}{n+1}}{\frac{2}{3}}{\COMPOSE{\ASSIGN{n}{0}}{\ASSIGN{c}{0}}}$
    $\annotate{[c=0] \cdot \dwcset{\vvechorizon{i}{n}} \mminkow  [c=1] \cdot {\displaystyle \bigmecup_{m \in \Nats}} \dwcset{\vvechorizon{i + 3-\frac{2^{m+1}}{3^m}}{\left(\frac{2}{3}\right)^m \cdot (n+m)}}}$
$\GCSECCLOSE$
$\annotate{[c=0] \cdot \dwcset{\vvechorizon{i}{n}} \mminkow  [c=1] \cdot {\displaystyle \bigmecup_{m \in \Nats}} \dwcset{\vvechorizon{i + 3-\frac{2^{m+1}}{3^m}}{\left(\frac{2}{3}\right)^m \cdot (n+m)}}}$
\end{lstlisting}%
\caption{Annotations for computing $\mopsymbol$ of the loop body of $\program_{\text{casino}}$ and invariant $I$.}
\label{fig:casino_invariant}

\Description{
    A program listing shows annotations for computing the multiobjective preexpectation of the loop body of the casino program $C_{\mathrm{casino}}$ and verifying a loop invariant $I$.
    The program first increments the counter $i$ by one and then makes a nondeterministic choice between stopping and gambling.
    Mathematical annotations are highlighted in orange and appear before and after the program statements.

    At the top of the listing, the annotation $R[i/i+1]$ represents the substitution of $i+1$ for $i$ in the preexpectation $R$.
    The next annotation defines $R$ using the multiobjective postexpectation $\operatorname{dwc}\{(i,n)\}$, probabilistic weighting, and a distinction between the cases $c=0$ and $c=1$.

    The first branch represents stopping and cashing out.
    Before the assignment $c:=0$, the annotation is $\operatorname{dwc}\{(i,n)\}$.
    After the assignment, the annotation distinguishes between the terminated case $c=0$ and the continuing case $c=1$.
    The continuing case contains a union over $m\in\mathbb{N}$ of multiobjective expectations involving the expressions $i+3-\frac{2^{m+1}}{3^m}$ and $(\frac{2}{3})^m(n+m)$.

    The second branch represents gambling.
    With probability $2/3$, the variable $n$ is incremented by one, while with probability $1/3$, the variables $n$ and $c$ are reset to zero.
    The annotation before this probabilistic choice combines the corresponding preexpectations using probability-weighted operations.
    The annotation after the choice has the same form as the annotation after the stopping branch.

    At the bottom of the listing, the resulting annotation combines the cases $c=0$ and $c=1$ using indicator expressions and a union over all $m\in\mathbb{N}$.
    The annotations illustrate the backward computation of multiobjective preexpectations through assignments, probabilistic choices, and nondeterministic branching.
    Together, they provide the symbolic expressions used to establish the loop invariant and compute the multiobjective preexpectation of the casino program.
}

\end{figure}
\Cref{sec:casino} proposed the following invariant to compute $\mop{\program_{\text{casino}}}{\dwcset{\vvechorizon{i}{n}}}$:
\[
    \invariant = [c=0] \cdot \dwcset{\vvechorizon{i}{n}} \minkow  [c=1] \cdot \bigmecup_{m \in \Nats} \dwcset{\vvechorizon{i + 3-\frac{2^{m+1}}{3^m}}{\left(\frac{2}{3}\right)^m \cdot (n+m)}}.
\]
To prove this invariant correct, we use \Cref{theo:park-upper-mop} to prove that this is an upper bound, i.e., we need to show that $\phiMopNo(\invariant) \leq \invariant$.
The characteristic function is defined as
\[
    \phiMopNo(\invariant) \, \eeq\, \iverson{c=0} \cdot \dwcset{\vvechorizon{i}{n}} \mminkow  \iverson{c=1} \cdot \mop{\program'}{\invariant}~,
\]
with $\program'$ notating the loop body of $\program_{\text{casino}}$.
We compute $\mop{\program'}{\invariant}$ in \Cref{fig:casino_invariant}, where $R\subst{i}{i+1}$ denotes the final result.
In the characteristic function, this term is guarded by $\iverson{c=1}$.
We get that $\iverson{c=1} \cdot R\subst{i}{i+1}$ is equivalent to
\begin{align*}
   \iverson{c=1} \cdot &\ \left(\dwcset{\vvechorizon{i+1}{n}} \right. \\
   &\left. \mecup\, \left(\frac{2}{3} \cdot \left(\bigmecup_{m \in \Nats} \dwcset{\vvechorizon{i +1  + 3-\frac{2^{m+1}}{3^m}}{\left(\frac{2}{3}\right)^m \cdot (n+1+m)}}\right) \minkow  \frac{1}{3} \cdot \dwcset{\vvechorizon{i+1}{0}}\right)\right).
\end{align*}
Focusing on the second part, we apply $\frac{2}{3}$ and add $\dwcset{\vvechorizon{i+1}{0}}$ pointwise, getting us
\[
    \bigmecup_{m \in \Nats} \dwcset{\vvechorizon{\frac{2}{3} \cdot \left(i +1  + 3-\frac{2^{m+1}}{3^m}\right)+ \frac{1}{3}\cdot (i+1)}{\frac{2}{3} \cdot \left(\left(\frac{2}{3}\right)^m \cdot (n+1+m)\right)}}.
\]
The first component reduces to
\begin{align*}
    & \frac{2}{3} \cdot \left(i +1  + 3-\frac{2^{m+1}}{3^m}\right) + \frac{1}{3}\cdot (i+1) \\
    =\ & i + 1 + \frac{2}{3} \cdot 3 - \frac{2}{3} \cdot\frac{2^{m+1}}{3^m} \\
    =\ &  i + 3-\frac{2^{m+2}}{3^{m+1}}.
\end{align*}
Putting this together with the second component, we get
\begin{align*}
   \dwcset{\vvechorizon{i+1}{n}}
   \, \mecup\, \bigmecup_{m \in \Nats} \dwcset{\vvechorizon{i + 3-\frac{2^{m+2}}{3^{m+1}}}{\left(\frac{2}{3}\right)^{m+1} \cdot (n+m+1)}}
\end{align*}
If we shift the index in the least upper bound by one, we get
\begin{align*}
   \dwcset{\vvechorizon{i+1}{n}}
   \, \mecup\, \bigmecup_{m \geq 1} \dwcset{\vvechorizon{i + 3-\frac{2^{m+1}}{3^{m}}}{\left(\frac{2}{3}\right)^{m} \cdot (n+m)}},
\end{align*}
and since $\dwcset{\vvechorizon{i+1}{n}}$ is precisely the latter expression for $m=0$, we get that this is equivalent to
\begin{align*}
    \bigmecup_{m \in \Nats} \dwcset{\vvechorizon{i + 3-\frac{2^{m+1}}{3^{m}}}{\left(\frac{2}{3}\right)^{m} \cdot (n+m)}}.
\end{align*}
This lets us conclude that
\[
    \phiMopNo(\invariant) \, \eeq\, [c=0] \cdot \dwcset{\vvechorizon{i}{n}} \minkow  \iverson{c=1} \cdot \bigmecup_{m \in \Nats} \dwcset{\vvechorizon{i + 3-\frac{2^{m+1}}{3^{m}}}{\left(\frac{2}{3}\right)^{m} \cdot (n+m)}}  \, \eeq\, \invariant,
\]
so $\invariant$ is a fixpoint of $\phiMopNo$.
With \Cref{theo:park-upper-mop}, we get that $\lfp \phiMopNo \leq \invariant$.

The $\omega$-invariant for the lower bound is given by
\[
    \invariant_k = [c=0] \cdot \dwcset{\vvechorizon{i}{n}} \minkow [c=1] \cdot \bigmecup_{0 \leq m \leq k} \dwcset{\vvechorizon{i + 3-\frac{2^{m+1}}{3^m}}{\left(\frac{2}{3}\right)^m \cdot (n+m)}}.
\]
We omit the calculations of 
\[
    \invariant_0 \ssqsubseteq \phiMopNo(\zeroME)
    \qquad \text{and} \qquad
    \invariant_{k+1} \ssqsubseteq \phiMopNo(\invariant_k) \quad \text{for all } k \in \Nats
\]
here as they closely resemble those for the upper bound above.

\end{document}